\documentclass[a4paper,USenglish,cleveref,autoref,]{lipics-v2021}

\usepackage{commands-tam}

\usepackage[utf8]{inputenc}
\usepackage{subcaption}
\usepackage{siunitx}
\usepackage{blkarray}
\usepackage{cite}
\usepackage{amsmath}
\PassOptionsToPackage{french,english}{babel}
\usepackage[T1]{fontenc}
\usepackage[
    left = \flqq{},%
    right = \frqq{},%
    leftsub = \flq{},%
    rightsub = \frq{} %
]{dirtytalk}

\usepackage[textsize=scriptsize]{todonotes}

\newif\ifhighlight
\highlighttrue

\newif\ifshowtext
\showtextfalse

\title{Simulation of the abstract Tile Assembly Model Using Crisscross Slats}

\author{Phillip Drake}{University of Arkansas, USA}{padrake@uark.edu}{}{This author's work was supported in part by NSF grant 2329908}

\author{Daniel Hader}{University of Arkansas, USA}{dhader@uark.edu}{}{This author's work was supported in part by NSF grant 2329908}

\author{Matthew J. Patitz}{University of Arkansas, USA}{patitz@uark.edu}{https://orcid.org/0000-0001-9287-4028}{This author's work was supported in part by NSF grant 2329908}

\authorrunning{P.~Drake, D. Hader, and M.\,J. Patitz}

\Copyright{Phillip Drake, Daniel Hader, and Matthew J. Patitz}



\category{} 

\relatedversion{} 

\supplement{}

\usepackage[utf8]{inputenc}
\usepackage{subcaption}
\usepackage{tabularx}
\usepackage{makecell}

\begin{document}

\maketitle

\begin{abstract}
The abstract Tile Assembly Model provides an excellent foundation for the mathematical study of DNA-tile-based self-assembling systems, especially those wherein logic is embedded within the designs of the tiles so that they follow prescribed algorithms. While the theoretical power of such algorithmic self-assembling systems is immense, being computationally universal and capable of building complex shapes using information-theoretically optimal numbers of tiles, physical DNA-based implementations of these systems still encounter formidable error rates and undesired nucleation that stifle this theoretical potential. Slat-based self-assembly is a recent development wherein DNA forms long slats that combine together in 2 layers, rather than the aTAM's square tiles in a plane. In this approach, the length of the slats is key; while tiles only generally bind to 2 neighboring tiles at a time, slats may bind to dozens of other slats. This increased coordination between slats means that several mismatched slats must coincidentally meet in just the right way for errors to persist, unlike tiles where only a few are required. Consequently, while still a novel technology, large slat-based DNA constructions have been implemented in the lab with great success and incredible resilience to many of the problems that plague tile-based constructions. These improved error characteristics come at a cost however, as slat-based systems are often more difficult to design and simulate than analogous tile-based ones. Moreover, it has not been clear whether slats, with their larger sizes and different geometries, have the same theoretical capabilities as tiles. In this paper, we show that slats are in fact capable of doing anything that tiles can, at least at scale. We provide constructions demonstrating that any aTAM system may be converted to and simulated by an effectively equivalent system of slats. Furthermore, we show that these simulating slat systems can be made more efficiently, using shorter slats and a smaller scale factor, if the simulated tile system avoids certain growth patterns that are typically uncommon anyway. Specifically, we consider 5 classes of aTAM systems with increasing complexity, from zig-zag systems which grow in a rigid pattern to the full class of all aTAM systems, and show how they may be converted to equivalent slat systems. We show that the simplest class, zig-zag systems, may be simulated by slats at only a $2c \times 2c$ scale, where $c$ is the freely chosen coordination number (a.k.a. cooperativity) of the slats. We further show that the full class of aTAM systems can be simulated at only a $5c \times 5c$ scale, and the other intermediate classes may be simulated using scales between these. Together, these results prove that slats have the full theoretical power of aTAM tiles while also providing constructions that are compact enough to potentially provide designs for DNA-based implementations of slat systems that are both capable of powerful algorithmic self-assembly and possessing of the strong error resilience of slats. 

\end{abstract}

\maketitle

\clearpage
\pagenumbering{arabic}

\section{Introduction}\label{sec:intro}

In self-assembly, simple, disorganized components combine to form structures more complex than themselves, driven primarily by local interactions and environmental conditions. From the crystallization of water molecules into the intricate 6-fold symmetry of snowflakes, to the clustering of space dust and gasses into robust solar systems with mechanisms to mitigate the deleterious effects of debris and radiation, self-assembly processes occur at all scales of nature. Such processes are central to many fields of science and engineering, including the relatively young field of DNA-nanotechnology. Here, synthetic strands of DNA are used, not as a means to store genetic information, but rather as building blocks for nano-scale structures, far too small to assemble using conventional human building techniques. Taking advantage of the base-pairing dynamics of DNA, synthetic strands can be mixed in solution under carefully tuned conditions so that they naturally combine to form incredibly precise shapes \cite{RothOrigami, Douglas2009, ke2012three,OrigamiTiles,MonaLisa}, and even follow embedded algorithms \cite{evans2014crystals, SeemanDNARobots2010, WinfreeDNARobots2010, PadillaSignals, rothemund2004algorithmic, Zhang2017, woods2019diverse}. On an atomic scale, DNA is far too complex to completely and efficiently model so heuristics and simplifications are often used when designing DNA-based self-assembling systems. Tile-assembly models are one such simplification that have seen great success in facilitating the design of such systems. In tile-assembly, it is assumed that DNA strands are designed so that they tend to combine into small rigid units called \emph{tiles} resembling squares (or sometimes other shapes). These rigid units are augmented with extra lengths of single-stranded DNA that dangle from their sides (often called ``glues'' or ``handles'') to enable individual units to selectively combine with one another. The utility of tile-assembly comes from its simplicity and relationship with existing models in mathematics and computer science. While individual DNA strands are difficult to model, when designed to behave like tiles their self-assembly is relatively well understood and many important dynamics can be easily captured by simple mathematical rules. 

Theoretically, tile-assembly models have been extensively studied, and models such as the abstract Tile-Assembly Model (aTAM) have been shown to be algorithmically universal in that they are capable of simulating arbitrary Turing machines \cite{Winf98,jCCSA,jSADS,SolWin07}. Practically, tile-assembly models have seen significant use as design tools for complicated DNA-based nano-structures \cite{evans2014crystals, woods2019diverse}. There are however a few key difficulties that arise when attempting to realize tile-based DNA constructions. One primary difficulty is nucleation.
To ensure that the self-assembly process occurs as expected, it is generally important that assembly begins from a chosen starting \emph{seed} structure; however it can be extremely difficult to guarantee that growth does not begin spuriously by the improbable combination of a small number of tiles away from the seed. Using conventional approaches to DNA-based tile-assembly, spurious nucleation is a major hurdle to building large structures. Another difficulty comes from so-called \emph{growth errors}. While tiles may be designed so that the correct tile attachments are thermodynamically preferred, it is unlikely that erroneous attachments can be prevented altogether. Typically such errors are short lived due to the entropic penalties they incur, but if enough occur in quick succession and in just the right way, it's possible for the errors to become locked-in. Techniques such as proofreading (in various forms) \cite{SolCookWin08,SolWin05,WinBek03,ReiSahYin04,ChenHealing,CheSchGoeWin07,ChenGoel04} have been developed to mitigate these problems, but they still act as a major obstacle to larger scale DNA-based tile-assembly.

One recent development, however, has seemingly overcome both of these problems through the use of \emph{slat}-shaped DNA units  \cite{ShihNucleation,Shih-OrigamiSlats,SlatsEnzyme,SlatsBranching} rather than square tile-shaped ones. Unlike tiles which attach to at most 4 neighbors and combine in a plane, slats are long and designed to attach in multiple layers so that a single slat may span across and attach to dozens of others. For a square tile where 2 of its sides attach to an existing assembly, erroneous attachments often occur when just one of the sides correctly binds to the assembly. One of two attachments is still relatively strong and an erroneous tile may remain attached for a substantial amount of time. Even worse, it is only really necessary for 4 or so individual tiles to coincidentally bump into one another for spurious nucleation to occur. While unlikely, this is almost guaranteed to happen frequently in a mass-action system on the scale of moles. Because each slat needs to attach to 8 or even 16 others to achieve a stable bond, erroneous attachments are generally much shorter lived and less likely to lock-in, and the likelihood of spurious nucleation drops precipitously (effectively to zero \cite{ShihNucleation}).

In the lab, slats are generally implemented either using DNA-origami or as individual strands of DNA. In the DNA-origami motif, slats are often 6-helix bundles (a very common origami construction) with single-stranded DNA ``handles'' extending from one side. While both techniques are still novel, origami-based slats have been demonstrated to be incredibly robust to spurious nucleation and computer simulations have indicated that slats naturally exhibit error correcting behavior since individual erroneous attachments have little effect on correct growth later in the assembly process \cite{SlatAcceleration}. Theoretical models of slat-assembly have been introduced, naturally expanding on tile-assembly models, but little is currently known about their dynamics. In this paper, we consider the abstract Slat Assembly Model (aSAM) introduced in \cite{SlatAcceleration}, and investigate its relationship to the aTAM. Specifically we consider the extent to which aTAM tiles may be simulated by aSAM slats. In this context, simulation refers to ``intrinsic simulation'' a notion borrowed from the study of cellular automata \cite{Ollinger-CSP08,ollinger2003intrinsic} and which has been used extensively to compare tile-assembly models and develop a rich complexity theory for them \cite{IUSA,temp1notIU,DirectedNotIU,DDDIU,NeedForSeed,DDDxModelSimICALP,jDuples,WoodsIUSurvey}.
Unlike typical simulations between models of computation, where the dynamics of one model are captured symbolically by the dynamics of another, intrinsic simulation is inherently geometric. For a system $S$, be it of tiles or slats, to simulate another system $S'$ requires that $S$ ``looks like'' $S'$ when zoomed-out and furthermore that any order of attachments in $S'$ may be replicated by attachments in $S$.

\paragraph*{Our Results}

In this paper, we show that all systems in the aTAM may be intrinsically simulated by aSAM systems. Moreover, we show that if one is willing to forgo some less useful dynamics of the aTAM, then this simulation may be done quite efficiently, both in the scale factor required for the simulation and in the complexity of the necessary slats. Specifically, we consider 5 different classes of aTAM systems of increasing complexity.
The first class, \emph{zig-zag} systems, are still fully capable of Turing universal computation, but are restricted to growing solely in a back-and-forth pattern.
The second class of systems, called \emph{standard} systems, represents a simplified set of aTAM dynamics common to most theoretical constructions. These standard systems make simplifying assumptions such as requiring that no tiles mismatch with their neighbors, requiring each tile to attach with no more strength than necessary, and requiring that only 1 terminal assembly is possible. Despite this, all but the most convoluted theoretical aTAM constructions may generally be defined as standard systems. In standard systems, it is assumed that no tiles attach using ``across-the-gap'' cooperation, where a tile binds to 2 tiles of an existing assembly that are not adjacent to one another. Such attachments are generally more difficult to simulate and rarely, if ever, appear in physically implemented tile-based systems. Still, the 3rd class of aTAM systems considered in this paper are standard systems augmented with the ability to perform across-the-gap cooperation. In the 4th class we allow tiles to mismatch with their neighbors so long as they attach with sufficient strength and the system is directed (i.e. makes a unique assembly), and our 5th consists of all aTAM systems.

Table~\ref{tab:results-table} details our results with respect to a parameter $c$ of our aSAM systems called the ``cooperativity'' or sometimes ``coordination number'' of the slats. This number may be freely chosen, independent of our results, and its value effectively describes how many functionally redundant attachment domains appear along the length of each slat. In practice, increasing this number will result in a slat system more robust to spontaneous nucleation and growth errors at the cost of requiring longer and more numerous slats. In our results, we show that we can simulate arbitrary aTAM systems using slats of length no greater than $5c$ at a scale factor of $5c$, that is each aTAM tile is represented by a $5c \times 5c$ block of slats. We further show that this may be optimized to slats of maximum length $3c$ with a $2c$ scale factor for zig-zag systems. For classes in between the zig-zag systems and full aTAM, we show that the scale factor and maximum slat length grow accordingly. We present these results in increasing complexity (and note that software for converting classes of aTAM systems to slats, simulate their self-assembly, and visualize the results can be found online at \url{http://self-assembly.net/wiki/index.php/Abstract_Slat_Assembly_Model_(aSAM)}). It should be noted that a difference between a scale factor of $2c$ and $5c$ is negligible in a purely theoretical context. The real motivation for exploring and simulating different families of aTAM systems is to try and find ``practical'' transformations from the logic of tile-based assembly into error-robust slats which may be implementable. These constructions therefore are not only intended to show that aTAM dynamics may be simulated by slat dynamics, but also serve to illustrate the difficulties that arise when one tries to do so and how these difficulties may affect the slats necessary for simulation. Furthermore, while slat-based self-assembly is still in its infancy, we are optimistic that these constructions, while presented here purely theoretically, may provide designs that help to physically realize them in the not too distant future.

\begin{table}
    \centering
    \begin{tabular}{|c|c|c|c|c|}
        \hline
         aTAM class & Result& \makecell{Macrotile\\size} & \makecell{Greatest\\num slats} & \makecell{Greatest\\slat length}\\
         \hline\hline
         \makecell{Zig-zag}&Thm.\ref{thm:zig-zag}& $2c \times 2c$  & $4c$ & $3c$ \\
         \hline
         Standard&Thm.\ref{thm:standard}&$3c \times 3c$ & $8c$ & $3c$ \\
         \hline
         \makecell{Standard plus across-the-gap}&Thm.\ref{thm:standard+across-the-gap}&$3c \times 3c$  & $8c$ & $4c$ \\
         \hline
         \makecell{Directed temperature-2}&Thm.\ref{thm:deterministic+mismatches}&$4c \times 4c$ & $10c$ & $4c$ \\
         \hline
         \makecell{Nondeterministic (full aTAM)}&Thm.\ref{thm:full-aTAM}& $5c \times 5c$ & $13c$ & $5c$ \\
         \hline
    \end{tabular}
    \caption{An overview of our results. For each class of aTAM systems, corresponding to each of our theorems, we list: the simulation scale factor (the size of our macrotiles), the greatest number of slats that appear in any macrotile, and the largest slat length used during the simulation. Defined later, macrotiles represent blocks of slats which simulate individual tiles and $c$ is the cooperativity.
    }
    \label{tab:results-table}
    \vspace{-20pt}
\end{table}

\section{Preliminary Definitions and Models}

In this section, we provide definitions and overviews of the models and concepts used throughout the paper.

\subsection{The abstract Tile Assembly Model}\label{sec:aTAM}

Our conversions begin from systems within the abstract Tile-Assembly Model\cite{Winf98} (aTAM). These definitions are borrowed from \cite{DDDIU} and we note that \cite{RotWin00} and \cite{jSSADST} are good introductions to the model for unfamiliar readers. 

Let $\mathbb{N}$ be the set of non-negative integers, and for $n \in \mathbb{N}$, let $[n] = \{0, 1, ..., n-2, n-1\}$.
Let $\Sigma$ to be some alphabet with $\Sigma^*$ its finite strings. A \emph{glue} $g\in\Sigma^*\times\mathbb{N}$ consists of a finite string \emph{label} and non-negative integer \emph{strength}. There is a single glue of strength $0$, referred to as the \emph{null} glue. A \emph{tile type} is a tuple $t\in(\Sigma^*\times\mathbb{N})^{4}$, thought of as a unit square with a glue on each side. A \emph{tile set} is a finite set of tile types. We always assume a finite set of tile types, but allow an infinite number of copies of each tile type to occupy locations in the $\mathbb{Z}^2$ lattice, each called a \emph{tile}.
Given a tile set $T$, a \emph{configuration} is an arrangement (possibly empty) of tiles in the lattice $\mathbb{Z}^2$, i.e.\ a partial function $\alpha:\mathbb{Z}^2\dashrightarrow T$. Two adjacent tiles in a configuration \emph{interact}, or are \emph{bound} or \emph{attached}, if the glues on their abutting sides are equal (in both label and strength) and have positive strength. Each configuration $\alpha$ induces a \emph{binding graph} $B_\alpha$ whose vertices are those points occupied by tiles, with an edge of weight $s$ between two vertices if the corresponding tiles interact with strength $s$. An \emph{assembly} is a configuration whose domain (as a graph) is connected and non-empty. The \emph{shape} $S_\alpha \subseteq \mathbb{Z}^2$ of assembly $\alpha$ is the domain of $\alpha$. For some $\tau\in\mathbb{Z}^+$, an assembly $\alpha$ is \emph{$\tau$-stable} if every cut of $B_\alpha$ has weight at least $\tau$, i.e.\ a $\tau$-stable assembly cannot be split into two pieces without separating bound tiles whose shared glues have cumulative strength $\tau$. 

A \emph{tile-assembly system} (TAS) is a triple $\calT=(T,\sigma,\tau)$, where $T$ is a tile set, $\sigma$ is a finite $\tau$-stable assembly called the \emph{seed assembly}, and $\tau\in\mathbb{Z}^+$ is called the \emph{binding threshold} (a.k.a. \emph{temperature}).
Given a TAS $\calT=(T,\sigma,\tau)$ and two $\tau$-stable assemblies $\alpha$ and $\beta$, we say that $\alpha$ \emph{$\calT$-produces} $\beta$ \emph{in one step} (written $\alpha \to^{\calT}_1 \beta$) if $\alpha \sqsubseteq \beta$ and $|S_\beta \setminus S_\alpha| = 1$.
That is, $\alpha \to^{\calT}_1 \beta$ if $\beta$ differs from $\alpha$ by the addition of a single tile.
The \emph{$\calT$-frontier} is the set $\partial^{\calT}\alpha = \bigcup_{\alpha \to^{\calT}_1 \beta} S_\beta \setminus S_\alpha$ of locations in which a tile could $\tau$-stably attach to $\alpha$.
We use $\mathcal{A}^T$ to denote the set of all assemblies of tiles in tile set $T$. Given a TAS $\calT=(T, \sigma, \tau)$, a sequence of $k\in\mathbb{Z}^+ \cup \{\infty\}$ assemblies $\alpha_0, \alpha_1, \ldots$ over $\mathcal{A}^T$ is called a \emph{$\calT$-assembly sequence} if, for all $1\le i < k$, $\alpha_{i-1} \to^{\calT}_1 \alpha_i$. The \emph{result} of an assembly sequence is the unique limiting assembly of the sequence. For finite assembly sequences, this is the final assembly; whereas for infinite assembly sequences, this is the assembly consisting of all tiles from any assembly in the sequence. We say that \emph{$\alpha$ $\calT$-produces $\beta$} (denoted $\alpha\to^{\calT} \beta$) if there is a $\calT$-assembly sequence starting with $\alpha$ whose result is $\beta$. We say $\alpha$ is \emph{$\calT$-producible} if $\sigma\to^{\calT}\alpha$ and write $\prodasm{\calT}$ to denote the set of $\calT$-producible assemblies. We say $\alpha$ is \emph{$\calT$-terminal} if $\alpha$ is $\tau$-stable and there exists no assembly that is $\calT$-producible from $\alpha$. We denote the set of $\calT$-producible and $\calT$-terminal assemblies by $\termasm{\calT}$. If $|\termasm{\calT}| = 1$, i.e., there is exactly one terminal assembly, we say that $\calT$ is \emph{directed}. 

\subsection{Classes of aTAM systems}\label{sec:aTAM-classes}

In \cite{Winf98}, the aTAM was shown to be computationally universal when $\tau = 2$, but this is not the case when $\tau = 1$ \cite{TempOneNotIU}. Furthermore, any aTAM system with $\tau = 1$ can trivially be transformed into a $\tau = 2$ system by changing all of its strength-1 glues to be strength-2; and in any aTAM system where $\tau = 2$, any glue whose strength is greater than $2$ may trivially be replaced by a glue of strength $= 2$ without changing any behaviors of the system. Additionally, any aTAM system with $\tau > 2$ may be simulated by a system with $\tau = 2$ \cite{IUSA}. Therefore, all results in this paper will only discuss aTAM systems with $\tau = 2$ and it will be assumed that, other than the \emph{null} glue of strength 0, all glues in aTAM systems are of strength 1 or 2. When a tile initially binds to an assembly in a $\tau = 2$ system, it must immediately bind with at least (1) a single strength-2 glue, or (2) two strength-1 glues (we call this a \emph{cooperative} attachment). This is because the sum of the bond strengths must be $\ge 2$. If, when a tile initially binds to an assembly, it does so by immediately forming bonds of strength $> 2$, we call this \emph{overbinding}. When a tile $t$ attaches into a location $(x,y)$ by cooperatively binding to tiles on opposite sides of each other (i.e in locations $(x-1,y)$ and $(x+1,y)$, or $(x,y-1)$ and $(x,y+1)$), we say that $t$ has attached \emph{across-the-gap}. When two tiles in adjacent locations do not share a matching glue on their abutting faces, we say that their glues are \emph{mismatched} and note that this is only possible when their other glues contribute a cumulative attachment strength of at least $2$.

Let the symbols `$\vee$', `<', `$\wedge$', `>'  be called the \emph{input markings} for the directions $N, E, S, W$, respectively and the \emph{output markings} for the directions $S, W, N, E$, respectively. (Visually, if the input markings were placed on the corresponding sides of a tile, they would be ``pointing into'' the tile, and vice-versa for output markings) We say that a tile is \emph{IO-marked} if a subset of its glues whose sum is $\ge 2$ have input markings as prefixes to their labels, and all other non-null glues have output markings as prefixes to their labels. Since each direction has a unique input marking, and the input marking of each direction is the same as the output marking of the opposite direction, it is clear that for glues to match and form a bond, an input-marked glue on any given tile side $d$ can only bind with an output-marked glue on the opposite side of another tile, and vice versa. Note that it is possible to convert any aTAM system to an equivalent IO-marked aTAM system. Furthermore this conversion can be done so that each IO-marked tile has a minimal set of input glues, that is all input glues on a tile are necessary for the attachment (if a tile could attach in multiple ways using different combinations of input glues it is split into multiple tiles representing the same tile, each with a different subset of input glues). This ensures that assemblies made of IO-marked tiles always only have output glues exposed. See Section \ref{sec:IO-marking} of the Technical Appendix for examples conversions from unmarked to IO-marked tile types.
Given an IO-marked tile type $t$, we denote its \emph{signature} 
as the string ``Input='' followed by a pair for every input side $d$ of $t$, consisting of $d$ and the integer strength of the glue on side $d$ of $t$, plus the string ``Output='' followed by a pair for every output side $d$ of $t$, consisting of $d$ and the integer strength of the glue on side $d$ of $t$. See Figure \ref{fig:signatures} for examples.
Additionally, with this notation, multiple strengths may be assigned to each direction, in which case the notation refers to a set of signatures, one for each combination of glue strengths assigned to each side. 

\begin{figure}
    \centering
    \includegraphics[width=1.5in]{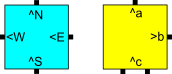}
    \caption{Examples of IO-marked tile type signatures: (Left) The light blue tile's signature is Input=(S,1),(E,1), Output=(N,1),(W,1). (Right) The yellow tile's signature is Input=(S,2), Output=(N,1),(E,1).\label{fig:signatures}}
    \vspace{-10pt}
\end{figure}




\begin{figure}
    \centering
    \includegraphics[width=4.0in]{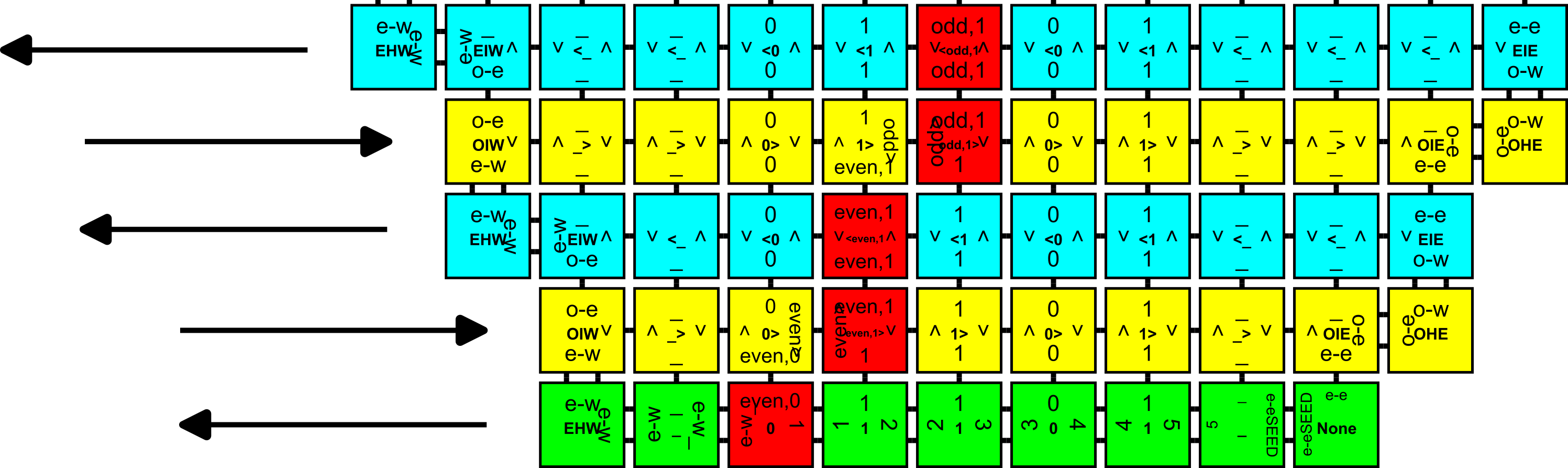}
    \caption{An example zig-zag aTAM system that simulates a Turing machine. The seed tile is the rightmost of the bottom row. The first row (green) grows right to left. After growing upward by one tile, the second row grows left to right and extends one extra tile beyond the row below. Subsequent rows continue to alternate direction and extend in length by 1. Each row represents a configuration of the Turing machine with each tile representing a tape cell, the north glues representing the contents of each cell, and the red tiles showing the location of the simulated tape head and current state of the machine. If a row is growing in the direction in which the tape head needs to move after the last transition, that occurs. If it is growing in the opposite direction, the tape head and state remain the same for that row, and then the next row (which will be of alternating direction) simulates the head movement and state change.\label{fig:zig-zag-TM}}
    \vspace{-10pt}
\end{figure}

Here we provide a quick overview of the different classes of aTAM systems considered in this paper. Formal characterizations may be found in Section~\ref{sec:aTAM-classes-appendix} of the Technical Appendix. It is assumed that all classes are IO-marked. In the first class, called \emph{zig-zag} systems, tiles never present an input glue to the south, instead growth occurs northward in rows that alternate between eastward and westward growth as illustrated in Figure~\ref{fig:zig-zag-TM}. Most tile attachments are cooperative except on the edges of the assembly and when a new row is started. Despite being the most restricted class of models considered in our results, this class of TASs is still capable of simulating the execution of arbitrary Turing machines. The next class consists of what we call \emph{standard} systems. These are directed systems where all tiles attach with exactly enough glues to meet the temperature threshold (which is $2$), no tiles attach across-the-gap, and no mismatches occur. We call such systems ``standard'' because, except for the most convoluted theoretical constructions, most systems defined in aTAM literature tend to satisfy these conditions or can easily be altered to satisfy these conditions. The third class considered consists of standard systems where across-the-gap attachments are allowed. The fourth class additionally allows mismatches but must still remain directed, in other words this class represents all directed temperature-2 systems. And finally, the fifth class consists of all aTAM systems.

    \subsection{The abstract Slat Assembly Model}\label{sec:aSAM}

The abstract Slat Assembly Model (aSAM), originally introduced in \cite{SlatAcceleration}, is a generalization of the aTAM. Since most of its definitions are analogous to those of the aTAM, in this section we provide an informal overview. (Formal definitions can be found in \Cref{sec:aSAM-appendix} of the Technical Appendix.) The primary difference between slats and tiles is that the former are defined as $n\times 1\times 1$ polyominoes of cubes in 3D space.
Therefore, with slats we expand our list of directions and sides to also include ``Up'' ($+z$ direction) and ``Down'' ($-z$ direction), resulting in the set of face directions $D = \{N,E,S,W,U,D\}$.
Similar to tiles, slats can have \emph{glues} (also referred to as \emph{handles}) on each of their $4n + 2$ faces. Each glue is identified by a string \emph{label}, and a non-negative integer \emph{strength}. Each glue has a complementary glue which shares its strength. In this paper we will often denote complementary glues using the same labels but with one appended by an asterisk (e.g.\ ``label'' and  ``label*''). Furthermore, we make a distinction between \emph{slats} and \emph{slat types}, the latter being just a description of the glues and length of a slat with no defined position or orientation. The position and orientation of slats is restricted to the 3D integer lattice and two slats which sit incident to one another are said to be \emph{attached} or \emph{bound} with strength $s$ if they share complementary glues of strength $s$ on their abutting faces. An \emph{assembly} is simply a set of slats such that no two occupy the same coordinates in $\mathbb{Z}^3$.

A \emph{slat assembly system} (SAS) $\mathcal{S} = (S, \sigma, \tau)$ consists of a finite set of slat types $S$, an assembly $\sigma$ called the \emph{seed assembly} that acts as the starting point for growth, and a positive integer $\tau$ called the \emph{binding threshold} (a.k.a. \emph{temperature}). The binding threshold describes the minimum cumulative glue strength needed for a slat to stably attach to a growing assembly. Growth in the aSAM is described by a sequence of slat attachments. Any slat which could sit on the perimeter of an assembly so that it would be attached to other slats with a cumulative strength meeting the binding threshold is a candidate for attachment, and attachments are assumed to happen non-deterministically. Any assembly that could result from a sequence of slat attachments beginning from the seed assembly of a SAS $\mathcal{S}$ and using only those slat types in the slat set of $\mathcal{S}$ is said to be \emph{producible} in $\mathcal{S}$. Any assembly that permits no additional slat attachments is called \emph{terminal}.


For all results of this paper, we work within a restricted class of systems of the aSAM satisfying the following conventions.
Slat types intended to be horizontally aligned always bind in the plane $z=1$ and we only assign glues to their $D$ faces, using only the ``starred'' versions of glue labels (i.e. those with the ``$*$'' symbol). Vertically aligned slat types always attach in the plane $z=0$, and we only assign glues to their $U$ faces, using the ``un-starred'' versions of glue labels. Additionally, we ensure that no two slats share more than one pair of complementary glues. 
Furthermore, all glues on slats are assumed to be strength-1 and each slat can only bind to any other single slat with a single glue. Therefore the temperature parameter $\tau$ effectively becomes the \emph{cooperativity} of a system (a.k.a. the \emph{coordination number}, as used in \cite{Shih-OrigamiSlats,ShihNucleation}). That is, if $\tau = c$, then each slat must cooperatively bind with $c$ other slats in order to attach to an assembly. We impose these restrictions on our designs so that their behavior is similar to the slat systems successfully experimentally demonstrated in \cite{Shih-OrigamiSlats}. Furthermore, systems with these restrictions allow for more efficient computer simulation.\footnote{A Python-based graphical simulator for the aSAM$^-$ called SlatTAS can be downloaded from \url{self-assembly.net} via a link on the page here \cite{SlatTAS}.}


\subsection{Definition of simulation of an aTAM system by an aSAM system}\label{sec:aSAM-simulation-informal}

Here we describe what is meant by an aSAM system \emph{intrinsically simulating} an aTAM system. From here on, the term ``simulation'' will refer to intrinsic simulation. This definition is analogous to the typical definition of intrinsic simulation between aTAM systems which may be found in \cite{DDDxModelSimICALP}. Due to limited space, we leave the full technical details to the Technical Appendix in Section~\ref{sec:aSAM-simulation-formal}. Here we assume that $\calT=(T, \sigma, \tau)$ is a TAS being simulated by the SAS $\calS=(S,\sigma', \tau')$.  For $\calS$ to simulate $\calT$, it must be the case that $\calS$ ``looks like'' $\calT$ at scale. To this end, we also require the definition of a \emph{macrotile representation function} $R$ and a \emph{scale factor} $m$. The function $R$ maps $m\times m$ blocks, called \emph{macrotiles}, of slat locations (which may or may not contain slats) to individual tile types in $T$. To be precise, $R$ is a partial function since it might be the case that a macrotile does not immediately map to a tile type in $T$. Once $R$ does map a macrotile to a tile type however, it must continue to map the macrotile to the same tile type regardless of any additional slats that attach within the macrotile. This property reflects the fact that tiles in the aTAM may not change type or detach after they have attached. When a slat attachment causes a macrotile to map under $R$ to tile type $t$ for the fist time, it is said that the macrotile \emph{resolves} into $t$. Applying the macrotile representation function $R$ to each macrotile of an $\calS$-assembly yields a $\calT$-assembly. This process defines the \emph{assembly representation function} $R^*$ from $\calS$-assemblies to $\calT$-assemblies. While it is allowed for macrotiles to contain slats even if it does not map to a tile type, we only allow slats to attach in macrotiles adjacent (not-diagonally) to ones which have already resolved. This prevents a ``simulator'' from growing slats to perform complex calculations in region that will never map to a tile in the simulated system and ensures that slats only grow in macrotile locations that could feasibly map to tiles. The macrotile blocks that admit slat attachments but have not yet resolved are called \emph{fuzz} regions since at a scale they resemble small hairs growing along the side of a simulated assembly.

For $\calS$ to simulate $\calT$, 3 conditions must be satisfied. First, $\calS$ and $\calT$ must have \emph{equivalent productions} meaning that $R$ surjectively maps all $\calS$-assemblies to $\calT$-assemblies and all terminal $\calS$-assemblies to terminal $\calT$-assemblies. Second $\calT$ must \emph{follow} $\calS$ meaning that all sequences of slat attachments in $\calS$ map to corresponding slat attachments in $\calT$ ($\calT$ can only do what $\calS$ does). And finally, $\calS$ must \emph{model} $\calT$ meaning that all sequences of tile attachments in $\calT$ have at least one corresponding sequence of slat attachments in $\calS$ ($\calS$ can only do what $\calT$ does). The formal definition of \emph{models} also has a provision that ensures all non-deterministic attachments in $\calT$ are truly simulated by non-deterministic attachments in $\calS$ rather than being predetermined in advance.

\section{Results}\label{sec:results}

In this section we present our results showing that classes of aTAM systems, with increasingly complex dynamics, can be simulated by aSAM systems. Each result is proven by construction and associated software for designing, converting, simulating, and visualizing these systems can be found online: \url{http://self-assembly.net/wiki/index.php/Abstract_Slat_Assembly_Model_(aSAM)}. Note that the first four results are for classes of aTAM systems defined to have $\tau=2$, but each construction trivially works for $\tau=1$ as well, simply by treating all glues of the simulated aTAM systems as $\tau$-strength. The final result is presented for $\tau=2$, but explanation of a simple extension to handle arbitrary values of $\tau$ is presented in the Technical Appendix (as are the details of most proofs) due to space constraints.

\subsection{Zig-zag systems}\label{sec:zig-zag}

Zig-zag systems are particularly interesting because, despite their incredibly limited range of dynamics, they are computationally universal \cite{CookFuSch11,DirectedNotIU,jCCSA,jSADS,SingleNegative}. Our first result shows that any zig-zag aTAM system can be simulated by an aSAM system with macrotiles of size only $2c \times 2c$.

\begin{theorem}\label{thm:zig-zag}
    Let $\calT = (T, \sigma, 2)$ be an arbitrary zig-zag aTAM system. For any $c > 2$ such that $c \mod 2 = 0$, there exists an aSAM system $\mathcal{S} = (S,\sigma',c)$ and macrotile representation function $R$ such that $\mathcal{S}$ simulates $\calT$ under $R$ using cooperativity $c$ and macrotiles of size $2c \times 2c$. Furthermore, the longest slat in $S$ is of length $3c$.
\end{theorem}

\begin{proof}
We prove Theorem \ref{thm:zig-zag} by construction, and thus, starting with arbitrary zig-zag aTAM system $\calT = (T, \sigma, 2)$ and given any $c > 2$ such that $c \mod 2 = 0$ we show how to create aSAM system $\mathcal{S} = (S, \sigma', c)$ and macrotile representation function $R$ such that $\mathcal{S}$ simulates $\calT$ under $R$. First, without loss of generality we assume that $\calT$ grows its first row from its seed tile from the right to the left (i.e. ``RtoL''), and then its second row grows immediately above that, from left to right (i.e. ``LtoR''), and then all subsequent rows zig-zag from RtoL then LtoR. (The construction could simply be rotated appropriately to handle any direction of zig-zag growth.)

\begin{figure}
    \centering
    \begin{subfigure}{0.15\textwidth}
        \centering
        \includegraphics[width=1.0\textwidth]{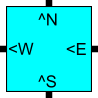}
        \caption{\label{fig:zig-zag-tile}}
    \end{subfigure}
    \hspace{20pt}
    \begin{subfigure}{0.2\textwidth}
        \centering
        \includegraphics[width=1.0\textwidth]{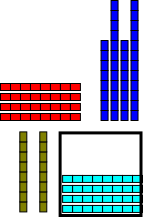}
        \caption{\label{fig:zig-zag-slat-groups}}
    \end{subfigure}
    \hspace{20pt}
    \begin{subfigure}{0.17\textwidth}
        \centering
        \includegraphics[width=1.0\textwidth]{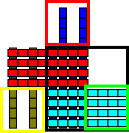}
        \caption{\label{fig:zig-zag-macrotile}}
    \end{subfigure}
    \hspace{20pt}
    \begin{subfigure}{0.15\textwidth}
        \centering
        \includegraphics[width=1.0\textwidth]{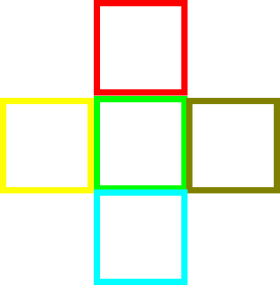}
        \caption{\label{fig:cell-coloring}}
    \end{subfigure}
    
    \caption{(a) An IO-marked tile type $t$ from a zig-zag aTAM tile set, for a row that grows RtoL. It has strength-1 inputs on the south and east, and strength-1 outputs on the north and west. (b) A set of 4 slat groups for the macrotile simulating $t$ at $c = 4$. The light blue group contains body slats that are entirely within the square of the macrotile (depicted by the black square) which they cause to resolve to $t$. The dark blue group contains two body slats and two output slats (i.e. the two extending into the north macrotile to serve as the north output of strength-1). The red and gold slat groups combined are the output slats that serve as the west output and extend into the western neighboring macrotile location. (c) An example of the assembled $2c \times 2c$ macrotile for $t$, with cells marked to show the portions of $t$ that they represent, following the conventions of (d). (d) A cell enclosed in a green square represents the cell in which the initial body slats of a macrotile bind, causing it to resolve to $t$. The cells enclosed in red, gold, light blue, and yellow squares denote the cells in which the slats expose glues representing the output glues of the north, east, south, and west sides of $t$, respectively.\label{fig:zig-zag-slats}}
    \vspace{-10pt}
\end{figure}

Each tile in $\calT$ is simulated by a macrotile of size $2c \times 2c$ in $\mathcal{S}$. For $c = 4$, this means that the $8 \times 8$ square whose southwest coordinate is $(8i,8j)$, for every $i,j \in \mathbb{Z}$, will map under $R$ to either empty space or to a tile in $\calT$. An example is shown in Figures \ref{fig:zig-zag-slats} and \ref{fig:zig-zag-growth}. Each slat in a macrotile is of a unique type.\footnote{Using techniques of \cite{SlatAcceleration}, it is possible to reuse slat types within macrotiles to reduce the slat complexity, but for ease of explanation we present our constructions without that optimization.} We use the term \emph{slat group} to refer to each set of $c$ (or sometimes $c/2$) slats that are oriented in the same direction and grouped together (both logically, and also in that each slat in a slat group can attach to a growing assembly at exactly the same time as the others in that group). (For example, in Figure \ref{fig:zig-zag-slat-groups} there are 4 slat groups.) For convenience, we will characterize all of the slat types of a macrotile in two categories: (1) \emph{body slats}: slats that are completely contained within one of the $2c \times 2c$ macrotile regions and either (a) their binding causes that region to map to a tile in $\calT$ under $R$, or (b) they bind after that macrotile has resolved, and (2) \emph{output slats}: slats that either (a) cause a macrotile to map to a tile of $\calT$ but also extend into a neighboring macrotile location, or (b) bind in a macrotile location this is unresolved both before and immediately after their binding. (In Figure \ref{fig:zig-zag-macrotile}, the 4 light blue, and the shorter 2 of the dark blue slats are body slats. The longer two dark blue, and all of the red and gold slats are output slats. Furthermore, the longer dark blue slats are of length $3c$, which is the greatest length of slats in this construction.)

\begin{figure}
    \centering
    \includegraphics[width=0.4\textwidth]{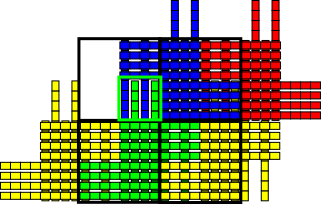}
    \caption{An example of a portion of an assembly composed of $2c \times 2c$ macrotiles (some partial) simulating a zig-zag aTAM system for $c = 4$. Four of the macrotile locations are outlined in black squares. The macrotile simulating $t$ from Figure \ref{fig:zig-zag-slats} would attach into the top left such macrotile location. Its (light blue) body slats would attach to the 2 west output slats from the macrotile to the east (dark blue) and the 2 north output slats from the macrotile to the south (green). These light blue slats can bind in any order, and as soon as one binds, the macrotile resolves to $t$. Due to the fact that $\tau = 4$, only once they have all 4 bound can the north output slats (dark blue) bind. Only once all 4 of those have bound can the red output slats bind, then finally the 2 gold output slats. Thus, the growth of a macrotile is well-ordered, and outputs are only presented after a macrotile resolves, enforcing the restrictions of simulation.\label{fig:zig-zag-growth}}
    \vspace{-10pt}
\end{figure}

Let $t_n \in T$, for $0 \le n < |T|$, be the $n$th tile in tile set $T$. We will refer to the string ``$t_n$'' as the unique name of $t_n$. We now discuss how the slats that form a macrotile simulating $t_n$ are designed.
Given the directions of growth, and the dynamics of a zig-zag aTAM system, the following list contains all possible valid signatures for $t_n$ of any zig-zag system (with the trivial exception that some tile type could have one or more fewer outputs, and the binding of such a tile into an assembly would cause growth to terminate and the assembly to become terminal, as such tiles are trivially handled by macrotiles without corresponding output slats).

\begin{enumerate}
    \item Initial (seed) row tiles (Figure \ref{fig:zig-zag-seed}):
        \begin{enumerate}
            \item Seed tile: Input=$\emptyset$, Output=(W,2),(N,1)
            \item Row interior tiles: Input=(E,2), Output=(W,2),(N,1)
            \item Leftmost tile: Input=(E,2), Output=(N,2)
        \end{enumerate}
    \item LtoR row tiles (Figure \ref{fig:zig-zag-LtoR}):
        \begin{enumerate}
            \item Leftmost tile: Input=(S,2), Output=(E,1),(N,1)
            \item Row interior tiles: Input=(W,1),(S,1), Output=(E,1),(N,1)
            \item Right row pre-extension tile: Input=(W,1),(S,1), Output=(E,2),(N,1)
            \item Right row extension tile: Input=(W,2), Output=(E,2),(N,1)
            \item Rightmost tile: Input=(W,2), Output=(N,2)
        \end{enumerate}
    \item RtoL row tiles (Figure \ref{fig:zig-zag-RtoL}):
        \begin{enumerate}
            \item Rightmost tile: Input=(S,2), Output=(W,1),(N,1)
            \item Row interior tiles: Input=(E,1),(S,1), Output=(W,1),(N,1)
            \item Left row pre-extension tile: Input=(E,1),(S,1), Output=(W,2),(N,1)
            \item Left row extension tile: Input=(E,2), Output=(W,2),(N,1)
            \item Leftmost tile: Input=(E,2), Output=(N,2)
        \end{enumerate}
\end{enumerate}

Figures \ref{fig:zig-zag-seed}-\ref{fig:zig-zag-RtoL} show tiles with those signatures and their corresponding \emph{macrotile templates}, which are sets of slat groups that correspond to the particular set of input and output glue directions and strengths that a simulated tile has.
To build $S$, for each $t_n$, we instantiate the macrotile template associated with $t_n$'s signature.
Instantiating a macrotile consists of first making a unique copy of each slat type in the macrotile template whose name has the prefix ``$t_n$'' prepended to the unique name of that slat type in the macrotile template.
We will refer to the set of slats for the macrotile template instantiated for $t_n$ as $S_n$.
For every location where a vertical slat of $S_n$ is at the same $(x,y)$ coordinates as a horizontal slat of $S_n$ (but under it since it will be at $z=0$ and the horizontal slat at $z=1$), an un-starred glue unique to that location is placed there on the vertical slat, and the starred complement of that glue is placed on the horizontal slat. (For technical details about the conventions used to generate the unique glue labels see Section \ref{sec:glue-naming} of the Technical Appendix.) All such glues are also given the prefix ``$t_n$'' to ensure that they will not match glues of macrotiles instantiated for any other tiles. These glues are called \emph{interior glues}, since they bind slats of the same macrotile to each other.

\begin{figure}
    \centering
    \includegraphics[width=2.5in]{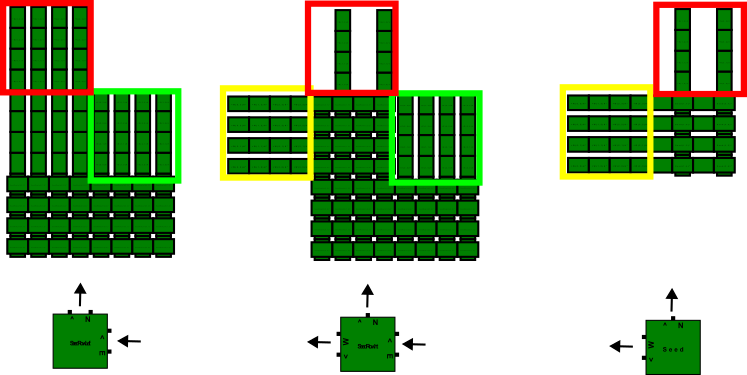}
    \caption{The tiles for the seed and initial row of a zig-zag aTAM system (that grows RtoL from the seed tile), and their corresponding slat-based macrotile templates for $c = 4$. The rightmost corresponds to the seed tile, the middle corresponds to the interior tiles of the initial row, and the leftmost corresponds to the leftmost tile of the initial row. Cells are bounded by squares to show their functionality, following the coloring convention from Figure \ref{fig:cell-coloring}.\label{fig:zig-zag-seed}}
\end{figure}

The final step of building $S_n$ for $t_n$ is to account for the glues of $t_n$ and to place glues on the slats of $S_n$ to cause their behavior to be simulated. To do this, within the $c \times c$ cell representing a glue of $t_n$ (whose location is determined by the particular macrotile template that matches $t_n$'s signature), we label the glue in each location of each slat with the same glue label as the corresponding glue on $t_n$, followed by the cell coordinates of the location (i.e. ``$(i,j)$'' for $0 \le i,j < c$, with $(0,0)$ being the south-westernmost location), followed by a star for glues on horizontal slats. This guarantees that each glue label in each cell is unique. (Note that for some later constructions an additional ``marker'' symbol may be necessary for these glue labels, and details can be found in the Technical Appendix in Section \ref{sec:glue-naming}.)
Output glues of strength-2 on $t_n$ are represented by $c$ slats filling a $c \times c$ cell, allowing slats of the opposite orientation to attach to the assembly by binding solely to them (analogous to the strength-2 glue of $t_n$ being sufficient to allow a tile attachment). Output glues of strength-1 are represented by just $c/2$ slats extending across a $c \times c$ cell. Because of this, $c/2$ additional slats extending across that cell from the opposite direction, representing the strength-1 output glue of an adjacent macrotile, are required before the necessary glues are in place to allow body slats for the next macrotile to attach. In this way, cooperative behavior is enforced. An example can be see in Figure \ref{fig:zig-zag-growth}.

\begin{figure}
    \centering
    \includegraphics[width=5in]{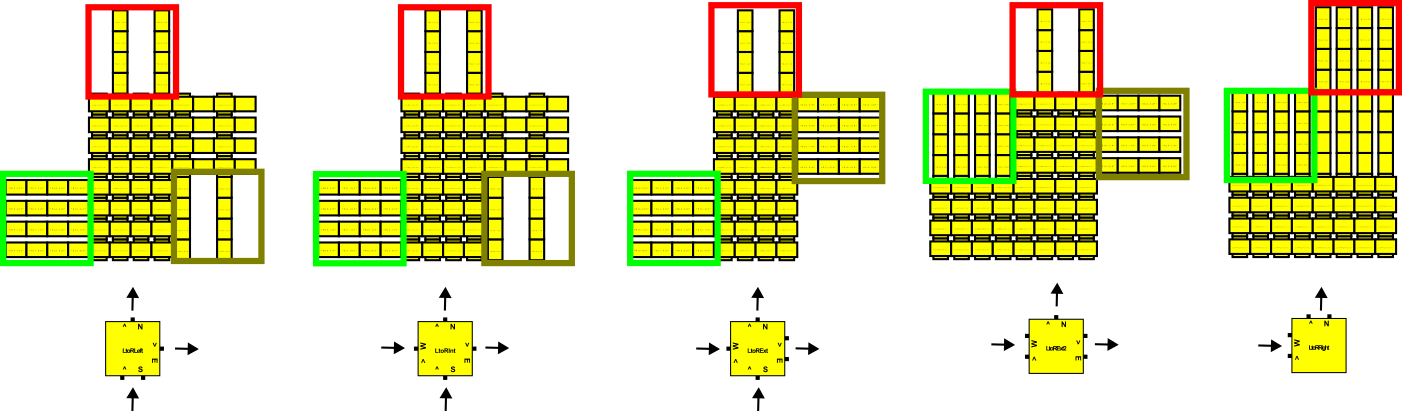}
    \caption{Tile types of all possible valid signatures for tiles of a row that grows LtoR in a zig-zag aTAM system, and their corresponding slat-based macrotile templates for $c = 4$. Cells are bounded by squares to show their functionality, following the coloring convention from Figure \ref{fig:cell-coloring}.\label{fig:zig-zag-LtoR}}
\vspace{-10pt}
\end{figure}

The design of the conventions and macrotile templates guarantee that slats can bind only in desired locations, and that they do so with total binding strength exactly $c$. For a vertical slat to attach, it must initially bind with $c$ distinct horizontal slats, and vice versa. (Recall that only strength-1 glues are used, and also that all vertical slats have all glues on their $U$ sides, which are un-starred, and all horizontal slats have all glues on their $D$ sides, which are starred.) From Figure \ref{fig:zig-zag-slats} it is clear to see how an individual macrotile (representing a tile in a RtoL row) assembles in a well-ordered progression, causing it to first resolve and only then attach slats that provide the outputs. 
It is also clear how any tile from $T$ can be converted to a set of slats that will simulate it in a similar way, simply noting the signature for any such tile and the macrotile template for the matching signature, selected from those shown in Figures \ref{fig:zig-zag-seed}-\ref{fig:zig-zag-RtoL}.

\begin{figure}
    \centering
    \includegraphics[width=5in]{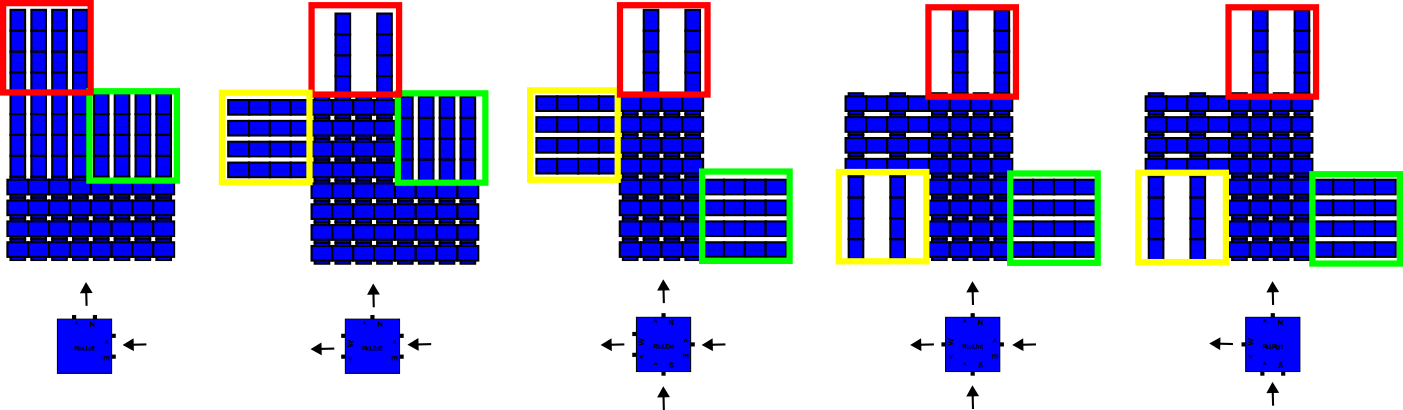}
    \caption{Tile types of all possible valid signatures for tiles of a row that grows RtoL in a zig-zag aTAM system, and their corresponding slat-based macrotile templates for $c = 4$.  Cells are bounded by squares to show their functionality, following the coloring convention from Figure \ref{fig:cell-coloring}.\label{fig:zig-zag-RtoL}}
\vspace{-10pt}
\end{figure}

By inspecting the macrotile templates associated with the valid zig-zag tile signatures, it can be verified that the outputs of any macrotile are always positioned correctly for the binding of body slats of a macrotile that needs to use those as inputs, while keeping that next macrotile in the correct relative position. The careful design of the glues ensures that only the slats designed to attach to any given location can do so.
The macrotile representation function $R$ can simply contain a mapping of body slats to the tile types from which they were derived and use that mapping for any macrotile location containing a body slat, while mapping any macrotile location without a body slat to an empty location. The seed $\sigma'$ simply consists of the set of slats of the macrotile created for the seed tile of $\calT$, which has exactly one $c \times c$ cell where there are $c$ slats representing the strength-2 output glue of $\calT$'s seed and to which a slat can bind. 
Starting from this assembly it is also clear to see that, as the macrotiles of $\mathcal{S}$ assemble, there will always be exactly one $c \times c$ cell in which slats can bind. In $\calT$, the frontier is always of size 1, and any slats that can attach in $\mathcal{S}$ will either be (1) body slats that cause the macrotile location mapping to that frontier location to resolve under $R$ into the next tile that could attach in $\calT$, or (2) body or output slats of the macrotile that was most recently resolved in that frontier location. This provides an inductive argument where the inductive hypothesis is that, given an assembly $\beta$ producible in $\mathcal{S}$ mapping under $R$ to assembly $\alpha$ producible in $\calT$, the exposed glues on $\beta$ allow exactly one macrotile, mapping to the correct next tile of $\alpha$ under $R$, to assemble next. This is true of the seed macrotile (the base case), and also given any assembly producible via the macrotiles generated by the macrotile templates shown in Figures \ref{fig:zig-zag-seed}-\ref{fig:zig-zag-RtoL}, so the induction holds and $\mathcal{S}$ correctly simulates $\calT$ under $R$. Therefore, $\mathcal{S}$ simulates $\calT$, an arbitrary zig-zag aTAM system, under $R$ using cooperativity $c$ and macrotiles of size $2c \times 2c$ with the longest slats being of length $3c$ and Theorem \ref{thm:zig-zag} is proven.
\end{proof}


\subsection{Standard systems}\label{sec:standard} 

Next, we prove that by only slightly increasing the scale factor of the simulation, i.e. the size of macrotiles, from $2c \times 2c$ for zig-zag systems to $3c \times 3c$, that any standard aTAM system can be simulated by an aSAM system. Since the majority of aTAM constructions in the literature are standard systems (e.g. \cite{jCCSA,jSADS,SolWin07,RotWin00}), this shows that a very modest scale factor can be used to simulate a huge diversity of very complex aTAM systems.

\begin{theorem}\label{thm:standard}
    Let $\calT = (T,\sigma,2)$ be an arbitrary standard aTAM system. For any $c > 2$ such that $c \mod 2 = 0$, there exists an aSAM system $\mathcal{S} = (S,\sigma',c)$ such that $\mathcal{S}$ simulates $\calT$ using cooperativity $c$ and scale factor $3c$. Furthermore, the longest slat in $S$ is of length $3c$.
\end{theorem}

The simulation construction for standard aTAM systems is very similar to the construction for zig-zag systems, except that a larger set of input and output direction combinations need to be considered. To accommodate this change, the macrotiles during standard aTAM simulations are $3c \times 3c$ instead of $2c \times 2c$, but the argument is essentially unchanged. Strength-$2$ glues are still simulated by leaving all $c$ slats available for binding in a corresponding macrotile cell while strength-$1$ glues are simulated by using half this many from each of two outputs. However, the geometries of the macrotiles are a bit different. To handle the more diverse sets of signatures, our construction makes use of macrotile templates of different geometries for tiles with strength-2 input glues (see Figures~\ref{fig:standard_3x3_S2} and \ref{fig:standard_3x3_S2_example_example}) versus those with two strength-1 input glues (see Figure~\ref{fig:standard_3x3_SW_example}), as well as \emph{output slat templates} that differ for strength-2 versus strength-1 output glues as well as for those that are on vertical sides (north, south) versus horizontal sides (east, west) (see Figure~\ref{fig:standard_3x3_S2_example_example}). Otherwise, all of the same techniques from the proof of Theorem~\ref{thm:zig-zag} apply. Figure~\ref{fig:standard_3x3_examples} shows a few example macortiles, and the full proof is in Section~\ref{sec:standard-append} of the Technical Appendix.



\begin{figure}
\centering
\begin{subfigure}{0.15\textwidth}
    \centering
    \includegraphics[width=1.0\textwidth]{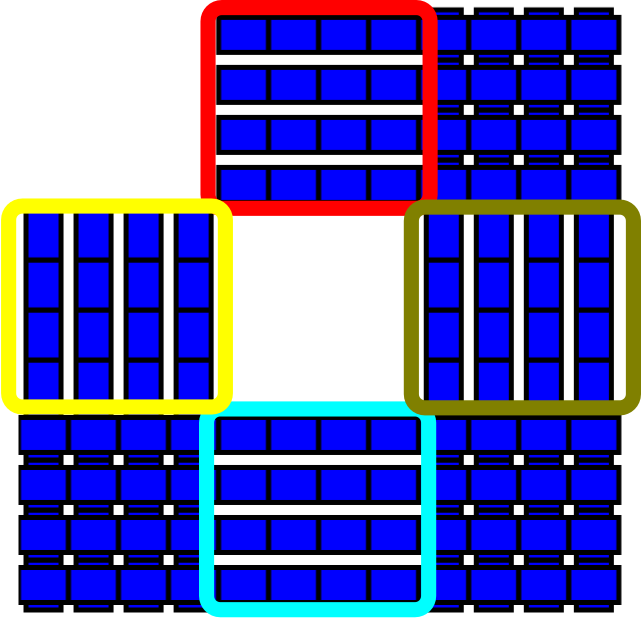}
    \caption{\label{fig:standard_3x3_S2}}
\end{subfigure}
\hfill
\begin{subfigure}{0.225\textwidth}
    \centering
    \includegraphics[width=1.0\textwidth]{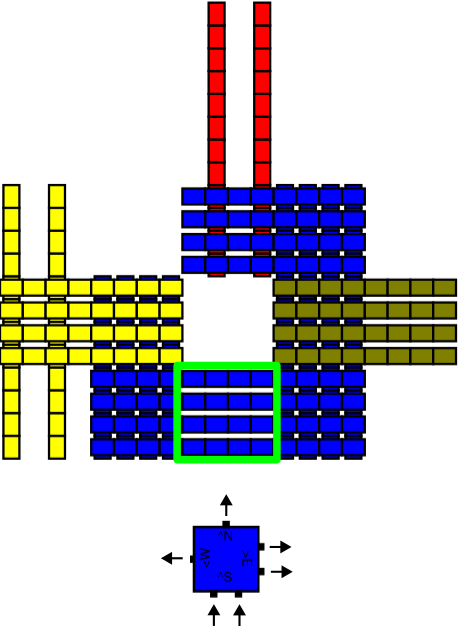}
    \caption{\label{fig:standard_3x3_S2_example_example}}
\end{subfigure}
\hfill
\begin{subfigure}{0.1875\textwidth}
    \centering
    \includegraphics[width=1.0\textwidth]{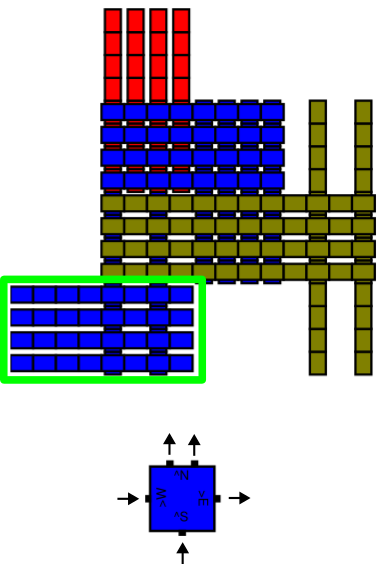}
    \caption{\label{fig:standard_3x3_SW_example}}
\end{subfigure}
\caption{(a) Strength-2 macrotile template for a standard aTAM system. Cells are bounded by squares to show their functionality, and mark cell locations where output slat templates may be added to the macrotile. One of the marked cells may be designated as an input, and have its domains assigned such that they connect with those provided by the output slats of a neighboring macrotile whose output is of the same glue type. Cells are signified using the same color conventions as Figure~\ref{fig:cell-coloring}. (b) Macrotile experiencing south strength-2 input. (b) Macrotile exhibiting south and west strength-1 inputs. Output slat templates are colored in accordance to Figure~\ref{fig:cell-coloring}, and input domain locations are marked with a green box. \label{fig:standard_3x3_examples}}
\vspace{-10pt}
\end{figure}



\subsection{Standard with across-the-gap simulation}\label{sec:standard-across-the-gap}
Next, we prove that by only slightly increasing the maximum slat length of the simulation, from $3c$ to $4c$, any standard aTAM system with across-the-gap cooperation can be simulated by an aSAM system supporting both types of cooperative binding (adjacent and across-the-gap). 

\begin{theorem}\label{thm:standard+across-the-gap}
    Let $\calT = (T,\sigma,2)$ be an arbitrary standard with across-the-gap aTAM system. For any $c > 2$ such that $c \mod 2 = 0$, there exists an aSAM system $\mathcal{S} = (S,\sigma',c)$ such that $\mathcal{S}$ simulates $\calT$ using cooperativity $c$ and macrotiles of size $3c \times 3c$. Furthermore, the longest slat in $S$ is of length $4c$.
\end{theorem}

The construction for Theorem~\ref{thm:standard+across-the-gap} is similar in form to the previous two, so here we just describe a few important features. As with the proof of Theorem~\ref{thm:standard}, there are distinct macrotile templates for tiles with strength-2 inputs, but here there are also distinct macrotile templates for tiles with two strength-1 input glues that are adjacent and those that are across-the-gap. Again, there are output slat templates specific to strengths and orientations. Specifically, across-the-gap cooperation is handled in the center cell of each macrotile. Growth in each macrotile is otherwise very similar. All of the same techniques from the proof of Theorem~\ref{thm:zig-zag} apply. Full details may be found in Section \ref{sec:standard-across-the-gap-append} of the Technical Appendix.

\subsection{Directed temperature-2 simulation}\label{sec:determ-with-mismatches}

\begin{theorem}\label{thm:deterministic+mismatches}
    Let $\calT = (T,\sigma,2)$ be an arbitrary directed temperature-2 aTAM system. For any $c > 2$ such that $c \mod 2 = 0$, there exists an aSAM system $\mathcal{S} = (S,\sigma',c)$ such that $\mathcal{S}$ simulates $\calT$ using cooperativity $c$ and macrotiles of size $4c \times 4c$. Furthermore, the longest slat in $S$ is of length $4c$.
\end{theorem}

This construction follows the same general format of the previous 3, though we defer it to Section~\ref{sec:determ-with-mismatches-append} of the Technical Appendix due to space constraints. Just as with the previous simulations, strength-2 aTAM glues are simulated using $c$ slats in neighboring macrotiles, while strength-1 glues are simulated using half that many. The main difference between this simulation and the previous ones comes from the fact that mismatches are allowed to occur in the simulated aTAM system. In order to simulate this behavior, it must be guaranteed that the presence of any additional output slats in a macrotile for any tile type $t$ both do not prevent the macrotile from resolving to $t$, and do not block any outputs from $t$ apart from those which are already occupied. To accommodate this, the scale factor is increased to provide specific macrotile cells wherein body slats may attach for every combination of potential input glues. To this end, output slats are also generally longer and reach into multiple cells of the adjacent macrotiles in order to ensure that there are cells corresponding to each combination of glues that may contribute to simulate a tile attachment.



\subsection{Full aTAM simulation}\label{sec:full-sim}

In this section, we present a theorem stating that all temperature-2 aTAM systems can be simulated by aSAM systems and give a brief overview of the proof's construction. Due to space constraints, we sketch our construction and just mention that arbitrary temperatures can be handled, and further details can be found in Section \ref{sec:full-sim-appendix} of the Technical Appendix.

\begin{theorem}\label{thm:full-aTAM}
    Let $\calT = (T, \sigma, 2)$ be an arbitrary aTAM system. For any $c > 2$ such that $c \mod 2 = 0$, there exists an aSAM system $\mathcal{S} = (S,\sigma',c)$ such that $\mathcal{S}$ simulates $\calT$ using cooperativity $c$. Furthermore this simulation uses a scale factor of $5c$ and a maximum slat length of $5c$.
\end{theorem}

\begin{figure}
    \centering
    \begin{subfigure}{\textwidth}
        \centering
        \includegraphics[width=0.6\textwidth]{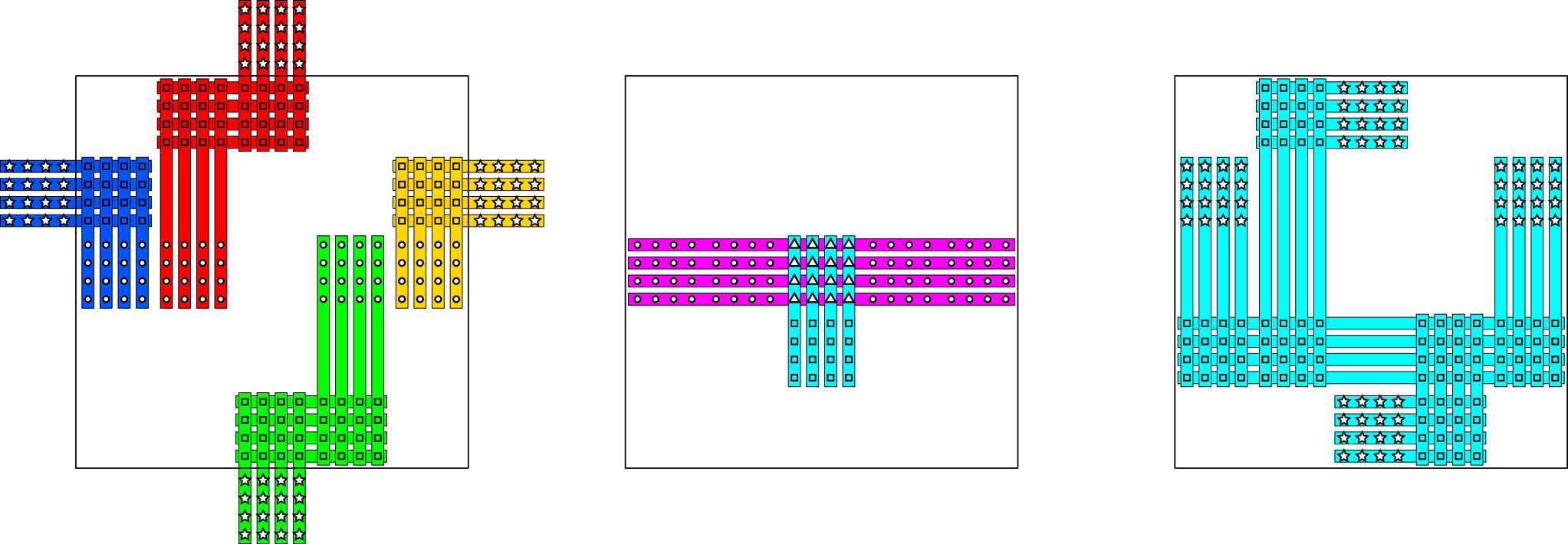}
        \caption{\label{fig:undirected-in-dec-out}}
    \end{subfigure}
    
    \begin{subfigure}{\textwidth}
        \centering
        \includegraphics[width=0.95\textwidth]{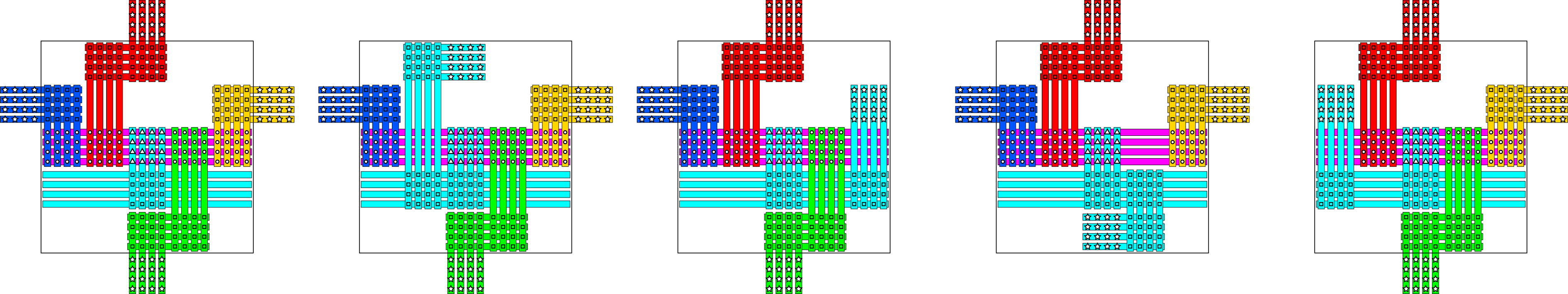}
        \caption{\label{fig:undirected-macrotiles}}
    \end{subfigure}

    \caption{(a)Left: Input slats for all directions. Center: Decision slats in the decision rows of a macrotile. Right: Output slats growing in all directions. (b) Left: a macrotile receiving inputs from all 4 of its neighbors. Red slats encode an incoming glue from the north, yellow from the east, green from the south, and blue from the west. Magenta slats attach non-deterministically to the glues presented by these slats and each encode a possible tile from $T$ to which this macrotile may resolve. Cyan slats decide a winner among the magenta slats. The remaining illustrations are example macrotiles which only receive input from 3 sides so the remaining side may act as an output.}
    \vspace{-20pt}
\end{figure}

In the construction of this proof, it is assumed that the aTAM system consists of IO-marked tiles (otherwise the method discussed in Section \ref{sec:IO-marking} of the Technical Appendix can be used to make it so), simulation takes place using $5c \times 5c$ macrotiles, and slats are always defined in groups of $c$. The general layout of the macrotiles does not change significantly with the type of tile being simulated, though some slats may or may not appear depending on whether the simulated tile has glues on all sides. 

Slats in $S$ may be divided logically into 3 families, \emph{input}, \emph{decision}, and \emph{output} slats, each of which is responsible for a different function. The general form of these slats is illustrated in Figure~\ref{fig:undirected-in-dec-out}. When a neighboring macrotile has resolved, it will eventually present glues along its sides indicating which output glues are present on the simulated tile. Input slats attach to these glues and act to move the information about simulated glues to the center of the macrotile forming $c$ horizontal rows. The left of Figure~\ref{fig:undirected-in-dec-out} illustrates input slats from the 4 cardinal directions using different colors. Note that the glues holding input slats together are unique to their specific location and the aTAM glue being represented. Consequently, input slats always attach as a group. Once enough input slats have attached, the central horizontal rows encode all the information about present input glues of the adjacent simulated tiles. In these horizontal rows, decision slats may attach (illustrated as magenta in Figure~\ref{fig:undirected-in-dec-out}. Decision slats are defined per tile type in the simulated system and the glues present on the decision slats ensure that each may only attach when the corresponding input glues are present, as encoded by the input slats. For instance, when simulating a tile attachment using a strength-2 north glue, the corresponding decision slats will be able to bind solely to the north input slats encoding the respective aTAM glue. On the other hand, when simulating a cooperative attachment, say from the north and west, the corresponding decision slats will only have half the necessary glues to attach to both the north and west input slats. In this way the decision slats for tile type $t\in T$ may only attach when the input slats encoding the input glues of $t$ are present. Note that in the case of overbinding or mismatches, there might be multiple decision slats that may attach in the center rows, allowing for the simulation of an undirected attachment. The specific tile type to which the macrotile resolves is the one encoded by the decision slat that attaches in the northmost decision row.
Vertical slats attaching to these decision tiles propagate the information from this northmost decision row into the row of $c$ slats below, which will be where the output slats attach. Once the macrotile has resolved and these vertical slats attach, the corresponding output slats will grow to each side that represents an output glue (and isn't already occupied). Output slats going to all directions are illustrated on the right of Figure~\ref{fig:undirected-in-dec-out}, but in actuality, only those corresponding to output glues on the simulated tile will be present. Non-determinism is thus only present in two locations of a macrotile during this simulation. First, multiple decision slats may be able to attach in the decision rows, and second when simulating tiles with mismatched glues, it may be possible for two adjacent macrotiles to present output slats to abutting sides. This is however not a problem since for this to occur, both macrotiles must have already resolved and while it may lead to input glues growing where output glues should have, this only happens in locations dedicated to the corresponding direction and thus cannot affect other parts of the macrotile. It should also be noted that this construction can be generalized to accommodate simulating arbitrary temperature thresholds, instead of just $\tau=2$, with minimal modification. Details to this end are described in Section~\ref{sec:full-sim-appendix} of the Technical Appendix.  


\bibliography{tam,experimental_refs,slats,ca}

\newpage

\section{Technical Appendix}
This Technical Appendix contains formal definitions and proof details omitted from the main body due to space constraints.

\subsection{IO marking example}\label{sec:IO-marking}

\begin{figure}[h!]
    \centering
    \includegraphics[width=0.7\textwidth]{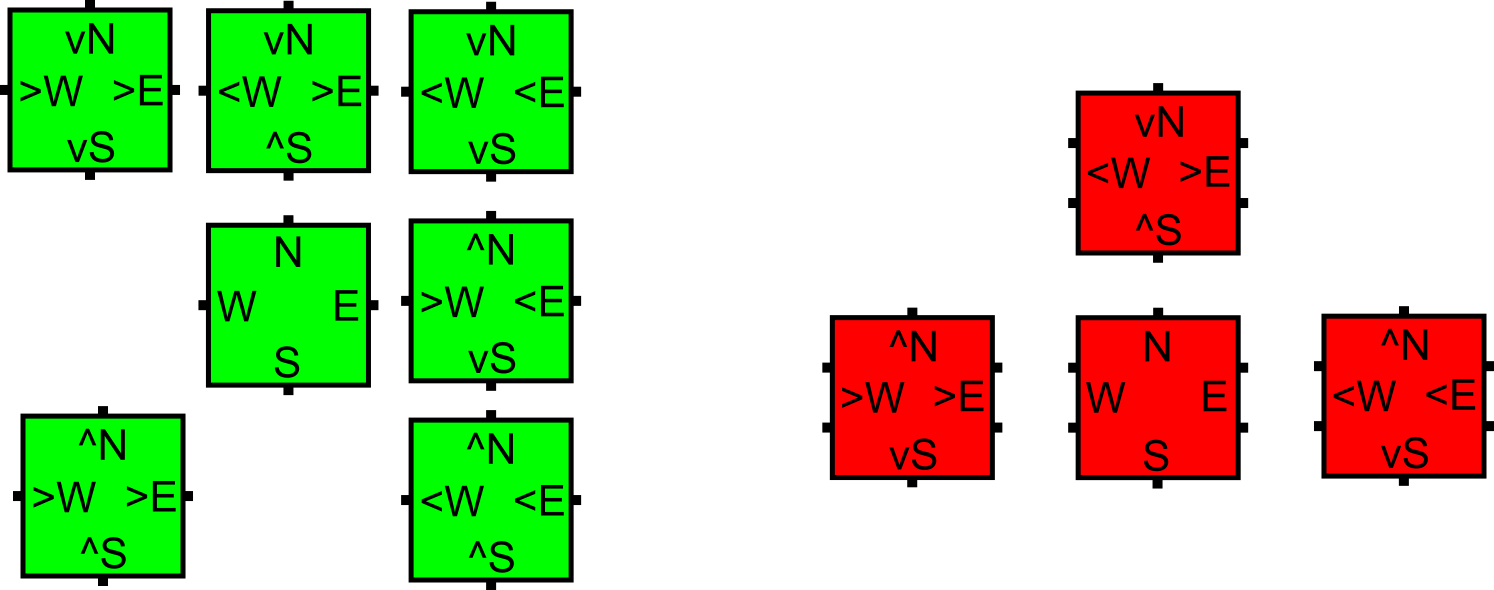}
    \caption{To make a set of IO-marked tile types whose collective behavior will be the same as an un-marked tile type, one IO-marked tile type is created for every minimal set of glues whose combined strength is $\ge \tau$ (minimal in that, if any individual glue was removed from the set, the combined strength would be $< \tau$), with those marked as input glues and the others as output glues. Two examples are shown. (Left) In the center, the un-marked tile type has a strength-1 glue on each side. Surrounding it are the six IO-marked tile types created from it, one for each possible pair of strength-1 glues marked as inputs. Note that this is the worst-case increase in tile complexity. (Right) In the center, the un-marked tile type has two strength-2 glues and two strength-1 glues. Surrounding it are the three IO-marked tile types created from it, one each with one of the strength-2 glues as the sole input, and the other with the pair of strength-1 glues as the inputs.}
    \label{fig:IO-marking}
\end{figure}

Here we give a simple demonstration of how a regular, unmarked aTAM tile set $T$ can be used to generate an equivalent an IO-marked aTAM tile set $T_{IO}$. First we note that this example demonstrates how an arbitrary, already existing tile set $T$ can be used to generate an equivalent $T_{IO}$, resulting in a constant-sized increase in tile complexity. Namely, in the worst case, each tile type of $T$ requires the creation of $6$ unique tile types in $T_{IO}$. Depending on the number and strengths of the glues on a tile type in $T$, the number of tile types generated for $T_{IO}$ could range from $1$ to $6$. However, an increase in tile complexity is not necessarily required for the creation of an IO-marked tile set, since for systems such as zig-zag systems, it is known in advance which glues of any given tile type may ever serve as its input glues (and that set is fixed for every given tile type). The IO-marked tile set in that case doesn't require any more tile types than the unmarked set. See Figure~\ref{fig:IO-marking} for two examples and explanation.

\subsection{Formal characterization of classes of aTAM systems}\label{sec:aTAM-classes-appendix}

\begin{definition}[IO TAS]\label{def:IO-system}
    An aTAM TAS $\calT = (T, \sigma, 2)$ is an \emph{IO TAS} iff all tile types in $T$ are IO-marked and all non-null glues on the perimeter of $\sigma$ have output markings.
\end{definition}

\begin{observation}\label{obs:IO}
    Given an IO TAS $\calT$, all non-null glues exposed on the perimeter of any producible assembly of $\calT$ have output markings. 
\end{observation}

Observation \ref{obs:IO} follows immediately from the a simple inductive argument. As the base case, the smallest producible assembly, the seed assembly, has only output-marked glues on its perimeter. The induction hypothesis is that, starting with a producible assembly with only output-marked glues on its perimeter, the addition of any tile to form a new producible assembly also results in an assembly with only output-marked glues on its perimeter. The induction hypothesis is proven by noting that all tiles are IO-marked and since input-marked glues can only bind to output-marked glues and no tile has input-marked glues whose strengths sum to $> 2$, any tile that $\tau$-stably binds to an assembly must do so by binding all of its input-marked glues, leaving only output-marked glues to potentially be unbound and exposed on the perimeter.

So called \emph{zig-zag} aTAM systems were originally defined in \cite{CookFuSch11} and are widely used in the literature (e.g. \cite{DirectedNotIU,jCCSA,jSADS,SingleNegative}) due to their very simple dynamics and the fact that they are computationally universal (i.e. for every Turing machine there exists a zig-zag aTAM system that simulates it). Here we provide a definition that utilizes IO-marked systems, but note that regular aTAM systems can easily be converted to IO systems with only a constant increase in the number of tile types (as done in \cite{Versus} and depicted in Section~\ref{sec:IO-marking}), and that zig-zag system can be designed to be IO-marked without any additional tile types being required.

Intuitively, a zig-zag tile assembly system is a system that grows to the left or right, grows upward by one tile, and then grows in the opposite direction.
(Note that our definition restricts zig-zag systems to add new rows only to the north, but we can trivially rotate such systems so that growth is in any one of the cardinal directions.  Therefore, for ease of presentation and w.l.o.g. we simply consider northward growing zig-zag systems, and our results still hold for the more general definition.)  For an example of a zig-zag system, see Figure~\ref{fig:zig-zag-TM}.

\begin{definition}[Zig-zag TAS]\label{def:zig-zag}
    An aTAM system $\calT = (T,\sigma,2)$ is a \emph{zig-zag} system if the following conditions hold:
    \begin{enumerate}
        \item $\calT$ is an IO TAS.

        \item $|\sigma| = 1$ (We define zig-zag systems to have seeds consisting of a single tile, but any system which has a seed consisting of a longer single row of tiles but is otherwise a zig-zag system can be trivially converted to one with a single seed tile and additional tiles that ``hard-code'' the initial row.)
        
        \item $\calT$ is directed.
        
        \item The frontier of any producible assembly in $\calT$ is never larger than 1, and thus $\calT$ has a single valid assembly sequence.

        \item There is never an exposed glue on the south of any tile of any producible assembly.

        \item Rows grow in alternating directions, i.e. if the first row grows from right to left (RtoL), then, after growing upward by one tile, the assembly grows left to right (LtoR) by 1 or more tiles. Once it stops growing LtoR and grows upward by one tile, it then grows RtoL. This growth pattern continues until the assembly becomes terminal (either finitely or in the limit).
    \end{enumerate}
\end{definition}

\begin{definition}[Standard TAS]\label{def:standard}
Let $\calT = (T, \sigma, 2)$ be a TAS in the aTAM. We say that $\calT$ is \emph{standard} if and only if:
    \begin{enumerate}
        \item $\calT$ is an IO TAS.
        
        \item $\calT$ is directed.

        \item For every $t \in T$, the sides that have input markings are either exactly (1) a single glue of strength-2, or (2) two diagonally adjacent strength-1 glues (i.e. not on opposite sides).

        \item There are no mismatches in the terminal assembly (i.e. all adjacent pairs of tile sides in $\alpha \in \termasm{\calT}$ have the same glue label and strength on both sides)
    \end{enumerate} 
\end{definition}

\begin{definition}[Standard TAS with across-the-gap]\label{def:standard-with-across-the-gap}
Let $\calT = (T, \sigma, 2)$ be a TAS in the aTAM. We say that $\calT$ is \emph{standard with across-the-gap} if and only if
    \begin{enumerate}
        \item $\calT$ is an IO TAS.

        \item $\calT$ is directed.

        \item For every $t \in T$, the sides that have input markings are either exactly (1) a single glue of strength-2, or (2) two strength-1 glues (i.e. on any combination of two sides).
        
        \item There are no mismatches in the terminal assembly (i.e. all adjacent pairs of tile sides in $\alpha \in \termasm{\calT}$ have the same glue label and strength on both sides.
    \end{enumerate}
\end{definition}




\subsection{Formal definitions for the abstract Slat Assembly Model}\label{sec:aSAM-appendix}

The abstract Slat Assembly Model (aSAM) is essentially a restricted version of the Polyomino Tile Assembly Model (polyTAM) \cite{Polyominoes}. The polyTAM itself is a generalization of the aTAM \cite{Winf98} in which, rather than square tiles, the basic components are polyominoes (which are shapes composed of unit squares attached along their edges).
(Note that we take much of our notation, slightly adapted, from aTAM definitions such as those in \cite{jSSADST}.)
The polyTAM was defined for two-dimensional polyominoes, but the aSAM utilizes three-dimensional polyominoes whose shapes are restricted to be linear arrangements of unit cubes.
The basic components of the aSAM are called \emph{slats} and each is a $1 \times 1 \times n$ polyomino composed of $n$ unit cubes, for some $n \in \mathbb{N}$. Each unit cube of a slat has $4$ or $5$ exposed faces: $5$ if it is on one of the two ends of the slat (in the dimension of length $n$), and $4$ otherwise (i.e. it is an interior cube). On each exposed face of a unit cube may be a \emph{glue}, and each glue is an ordered pair $(l,s)$ where $l$ is a string and serves as the glue's \emph{label} and $s \in \mathbb{Z}^+$ is its \emph{strength}. An exposed face may also have no glue (which may also be referred to as the \emph{null glue}). The character `$*$' is considered a special character in glue labels and any label may have at most a single `$*$' character, which must appear as its rightmost. Given a glue $g = (l,s)$, if the label $l$ does not end with the character `$*$', then we say the label $l' = l^*$ (i.e. the string $l$ concatenated with `$*$') represents the \emph{complement} of $l$. If $l$ does end with `$*$', then its complement is represented by the string $l$ truncated by one character to remove the `$*$'.
Thus a pair of labels are complementary to each other if they consist of the same string up to exactly one of them terminating in `$*$', e.g. `$foo$' and `$foo^*$' are complementary labels.
Any two glues that have complementary labels must have the same strength value.
If two slats are placed so that faces containing complementary glues are adjacent to each other, those glues \emph{bind} with strength equal to the common strength value of those two glues.

A \emph{slat type} is defined by its length $n$ and the set of glues on its constituent cubes. For convenience, each slat type is assigned a canonical placement and orientation in $\mathbb{Z}^3$, with the default being that it has one cube at $(0,0,0)$ and it extends along the $x$-axis to $(n-1,0,0)$, and the cubes have designated $N,E,S,W,U,D$ (i.e. north, east, south, west, up, down) sides which face in the $+y,+x,-y,-x,+z,-z$ directions, respectively (in the canonical placement). Additionally for convenience, the cubes are numbered from $0$ to $n-1$ starting from the cube at $(0,0,0)$ and proceeding in order along the $x$-axis (of the slat type in its canonical placement), and for a slat type $t$ the $i$th cube is denoted by $t[i]$. A slat is an instance of a slat type, and may be in any rotation or orientation in the $\mathbb{Z}^3$ lattice.
A slat is defined by (1) its type $t$, (2) its translation, which is identified by the coordinates of its $t[0]$, (3) its direction, taken from the set $\{-x,+x,-y,+y,-z,+z\}$ where the letter denotes which axis the length-$n$ dimension is parallel to and $+$ or $-$ denotes whether the coordinates of block $t[n-1]$ are more positive or more negative in that dimension than $t[0]$, respectively, and (4) its `up' direction which is the side of the cubes pointing in the $+z$ direction in the slat's current orientation (unless the slat is oriented along the $z$ axis, in which the `up' direction is the side of the cubes pointing in the $+x$ direction). See \Cref{fig:3D_slat_examples} for an example of a slat type in canonical placement and a rotated version.

\begin{figure}[!ht]
    \begin{subfigure}[t]{.47\linewidth}
        \includegraphics[width=2.5in]{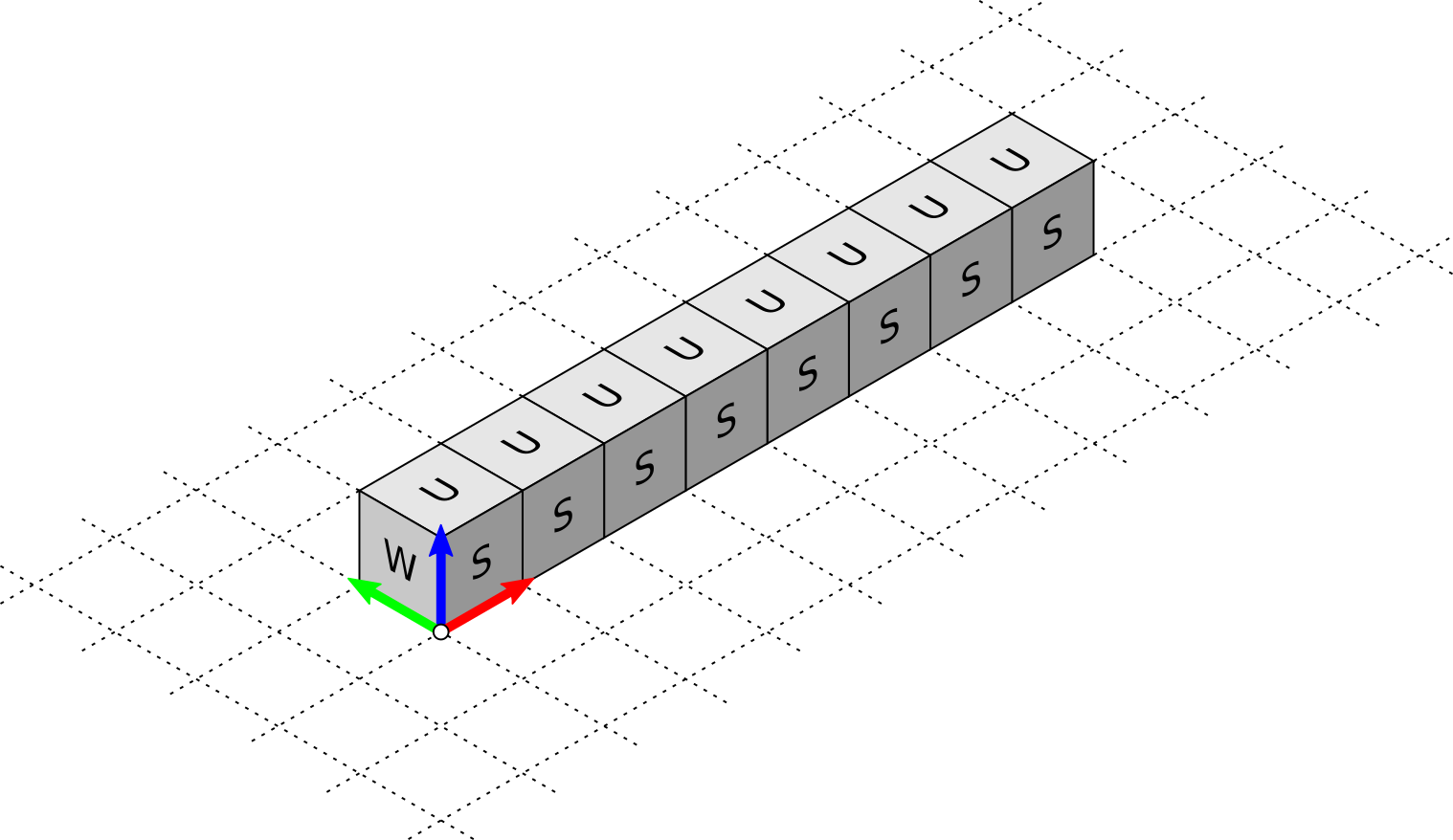}
        \caption{An example slat type $t$ of length 8 in its canonical placement.}
        \label{fig:3D_slat}
    \end{subfigure}
    \hspace{0.05\textwidth}
    \begin{subfigure}[t]{.47\linewidth}
        \includegraphics[width=2.5in]{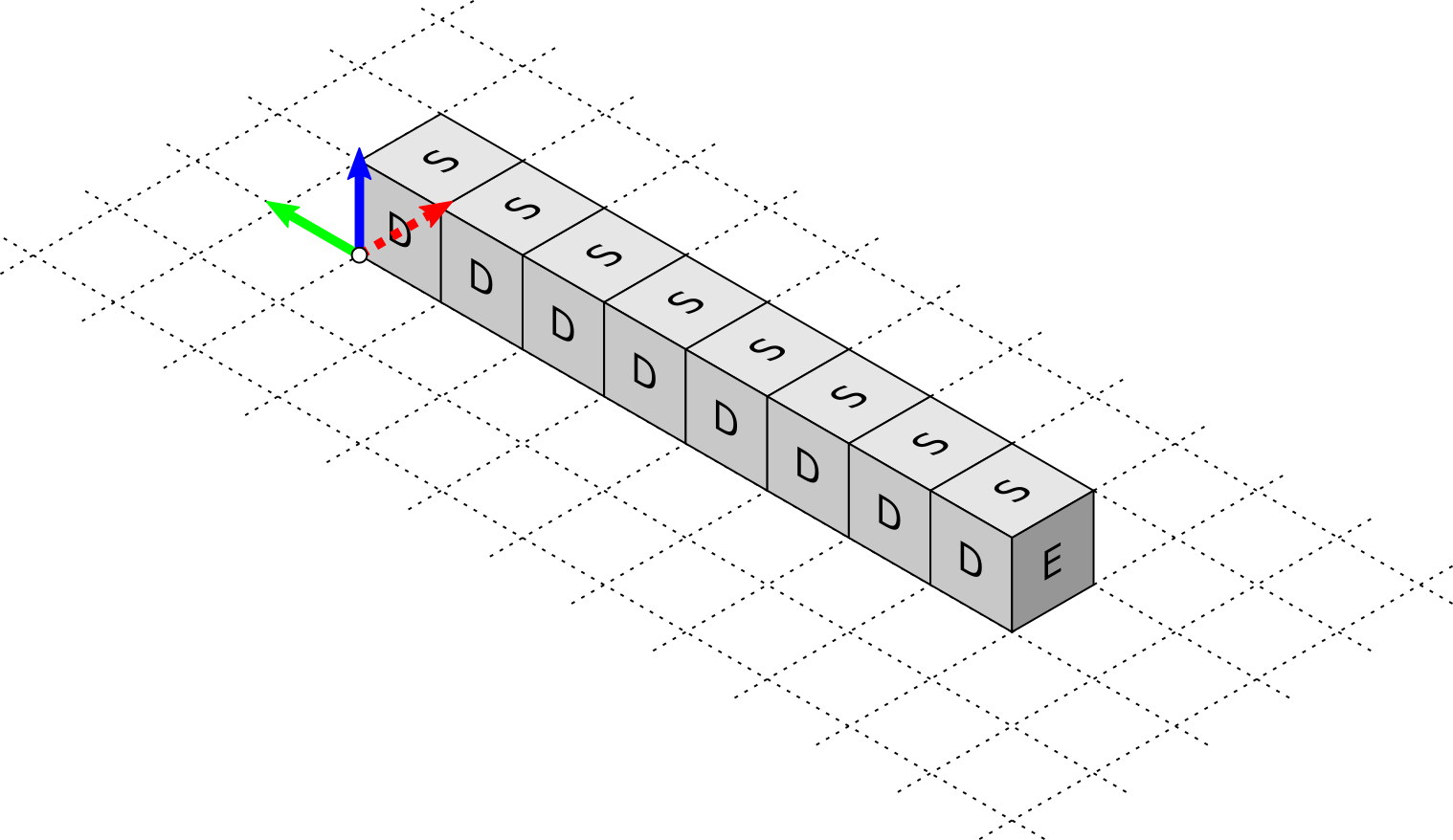}
        \caption{Slat $t$ rotated so that $t[0]$ is still at $(0,0,0)$ but its direction is $+y$ and its up direction is $-y$ (i.e. the `south' label appears on the top faces, in the $+z$ direction)}
        \label{fig:3D_slat_rotated}
    \end{subfigure}
    \caption{Depiction of slats in the aSAM.}
    \label{fig:3D_slat_examples}
\end{figure}

An \emph{assembly} over a set of slat types, $S$, consists of placements of slats of types from $S$ in $\mathbb{Z}^3$ such that no blocks of any two slats share the same space. Given an assembly $\alpha$, if two slats $t_1$ and $t_2$ are placed in $\alpha$ such that for some block of $t_1$, say $t_1[i]$, and some block of $t_2$, say $t_2[j]$, a face of $t_1[i]$ is adjacent to a face of $t_2[j]$ (irrespective of the directions of $t_i$ and $t_j$) and the glues of those faces are complementary, then they form a bond with the common strength value of those glues. If there is no cut in $\mathbb{Z}^3$ separating the slats of an assembly into two separate components without cutting bonds whose strengths sum to at least some value $\tau$, then we say the assembly is $\tau$-\emph{stable}. Given an assembly $\alpha$, a value $\tau \in \mathbb{Z}^+$, and a set of slat types $S$, any set $F$ of $i$ surfaces on blocks of the slats composing $\alpha$, where $0 < i \le \tau$, such that a slat of some type $t \in S$ can be positioned in $\mathbb{Z}^3$ (1) without any of its blocks overlapping any blocks of slats in $\alpha$ and (2) its glues form bonds of strength summing to $\ge \tau$ with the glues on the surfaces of $F$, we call $F$ a $S_\tau$-\emph{frontier} location of $\alpha$. Essentially, a $S_\tau$-frontier location of assembly $\alpha$ is a location to which a slat of some type in the set $S$ can validly attach with at least strength $\tau$. When $S$ and $\tau$ are clear from context we will simply refer to such locations as \emph{frontier} locations.

A \emph{slat assembly system}, or SAS, is an ordered triple $(S,\sigma,\tau)$ where $S$ is a set of slat types, $\sigma$ is an assembly of slats from $S$ referred to as the \emph{seed assembly}, and $\tau$ is the minimum binding threshold which is often referred to as the temperature in the aTAM and we call the \emph{cooperativity} in the aSAM. (Note that in \cite{ShihNucleation} they used the term \emph{coordination} for the same concept.) Given a SAS $\mathcal{S} = (S,\sigma,\tau)$, it is assumed that there are an infinite number of slats of each type from $S$ available for attachment, and assembly begins from the seed assembly $\sigma$. Assembly proceeds in discrete steps, with each step consisting of the nondeterministic selection of a frontier location $F$ for the current assembly $\alpha$, then the nondeterministic selection of a slat type $t \in S$ that can bind in $F$ (in case there are more than one), and then the attachment of a slat of type $t$ to $\alpha$, translated and rotated appropriately to bind to $\alpha$ using the glues of $F$. Note that the aSAM does not require that there be a path through which the slat of type $t$ must be able to move in $\mathbb{Z}^3$ from arbitrarily far from $\alpha$ into that location without encountering overlaps along the way (i.e. it can be considered to just ``appear'' in the correct location).\footnote{This aspect of the model is similar to that of other abstract models such as the aTAM, and is meant to make the model computationally tractable at the possible expense of reduced physical realism.} Given an assembly $\alpha$ in $\mathcal{S}$, if $\beta$ can result from $\alpha$ in a single such step, we say that $\alpha$ \emph{produces} $\beta$ \emph{in one step} and denote it as $\alpha \rightarrow^\mathcal{S}_1 \beta$. 

We use $\mathcal{A}^S$ to denote the set of all assemblies of slats over slat type set $S$. Given a SAS $\mathcal{S}=(S, \sigma, \tau)$, a sequence of $k\in\mathbb{Z}^+ \cup \{\infty\}$ assemblies $\alpha_0, \alpha_1, \ldots, \alpha_k$ over $\mathcal{A}^S$ is called a \emph{$\mathcal{S}$-assembly sequence} if $\alpha_0 = \sigma$ and, for all $1\le i < k$, $\alpha_{i-1} \to^{\mathcal{S}}_1 \alpha_i$. The \emph{result} of an assembly sequence is the unique limiting assembly of the sequence. For finite assembly sequences, this is the final assembly; whereas for infinite assembly sequences, this is the assembly consisting of all slats from any assembly in the sequence. We say that \emph{$\alpha$ $\mathcal{S}$-produces $\beta$} (denoted $\alpha\to^{\mathcal{S}} \beta$) if there is a $\mathcal{S}$-assembly sequence starting with $\alpha$ whose result is $\beta$. We say $\alpha$ is \emph{$\mathcal{S}$-producible} if $\sigma\to^{\mathcal{S}}\alpha$ and write $\prodasm{\mathcal{S}}$ to denote the set of $\mathcal{S}$-producible assemblies. We say $\alpha$ is \emph{$\mathcal{S}$-terminal} if $\alpha$ is $\tau$-stable and there exists no assembly which is $\mathcal{S}$-producible from $\alpha$. We denote the set of $\mathcal{S}$-producible and $\mathcal{S}$-terminal assemblies by $\termasm{\mathcal{S}}$. When $\mathcal{S}$ is clear from context, we may omit $\mathcal{S}$ from this notation. If there is only a single terminal assembly of $\mathcal{S}$, i.e. $|\termasm{\mathcal{S}}| = 1$, then we say that $\mathcal{S}$ is \emph{directed}. Otherwise, it is \emph{undirected}. Note that a SAS $\mathcal{S}$ may have multiple (even infinitely many) distinct valid assembly sequences but yet $\mathcal{S}$ may be directed.

\subsection{A restricted version of the aSAM}\label{sec:aSAM-minus-appendix}

The aSAM is defined to be a relatively general model for the assembly of slat-based structures. However, the previous experimental work of \cite{ShihNucleation,Shih-OrigamiSlats} and the abstract designs and simulation results of this paper all pertain to a restricted subset of the aSAM that we'll refer to as the \emph{Restricted aSAM} and denote as the aSAM$^-$.

The aSAM$^-$ utilizes the following restrictions:

\begin{enumerate}
    \item All blocks of all slats contained in an assembly are restricted to the two planes $z=0$ and $z=1$.
    
    \item All slats which attach in the vertical orientation (i.e. their longest dimension is parallel to the $y$-axis) do so in the plane $z=0$.
            
    \item All slats which attach in the horizontal orientation (i.e. their longest dimension is parallel to the $x$-axis) do so in the plane $z=1$.

    \item No slats attach in an orientation such that their longest dimension is parallel to the $z$-axis.

    \item For a given slat type $t$, all slats of type $t$ that attach to an assembly will do so in the same orientation. Slat types whose slats always bind in a vertical orientation are referred to as \emph{vertical slats}, and those in horizontal orientation as \emph{horizontal slats}.

    \item All glues of a polyomino are on a single side, $U$ for vertical slats and $D$ for horizontal slats.

    \item All glues have strength $= 1$.
        
\end{enumerate}

Essentially, the aSAM$^-$ requires that all slats remain in one of two planes, with the horizontal slats all being on the top plane and the vertical slats all being on the bottom, with the only glues being strength-1 glues between vertical and horizontal slats (i.e. no horizontal slat can bind to another horizontal slat, and no vertical to another vertical), and therefore for any slat to attach to an assembly, it must bind to $\tau$ different slats already in the assembly.

\subsection{Definition of simulation of an aTAM system by an aSAM system}\label{sec:aSAM-simulation-formal}

This section defines intrinsic simulation of an aTAM system by an aSAM system formally, expanding upon Section~\ref{sec:aSAM-simulation-informal}.

We now define what it means for an aSAM system to simulate an aTAM system. This definition is analogous to the definition of intrinsic simulation between aTAM systems (e.g. \cite{IUSA,DirectedNotIU,DDDIU}).
To that end, for the purposes of this section, it will be assumed that $\calS=(S,\sigma_\calS, \tau_\calS)$ is some SAS that simulates the TAS $\calT=(T, \sigma_\calT, \tau_\calT)$.

For each such simulation we fix a positive integer $m$ called the \emph{scale factor} with the intention that $m \times m$ blocks of locations from $\calS$ map to individual locations in $\calT$. In other words, simulation happens at scale so that by ``zooming out'' by a factor of $m$, the assemblies in $\calS$ resemble corresponding assemblies in $\calT$.

In the context of $\calS$, we call each $m\times m$ block of locations in $\mathbb{Z}^2$ a \emph{macrotile location} under the convention that the origin $(0,0)\in\mathbb{Z}^2$ occupies the south-westernmost location in one of the $m\times m$ blocks. That is, macrotiles locations are squares with southwest corners of the form $(cx, cy)$ and northeast corners of the form $(m(x+1)-1, m(y+1)-1)$ where $x,y\in\mathbb{Z}$. Given an assembly $\alpha' \in \prodasm{\calS}$, we use the notation $\mathcal{M}^c_{x,y}[\alpha']$ to refer to the set of slats in $\alpha'$ which occupy locations in the macrotile location whose southwest corner is $(cx, cy)$. This set of slats is referred to the \emph{macrotile} located at $(x,y)$. Note that in this context, we keep track of the (portions of) slats that occupy a macrotile with coordinates relative to $(x,y)$ so that two macrotiles which differ only by a translation are identical. We denote the set of all possible macrotiles of scale factor $m$ made from slats in $S$ as $\mathcal{M}^m[S]$.

A partial function $R:\mathcal{M}^m[S] \to T$ is called a \emph{macrotile representation function} from $S$ to $T$ if, for any macrotiles $\alpha, \beta \in \mathcal{M}^m[S]$ where $\alpha \sqsubseteq \beta$ and $\alpha \in \text{dom} R$, then $R(\alpha) = R(\beta)$. In other words, a macrotile representation function may not map a macrotile to an individual tile type in $T$, but if it does, then any additional slat attachments do not change how the macrotile is mapped under $R$. In the context of some $\calS$-assembly sequence, a macrotile is said to \emph{resolve} to a tile type $t\in T$ when an assembly maps under $R$ to $t$, but the prior assembly is not in the domain of $R$.

From a macrotile representation function $R$, a function $R^*:\mathcal{A}^S \to \mathcal{A}^T$, called the \emph{assembly representation function}, is induced which maps entire assemblies in $\calS$ to assemblies in $\calT$. This function is defined by applying the function $R$ to each macrotile location containing slats. We also use the notation ${R^*}^{-1}(\alpha)$ to refer to the \emph{producible pre-image} of an assembly $\alpha \in \mathcal{A}^T$. That is, ${R^*}^{-1}(\alpha) = \{\alpha' \in \prodasm{\calS} | R^*(\alpha') = \alpha \}$ so that ${R^*}^{-1}(\alpha)$ includes every $\calS$-producible assembly mapping to $\alpha$ under $R^*$. We say that an assembly $\alpha'\in \mathcal{A}^S$ maps cleanly to an assembly $\alpha\in \mathcal{A}^T$ under $R$ if $R(\alpha') = \alpha$ and no slats in $\alpha'$ occupy any macrotile whose location is not adjacent (not including diagonally) to a resolved macrotile. Formally, if $\alpha'$ maps cleanly to $\alpha$, then for each non-empty macrotile $\mathcal{M}^m_{x,y}[\alpha']$, there exists some vector $(u,v)\in\mathbb{Z}^2$ such that $\lVert (u,v) \rVert \le 1$ so that $\mathcal{M}^m_{x+u,y+v}[\alpha'] \in \text{dom} R$.\footnote{Note that this definition requires that the seed assembly $\sigma_\mathcal{S}$ resolves to at least one tile in $\calT$. This will suffice for the constructions of this paper, but if needed this definition can be relaxed to specifically allow for the seed assembly to grow before resolving.}
The definition of \emph{maps cleanly} requires that any non-empty but unresolved macrotile has a resolved macrotile to its N, E, S. or W, thus ensuring that the growth of slats is only occurring in macrotile locations that map to locations in $\calT$ adjacent to tiles, and thus with some potential to receive a tile.

\begin{definition}[Equivalent Productions] \label{def:sim-equiv-prod}
    We say that $\calS$ and $\calT$ have \emph{equivalent productions} (under $R$), written $\calS \Leftrightarrow_R \calT$, if the following conditions hold:
    \begin{enumerate}
        \item $\{R^*(\alpha') | \alpha' \in \prodasm{\calS}\} = \prodasm{\calT}$,
        \item $\{R^*(\alpha') | \alpha' \in \termasm{\calS}\} = \termasm{\calT}$, and
        \item For all $\alpha'\in\prodasm{\calS}$, $\alpha'$ maps cleanly to $R^*(\alpha')$.
    \end{enumerate}
\end{definition}

\begin{definition}[Follows] \label{def:sim-follows}
    We say that $\calT$ \emph{follows} $\calS$ (under $R$), written $\calT \dashv_R \calS$, if for all $\alpha', \beta' \in \prodasm{\calS}$, $\alpha' \to^\calS \beta'$ implies $R^*(\alpha') \to^\calT R^*(\beta')$.
\end{definition}

\begin{definition}[Models] \label{def:sim-models}
    We say that $\calS$ \emph{models} $\calT$ (under $R$), written $\calS \vDash_R \calT$, if for every $\alpha \in \prodasm{\calT}$, there exists a non-empty subset $\Pi_\alpha \subseteq {R^*}^{-1}(\alpha)$, such that for all $\beta \in \prodasm{\calT}$ where $\alpha \to^\calT \beta$, the following conditions are satisfied:
    \begin{enumerate}
        \item for every $\alpha' \in \Pi_\alpha$, there exists $\beta' \in {R^*}^{-1}(\beta)$ such that $\alpha' \to^\calS \beta'$
        \item for every $\alpha''\in {R^*}^{-1}(\alpha)$ and $\beta''\in {R^*}^{-1}(\beta)$ where $\alpha'' \to^\calS \beta''$, there exists $\alpha'\in \Pi_\alpha$ such that $\alpha' \to^\calS \alpha''$.
    \end{enumerate}
\end{definition}

In Definition~\ref{def:sim-models} above, the set $\Pi_\alpha$ is defined to be a set of assemblies representing $\alpha$ from which it is still possible to produce assemblies representing all possible $\beta$ producible from $\alpha$. Informally, the first condition specifies that all assemblies in $\Pi_\alpha$ can produce some assembly representing any $\beta$ producible from $\alpha$, while the second condition specifies that any assembly $\alpha''$ representing $\alpha$ that may produce an assembly representing $\beta$ is producible from an assembly in $\Pi_\alpha$. In this way, $\Pi_\alpha$ represents a set of the earliest possible representations of $\alpha$ where no commitment has yet been made regarding the next simulated assembly. Requiring the existence of such a set $\Pi_\alpha$ for every producible $\alpha$ ensures that non-determinism is faithfully simulated. That is, the simulation cannot simply ``decide in advance'' which tile attachments will occur.

\begin{definition}[Simulates]\label{def:simulate}
    Given an aTAM system $\calT = (T, \sigma, 2)$, an aSAM system $\mathcal{S} = (S, \sigma_\mathcal{S}, c)$, and a macrotile representation function $R$ from $\mathcal{S}$ to $\calT$, we say that $\mathcal{S}$ \emph{intrinsically simulates} $\calT$ (under $R$)  if $\mathcal{S} \Leftrightarrow_R \mathcal{T}$ (they have equivalent productions), $\mathcal{T} \dashv_R \mathcal{S}$ and $\mathcal{S} \models_R \mathcal{T}$ (they have equivalent dynamics).
\end{definition}



\subsection{Macrotile glue naming conventions}\label{sec:glue-naming}

In this section, we present details of the conventions used to generate glue labels for the slats of macrotiles. In order to ensure that a slat of a given type can only bind in the desired location and also with the desired orientation, translation, and rotation, during the instantiation of a macrotile template to create actual slat types, we carefully design unique glue labels.

A main component of the label created for a glue is based on its cell's location within a macrotile (called its \emph{macrotile coordinates}), as well as its location within a cell (called its \emph{cell coordinates}). To understand these aspects of the naming conventions, refer to Figure \ref{fig:glue-naming}. In locations where a horizontal slat of a macrotile crosses over a vertical slat of the same macrotile, and the intention is for there to be a glue binding them together, a glue (called an \emph{interior glue}) is designed for that location. The interior glue label will be a string consisting of a concatenation of (1) the name of the tile used to instantiate the macrotile template creating the slat, (2) the macrotile coordinates of the location, (3) the cell coordinates of the location, and (4) the ``*'' (star) symbol if the glue is to be placed on a horizontal slat, otherwise that is omitted. Such a pair of glue labels is thus guaranteed to be complementary and unique in the system (among all other macrotiles due to the tile name, and among the same macrotile due to the coordinates).

When a macrotile template is instantiated for a tile of type $t$, other than the interior glues needed to bind the horizontal and vertical slats of that macrotile together, glues must also be created to simulate the behaviors of the input and output glues of $t$, allowing assembled macrotiles to initiate the growth of adjacent macrotiles. Let $g$ be the label of an input glue of $t$ (e.g. '$\wedge$a' for an input glue on the S side of $t$ with label `a'). The glue labels created for the domains of these glues are strings consisting of a concatenation of (1) $g$, (2) the cell coordinates, (3) possibly a \emph{marker} symbol used to disambiguate between different versions of $g$ (discussed more below), and (4) the ``*'' (star) symbol if the glue is to be placed on a horizontal slat, otherwise that is omitted. The same naming scheme is used for the output glues of $t$ (which will cause them to be the same as the input glues to which they should bind except for the ``*'' marking that will cause them to be complementary and thus able to bind). In this way, the representation of the input and output glues are agnostic to the macrotile of which they are a part (since their labels don't contain the name of $t$ or the macrotile coordinates, as interior glues do), just as the glues of tiles are agnostic to the particular tile on which they reside and serve as generic binding interfaces. However, they are specific to their locations within a cell, allowing for correct alignment of attaching slats.

In some of our more complex constructions, the designs of macrotile templates can vary in the placements (i.e. macrotile coordinates) and/or orientations of slats representing input and output glues. Whenever it is possible for multiple macrotiles, say $m_1$ and $m_2$, to be created such that both must represent an input (respectively, output) glue that has the same label in $T$ (which we'll call $g$), but $m_1$ and $m_2$ place the domains representing $g$ in cells with different relative macrotile coordinates and/or on slats with opposite orientations (vertical versus horizontal), then it is necessary for $m_1$ and $m_2$ to have different marker symbols included in their glue labels and, of course, for the output (respectively, input) glues of another macrotile that are meant to bind to each to have the same markers. In this way, it is possible to represent the same input or output glue in multiple ways, which is sometimes necessary due to geometric constraints imposed by growth ordering and compact scale factors, but to ensure that the slats binding to them are specifically designed to account for the particular geometric offset or orientation of each representation. 

\begin{figure}[!ht]
    \centering
    \begin{subfigure}[t]{.4\linewidth}
        \includegraphics[width=1.0\textwidth]{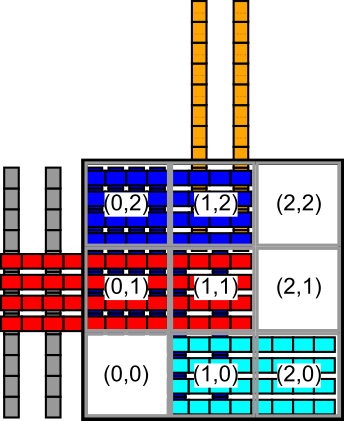}
        \caption{ \label{fig:macrotile-coords}}
    \end{subfigure}
    \hspace{0.1\textwidth}
    \begin{subfigure}[t]{.25\linewidth}
        \includegraphics[width=1.0\textwidth]{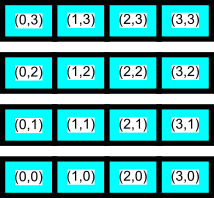}
        \caption{\label{fig:cell-coords}}
    \end{subfigure}
    \caption{This figure shows the slats of an example macrotile for $c = 4$ before glue assignments. (a) The larger, black square encloses the $3c \times 3c$ region defined for the body of the macrotile, with output slats extending outside of that. Each $c \times c$ cell of the macrotile is enclosed by a grey square, and the macrotile coordinates (i.e. the locations of the cells with respect to the southwest corner of the macrotile) are shown for each. (b) A zoomed in display of cell (2,0), showing the cell coordinates (i.e. the coordinates of each potential glue domain location with respect to the southwest corner of the cell).
    \label{fig:glue-naming}}
\end{figure}

\subsection{Technical details for the simulation of the class of standard aTAM systems}\label{sec:standard-append}

In this section we provide the full technical details for the proof of Theorem \ref{thm:standard}. This section expands on the overview provided in Section~\ref{sec:standard}.

\begin{proof}
    We prove Theorem \ref{thm:standard} by construction, and thus, starting with an arbitrary standard aTAM system $\calT = (T, \sigma, 2)$ and given any $c > 2$ such that $c \mod 2 = 0$ we show how to create aSAM system $\mathcal{S} = (S, \sigma', c)$ and macrotile representation function $R$ such that $\mathcal{S}$ simulates $\calT$ under $R$.

    \begin{figure}
    \centering
    \begin{subfigure}{0.15\textwidth}
        \centering
        \includegraphics[width=1.0\textwidth]{images/zig-zag-tile.png}
        \caption{\label{fig:standard-3x3-tile}}
    \end{subfigure}
    \hspace{20pt}
    \begin{subfigure}{0.2\textwidth}
        \centering
        \includegraphics[width=1.0\textwidth]{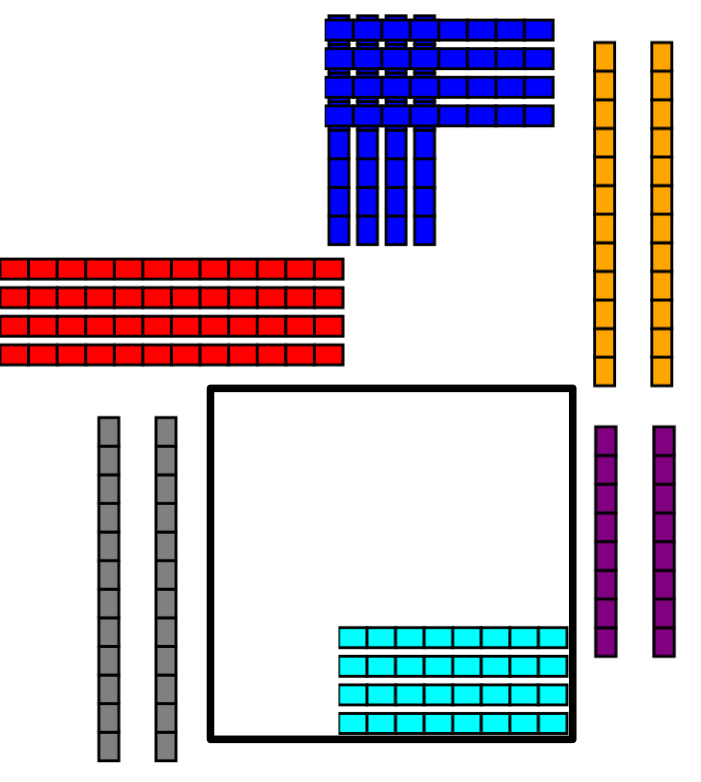}
        \caption{\label{fig:standard_3x3_marked_macrotile}}
    \end{subfigure}
    \hspace{20pt}
    \begin{subfigure}{0.2\textwidth}
        \centering
        \includegraphics[width=1.0\textwidth]{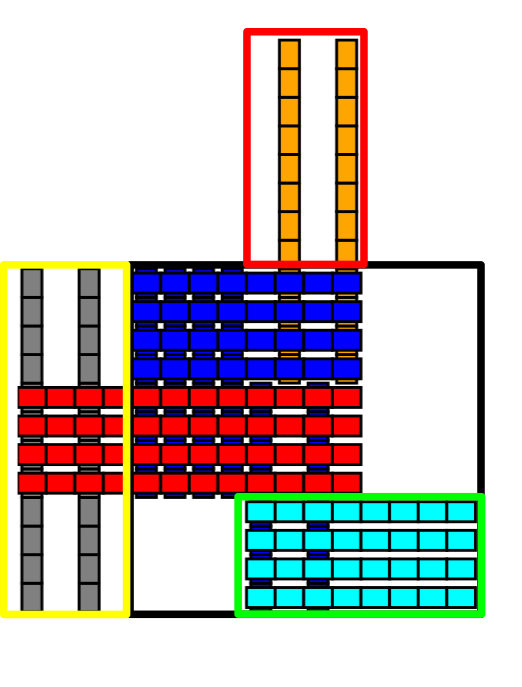}
        \caption{\label{fig:standard_3x3_macrotile_parts}}
    \end{subfigure}
    \hspace{20pt}
    \begin{subfigure}{0.2\textwidth}
        \centering
        \includegraphics[width=1.0\textwidth]{images/cell-color-conventions.png}
        \caption{\label{fig:standard_3x3_color_conventions}}
    \end{subfigure}
    \caption{(a) An IO-marked tile type \textit{t} from a standard aTAM tile set. It has strength-1 inputs on the south and east, and strength-1 outputs on the north and west. (b) A set of 7 slat groups for the macrotile simulating $t$ at $c = 4$. The light blue group contains body slats that are entirely within the square of the macrotile which they cause to resolve to $t$. The purple group contains two body slats that are placed within the same $c \times c $ region as two output slats from a neighboring macrotile. The red group contains four output slats which bind to both the purple body slats, and those output slats received from a neighboring macrotile. The red and grey groups combined are the output slats that serve as a west output and extend into the western neighboring macrotile location. The dark blue group contains body slats which serve to construct a region in which the orange output slats may bind, constructing a strength-1 output to the north. (c) An example of the assembled 3c × 3c macrotile for $t$, with cells marked to show the portions of $t$ that they represent, following the conventions of (d). (d) A cell enclosed in a green square represents the cell in which the initial body slats of a macrotile bind, causing it to resolve to $t$. The cells enclosed in red, gold, light blue, and yellow squares denote the cells in which the slats expose glues representing the output glues in the north, east, south, and west directions, respectively. \label{fig:standard_3x3_cell_coloring}}
    \end{figure}

    First, we are given that $\calT$ is a standard aTAM system, meaning it possesses the set of characteristics that are defined in Definition \ref{def:standard}. That is to say, $\calT$ is directed, an IO TAS, and for every $t \in T$, the sides that have input markings are either a single side with a strength-2 glue, or two diagonally adjacent sides each with a strength-1 glue. Additionally, there are no mismatches in the terminal assembly (i.e. all adjacent pairs of sides of pairs of tiles in $\alpha \in \termasm{\calT}$ have the same glue label and strength).

    \begin{figure}
    \centering
    \includegraphics[width=0.9\textwidth]{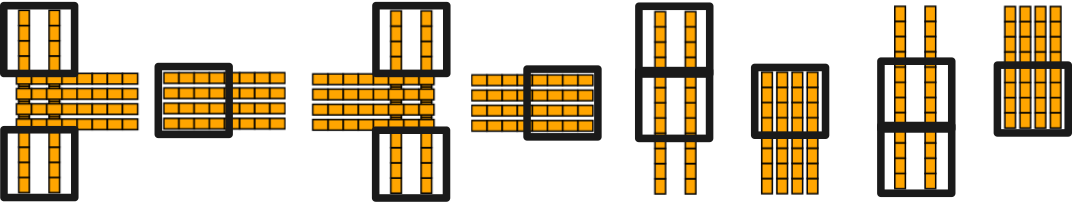}
    \caption{Output slat templates for standard systems, with output glue locations marked with a black box. Ordered west strength-1, strength-2, east strength-1, strength-2, north strength-1, strength-2, south strength-1, strength-2 respectively.}
    \label{fig:standard_3x3_output_templates}
    \end{figure}

    \begin{figure}
    \centering
    \includegraphics[width=0.9\textwidth]{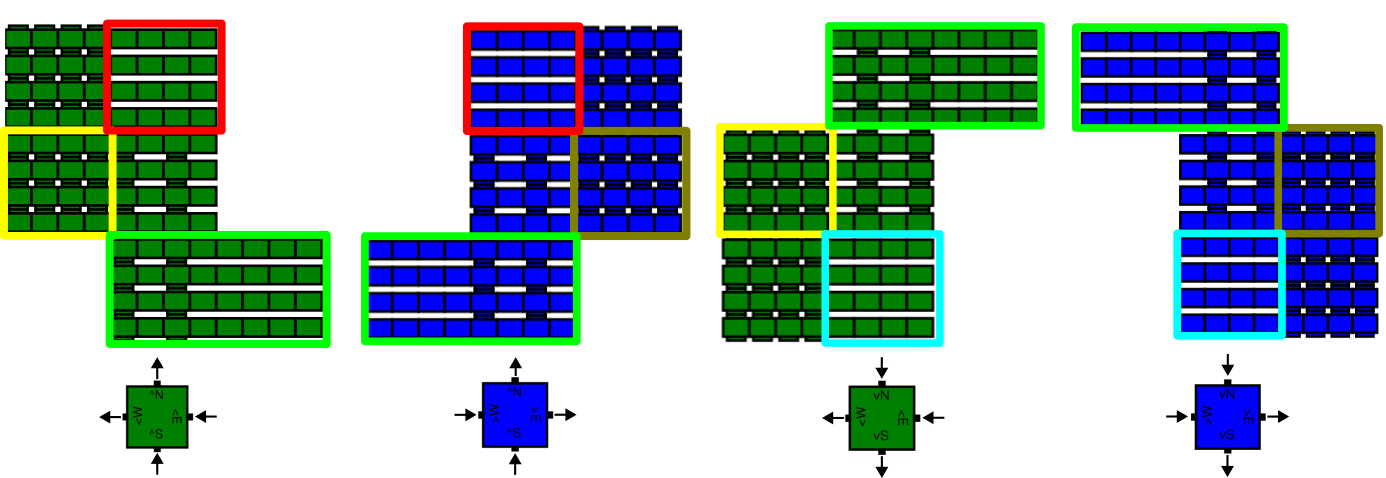}
    \caption{The slat-based macrotile templates for tiles within a standard system which experience adjacent inputs. Cells are bounded by squares to show their functionality, and correspond to an output slat template which may connect to the associated cell. Cells are signified using the same color conventions as Figure \ref{fig:standard_3x3_color_conventions}.}
    \label{fig:standard_3x3_adjacent}
    \end{figure}

    Each tile in $\calT$ is simulated by a macrotile of size $3c \times 3c$ in $\mathcal{S}$. For $c = 4$, this means that the $12 \times 12$ square whose southwest coordinate is $(12i, 12j)$, for every $i, j \in \mathbb{Z}$, will map under $R$ to either empty space or to a tile in $\calT$. An example is shown in Figure \ref{fig:standard_3x3_cell_coloring}. Each slat in a macrotile is of a unique type.\footnote{Using techniques of \cite{SlatAcceleration}, it is possible to reuse slat types within macrotiles to reduce the slat complexity, but for ease of explanation we present our constructions without that optimization.}  We use the same definitions for \emph{slat group}, \emph{body slats}, and \emph{output slats} as are defined in Section 3.1.
    
    \begin{figure}
    \centering
    \includegraphics[width=0.4\textwidth]{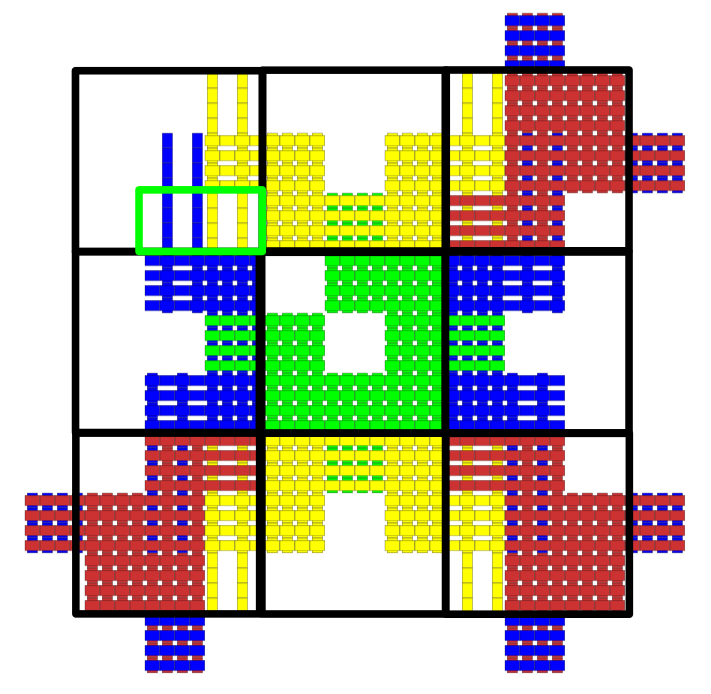}
    \caption{An example of a portion of an assembly composed of (partial) $3c \times 3c$ macrotiles simulating a standard aTAM system for $c = 4$. Nine of the macrotile locations are outlined in black squares. The macrotile simulating $t$ from Figure \ref{fig:standard_3x3_cell_coloring} would attach into the top left such macrotile location. Its (light blue) body slats would attach to the 2 west output slats from the macrotile to the east (yellow) and the 2 north output slats from the macrotile to the south (dark blue). These light blue slats can bind in any order, and as soon as one binds, the macrotile resolves to $t$. Only once they have all 4 bound can the east output slats (red) bind. Only once all 4 of those have bound can the dark blue body slats bind, then finally the 2 orange output slats. Thus, the growth of a macrotile is well-ordered, and outputs are only presented after a macrotile resolves, enforcing the restrictions of simulation.}
    \label{fig:standard_3x3_construction}
    \end{figure}

    Let $t_n \in T$, for $0 \le n < |T|$, be the $n$th tile in tile set $T$. We will refer to the string ``$t_n$'' as the unique name of $t_n$. Given the directions of growth, and the dynamics of a standard aTAM system, the following list contains all possible valid signatures for any $t_n$ of a standard system. 

    \begin{enumerate}
    \item Strength 2 Input Tiles (Figure \ref{fig:standard_3x3_S2}):
        \begin{enumerate}
            \item Seed tile: Input=$\emptyset$, Output=(N,$\{0,1,2\}$),(E,$\{0,1,2\}$),(S,$\{0,1,2\}$),(W,$\{0,1,2\}$) 
            \item North tile: Input=(N,2), Output=(E,$\{0,1,2\}$),(S,$\{0,1,2\}$),(W,$\{0,1,2\}$) 
            \item East tile: Input=(E,2), Output=(N,$\{0,1,2\}$),(S,$\{0,1,2\}$),(W,$\{0,1,2\}$) 
            \item South tile: Input=(S,2), Output=(N,$\{0,1,2\}$),(E,$\{0,1,2\}$),(W,$\{0,1,2\}$)
            \item West tile: Input=(W,2), Output=(N,$\{0,1,2\}$),(E,$\{0,1,2\}$),(S,$\{0,1,2\}$) 
        \end{enumerate}
    \item Adjacent Input Tiles (Figure \ref{fig:standard_3x3_adjacent}):
        \begin{enumerate}
            \item North-East tile: Input=(N,1),(E,1) Output=(S,$\{0,1,2\}$),(W,$\{0,1,2\}$)
            \item North-West tile: Input=(N,1),(W,1), Output=(E,$\{0,1,2\}$),(S,$\{0,1,2\}$)
            \item South-East tile: Input=(S,1),(E,1), Output=(N,$\{0,1,2\}$),(W,$\{0,1,2\}$)
            \item South-West tile: Input=(S,1),(W,1), Output=(N,$\{0,1,2\}$),(E,$\{0,1,2\}$)
        \end{enumerate}
    \end{enumerate}

    When $t_n$ is simulated using a standard macrotile, its output slats may only assemble into a small, fixed set of configurations to allow for cooperation with macrotiles which will simulate tiles in the $\calT$-frontier. Each item in this set of configurations will be referred to as an \emph{output slat template}, and each valid output slat template for standard systems is shown in Figure \ref{fig:standard_3x3_output_templates}.
    
    Figure \ref{fig:standard_3x3_S2} shows a macrotile template for any strength-2 input tile signatures, and Figure \ref{fig:standard_3x3_adjacent} shows tiles with adjacent input tile signatures and their corresponding macrotile templates, both of these with output locations marked to signify the ability to include an output slat template. The longest slat-length to be exhibited by a macrotile is of length $3c$, which is shown in Figure \ref{fig:standard_3x3_examples}. To build $\mathcal{S}$ for each $t_n$, we use the same methods as are described in section {\ref{sec:zig-zag}}. As with the macrotile representation function $R$ in the proof of Theorem \ref{thm:zig-zag}, $R$ can simply inspect each macrotile location (each of size $3c \times 3c$) and resolve to the correct corresponding tile as soon as a body slat appears, determining the tile of $\calT$ to which it should map by the prefix of its name.  Furthermore, following the same arguments as for the proof of Theorem \ref{thm:zig-zag}, the construction succeeds in building $\mathcal{S}$ so that it simulates $\calT$ under $R$ by proper and consistent alignment of input and output regions and correct assignment of glues. (An example of a partial assembly can be seen in Figure \ref{fig:standard_3x3_construction}.)
    Therefore, $\mathcal{S}$ simulates $\calT$, an arbitrary standard aTAM system, under $R$ using cooperativity $c$ and macrotiles of size $3c \times 3c$ and maximum slat length $3c$, and Theorem \ref{thm:standard} is proven.
\end{proof}


\subsection{Technical details for the simulation of the class of standard with across-the-gap aTAM systems}\label{sec:standard-across-the-gap-append}

In this section we provide the full technical details for the proof of Theorem~\ref{thm:standard+across-the-gap} from Section~\ref{sec:standard-across-the-gap}.

\begin{proof}
    We prove Theorem \ref{thm:standard+across-the-gap} by construction. Thus, starting with an arbitrary standard aTAM system with across-the-gap $\calT = (T, \sigma, 2)$ and given any $c > 2$ such that $c \mod 2 = 0$ we show how to create aSAM system $\mathcal{S} = (S, \sigma', c)$ and macrotile representation function $R$ such that $\mathcal{S}$ simulates $\calT$ under $R$.
    
    First, we are given that $\calT$ is a standard aTAM system with across-the-gap, meaning that it possesses the set of characteristics that are defined in Definition \ref{def:standard-with-across-the-gap}. That is to say, $\calT$ is directed, an IO TAS, and for every $t \in T$, the sides that have input markings are either a single side with a strength-2 glue, or two sides each with a strength-1 glue. Additionally, there are no mismatches in the terminal assembly (i.e. all adjacent pairs of sides of pairs of tiles in $\alpha \in \termasm{\calT}$ have the same glue label and strength).

    \begin{figure}
    \centering
    \begin{subfigure}{0.1\textwidth}
        \centering
        \includegraphics[width=1.0\textwidth]{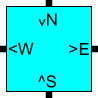}
        \caption{\label{fig:standard+ATG_tile}}
    \end{subfigure}
    \hspace{20pt}
    \begin{subfigure}{0.3\textwidth}
        \centering
        \includegraphics[width=1.0\textwidth]{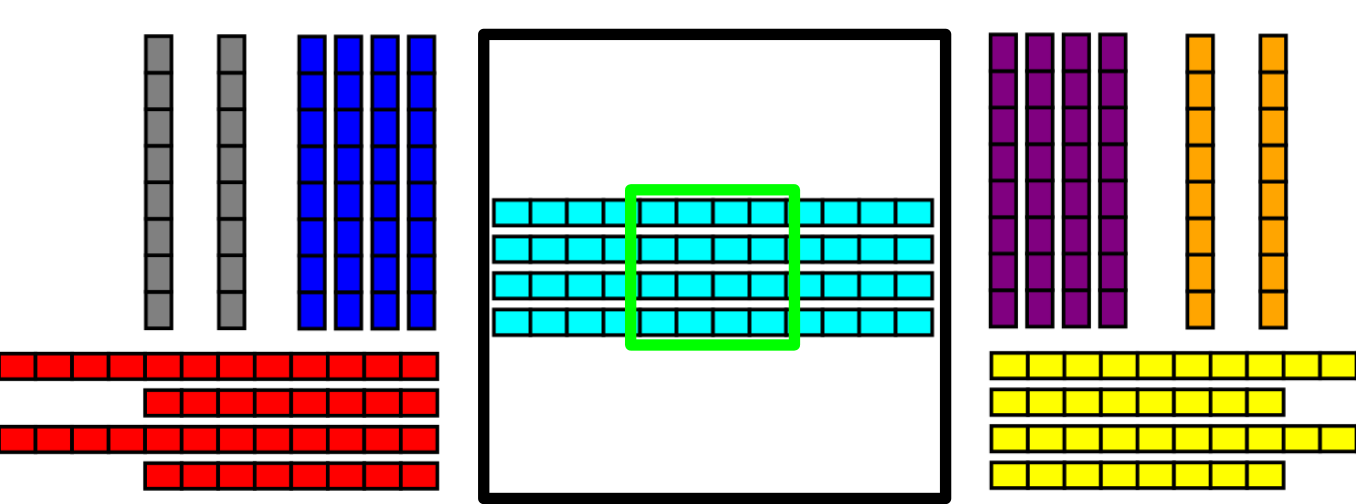}
        \caption{\label{fig:standard_3x3_with_across-the-gap_marked_macrotile}}
    \end{subfigure}
    \hspace{20pt}
    \begin{subfigure}{0.25\textwidth}
        \centering
        \includegraphics[width=1.0\textwidth]{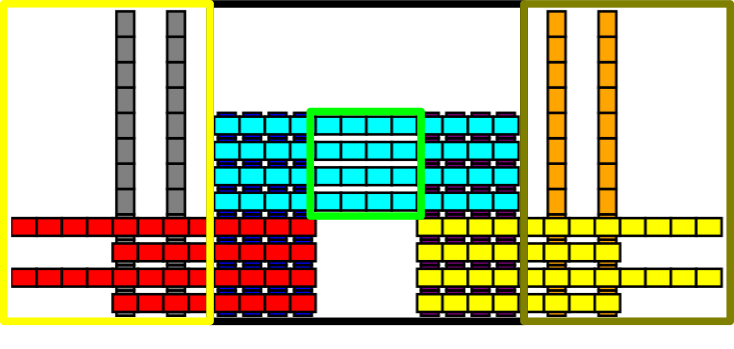}
        \caption{\label{fig:standard_3x3_with_across-the-gap_macrotile_parts}}
    \end{subfigure}
    \hspace{20pt}
    \begin{subfigure}{0.15\textwidth}
        \centering
        \includegraphics[width=1.0\textwidth]{images/cell-color-conventions.png}
        \caption{\label{fig:standard_3x3+ATG_color_conventions}}
    \end{subfigure}
    \caption{(a) An IO-marked tile type \textit{t} from a standard aTAM tile set with across-the-gap. It has strength-1 inputs on the north and south, and strength-1 outputs on the east and west. (b) A set of 7 slat groups for the macrotile simulating $t$ at $c = 4$. The light blue group contains body slats that are entirely within the square of the macrotile which they cause to resolve to $t$. The purple group contains body slats which serve to construct a region in which the yellow output slats may bind. The yellow group contains four output slats, two of which have their length extended by $c$ in order to allow across-the-gap cooperation. The yellow group and the orange group combined are the output slats that serve as an east output and extend into the eastern neighboring macrotile location. The dark blue group contains body slats which serve to construct a region in which the red group of output slats may bind. Similar to the yellow group, the red group contains four output slats, two of which have their length extended by $c$ in order to exhibit across-the-gap cooperation. The red and grey groups combined are the output slats that serve as a west output and extend into the western neighboring macrotile location.
    (c) An example of the assembled 3c × 3c macrotile for $t$, with cells marked to show the portions of $t$ that they represent, following the conventions of (d). (d) A cell enclosed in a green square represents the cell in which the initial body slats of a macrotile bind, causing it to resolve to $t$. The cells enclosed in red, gold, light blue, and yellow squares denote the cells in which the slats expose glues representing the output glues in the N, E, S, and W directions, respectively. \label{fig:standard_3x3+ATG_cell_coloring}}
    \end{figure}
    

    \begin{figure}
    \centering
    \includegraphics[width=0.9\textwidth]{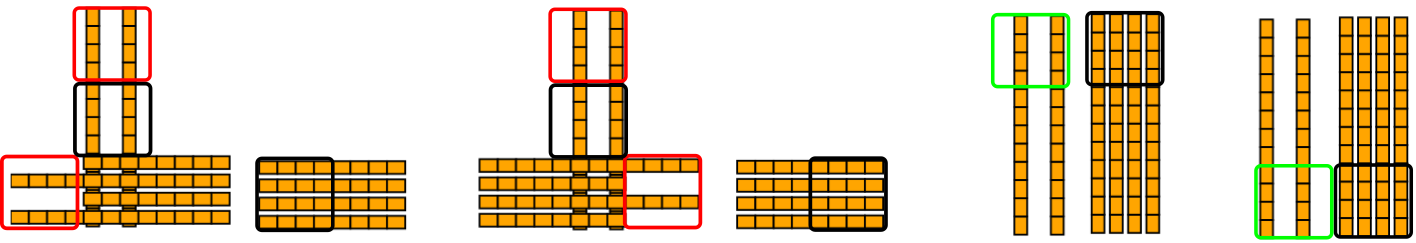}
    \caption{Output slat templates for standard aTAM systems with across-the-gap. Output glue locations which are used for adjacent, or strength-2 output are marked with a black box, ones which are used for only across-the-gap output are marked with a red box, and ones which are used for both adjacent and across-the-gap output are marked with a green box. Ordered west strength-1, strength-2, rast strength-1, strength-2, north strength-1, strength-2, south strength-1, strength-2 respectively.}
    \label{fig:standard_3x3+ATG_output_templates}
    \end{figure}
    
    \begin{figure}
    \centering
    \begin{subfigure}{0.2\textwidth}
        \centering
        \includegraphics[width=1.0\textwidth]{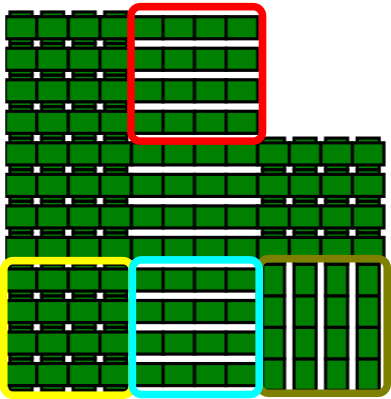}
        \caption{\label{fig:standard_3x3+ATG_S2_macrotile}}
    \end{subfigure}
    \hspace{20pt}
    \begin{subfigure}{0.2\textwidth}
        \centering
        \includegraphics[width=1.0\textwidth]{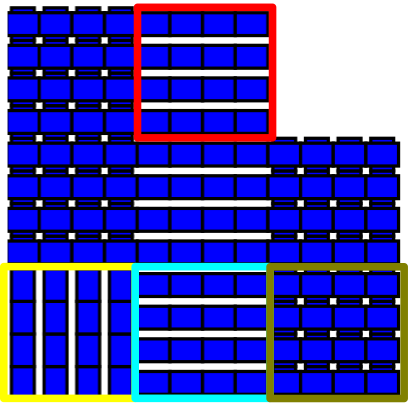}
        \caption{\label{fig:standard_3x3+ATG_S2_macrotile_2}}
    \end{subfigure}
    \caption{Strength-2 macrotile templates for standard aTAM systems with across-the-gap. Cells are bounded by squares to show their functionality, and mark cell locations where output slat templates may be added to the macrotile. One of the marked cells may instead be designated as an input, and have its domains assigned such that they connect with those provided by the output slats of a neighboring macrotile whose output is of the same glue type. Cells are signified using the same color conventions as Figure \ref{fig:standard_3x3+ATG_color_conventions}. If a macrotile accepts west input, the macrotile shown in (b) is chosen as to not block output to the south, otherwise, the macrotile shown in (a) is chosen.}.\label{fig:standard_3x3+ATG_S2}
    \end{figure}

    \begin{figure}
    \centering
    \includegraphics[width=0.9\textwidth]{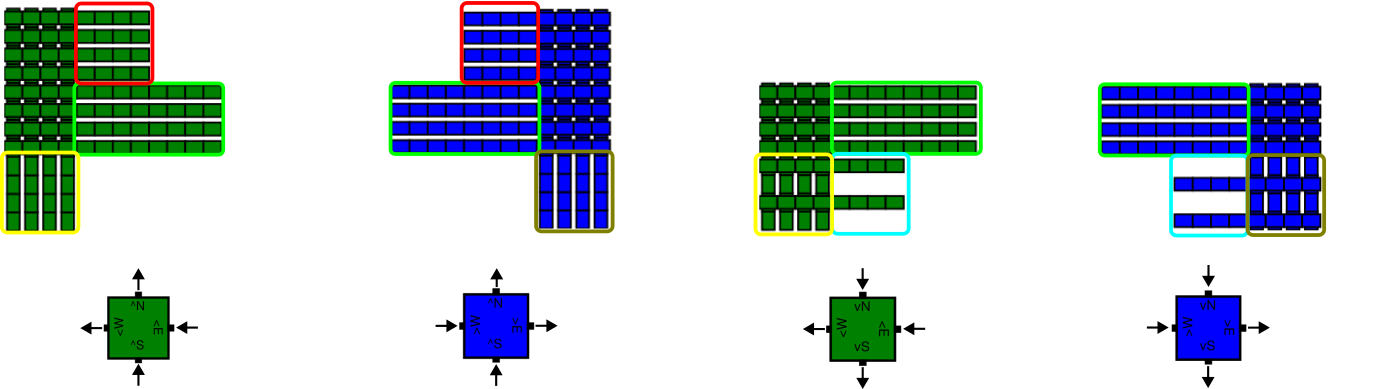}
    \caption{The slat-based macrotile templates for tiles within a standard aTAM system with across-the-gap which experience adjacent input. Cells are bounded by squares to show their functionality, and correspond to an output slat template which may connect to the associated cell. Cells are signified using the same color conventions as Figure \ref{fig:standard_3x3+ATG_color_conventions}. Macrotiles which exhibit north input construct output slats to the south through the cooperation with the abutting across-the-gap output domains from the neighboring macrotile which provided input to the east, or west respectively.}
    \label{fig:standard_3x3+ATG_adjacent}
    \end{figure}

    \begin{figure}
    \centering
    \includegraphics[width=0.4\textwidth]{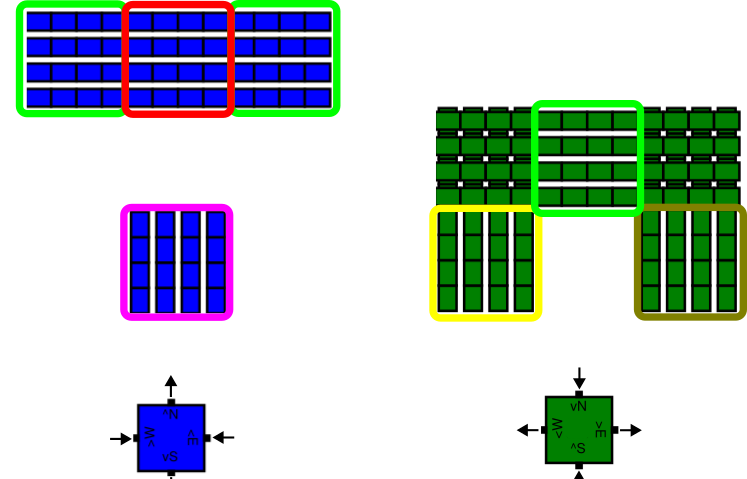}
    \caption{The slat-based macrotile templates for tiles within a standard aTAM system with across-the-gap which experience across-the-gap inputs. Cells are bounded by squares to show their functionality, and correspond to an output slat template which may connect to the associated cell. Cells are signified using the same color conventions as Figure \ref{fig:standard_3x3+ATG_color_conventions}. The macrotile template which accepts strength-1 input from the east and west has their southern slat group marked with a pink square in order to signify that it accepts input from both the east and west neighboring macrotiles, as well as be able to connect to a corresponding output template to the south.}
    \label{fig:standard_3x3+ATG_ATG}
    \end{figure}
    
    \begin{figure}
    \centering
    \begin{subfigure}{0.24\textwidth}
        \centering
        \includegraphics[width=1.0\textwidth]{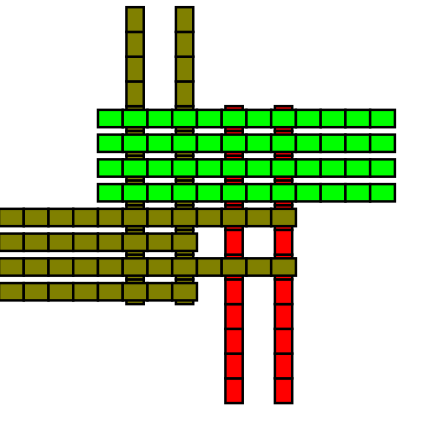}
        \caption{\label{fig:standard+ATG_tile_ES_cooperation}}
    \end{subfigure}
    \hspace{20pt}
    \begin{subfigure}{0.3\textwidth}
        \centering
        \includegraphics[width=1.0\textwidth]{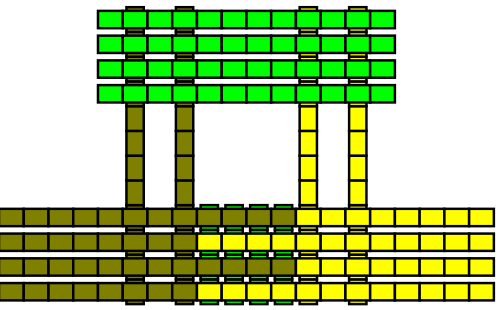}
        \caption{\label{fig:standard_3x3+ATG_EW_cooperation}}
    \end{subfigure}
    \hspace{20pt}
    \begin{subfigure}{0.0712\textwidth}
        \centering
        \includegraphics[width=1.0\textwidth]{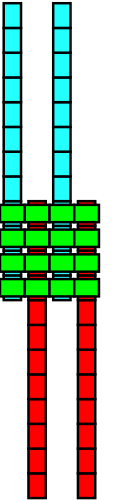}
        \caption{\label{fig:standard_3x3+ATG_NS_cooperation}}
    \end{subfigure}
    \caption{Examples of cooperation in the construction for standard aTAM systems with across-the-gap, showing only the output slats from neighbors and body slats of the resolving macrotile. (a) The east and north output slats allow body slats to attach and resolve macrotile $t_n$ through adjacent cooperation. (b) The east and west output slats allow body slats to attach and resolve macrotile $t_n$ through across-the-gap cooperation. (c) The north and south output slats allow body slats (truncated in the figure) to attach and resolve macrotile $t_n$ through across-the-gap cooperation. Slat groups are colored in accordance with the color conventions in Figure \ref{fig:standard_3x3+ATG_color_conventions}.\label{fig:standard_3x3+ATG_output_cooperation}}
    \end{figure}

    \begin{figure}
    \centering
    \begin{subfigure}{0.30\textwidth}
        \centering
        \includegraphics[width=1.0\textwidth]{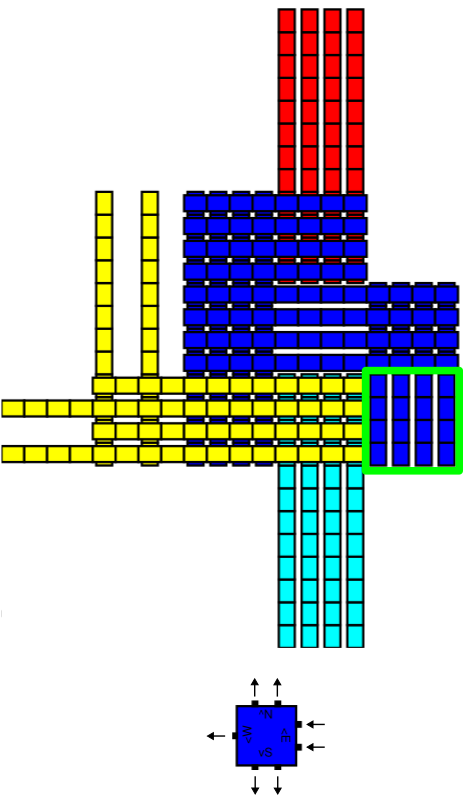}
        \caption{\label{fig:standard_3x3+ATG_S2_example}}
    \end{subfigure}
    \hspace{20pt}
    \begin{subfigure}{0.19\textwidth}
        \centering
        \includegraphics[width=1.0\textwidth]{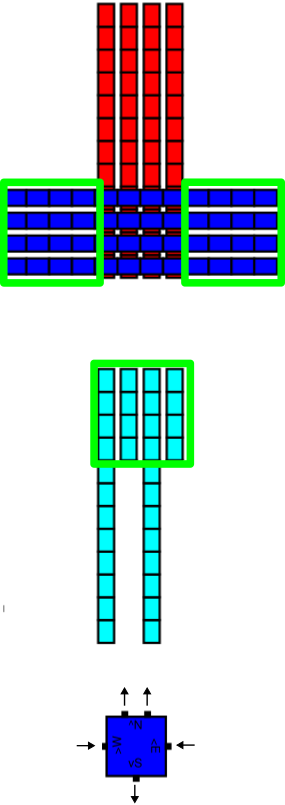}
        \caption{\label{fig:standard_3x3+ATG_ATG_example}}
    \end{subfigure}
    \hspace{20pt}
    \caption{Example macrotiles for a standard aTAM system with across the gap. (a) A macrotile which takes an east strength-2 input, and provides a north strength-2 output, south strength-2 output, and east strength-1 output. (b) A macrotile which takes east and west strength-1 inputs, and provides a north strength-2 output, and south strength-1 output. Despite the disconnection between each slat groups, both groups represent the same tile $t$, and the placement of any slat from either group resolves the macrotile to $t$. (The output slats allowing for this binding are shown in Figure \ref{fig:standard_3x3+ATG_EW_cooperation}.) Output slat templates are colored in accordance with Figure \ref{fig:standard_3x3+ATG_color_conventions}, and cells which may accept input are marked with a green box. \label{fig:standard_3x3_cell_coloring2}}
    \end{figure}
    
    Each tile in $\calT$ is simulated by a macrotile of size $3c \times 3c$ in $\mathcal{S}$. For $c = 4$, this means that the $12 \times 12$ square whose southwest coordinate is $(12i, 12j)$, for every $i, j \in \mathbb{Z}$, will map under $R$ to either empty space or a tile in $\calT$. An example is shown in Figure \ref{fig:standard_3x3+ATG_cell_coloring}. Each slat in a macrotile is of a unique type.\footnote{Using techniques of \cite{SlatAcceleration}, it is possible to reuse slat types within macrotiles to reduce the slat complexity, but for ease of explanation we present our constructions without that optimization.} We use the same definitions for \emph{slat group}, \emph{output slat group}, \emph{body slats}, and \emph{output slats} as are defined in Sections \ref{sec:zig-zag} and \ref{sec:standard}.
    
    \begin{figure}
    \centering
    \includegraphics[width=0.7\textwidth]{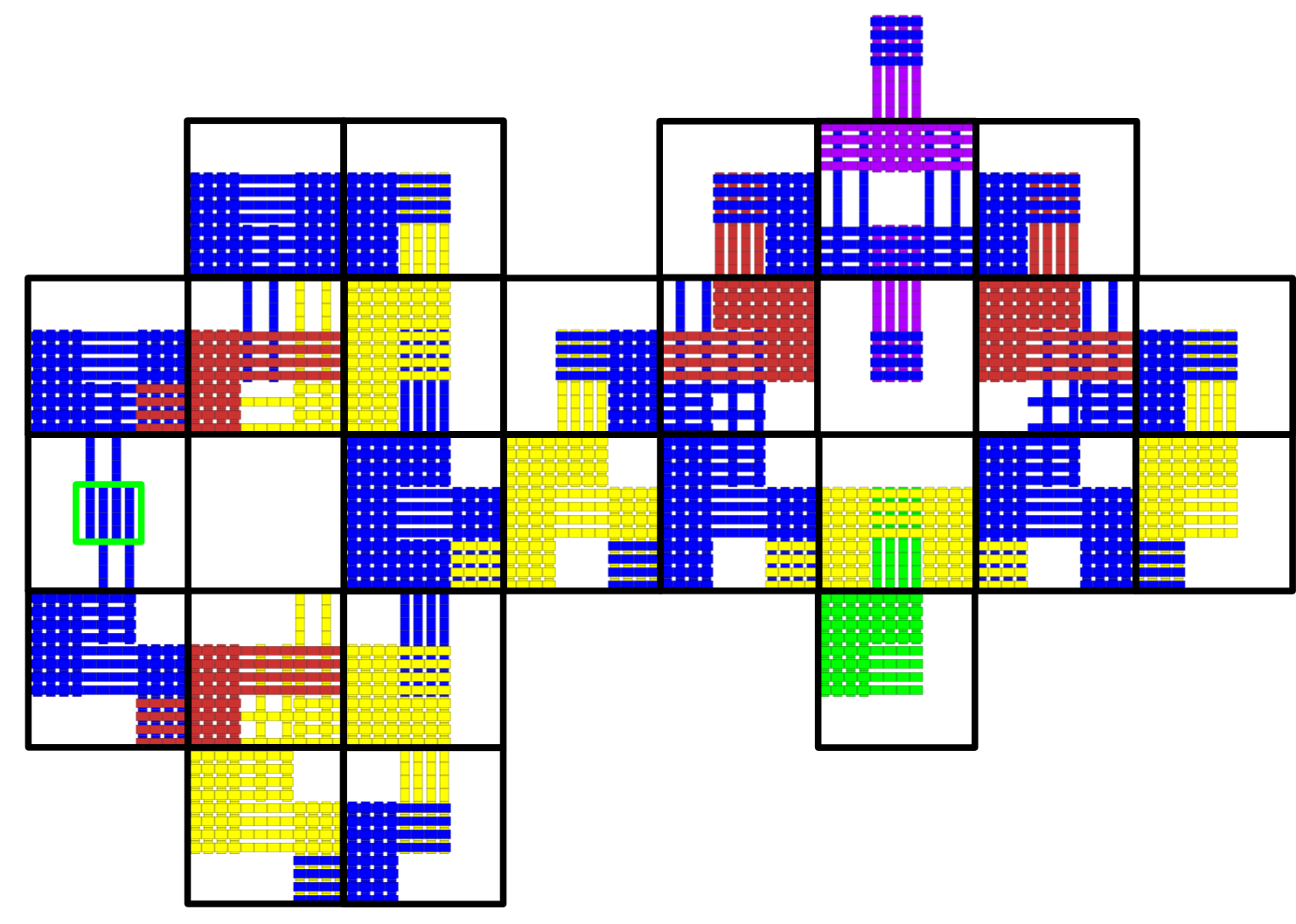}
    \caption{An example of an assembly composed of (partial) $3c \times 3c$ macrotiles simulating a standard aTAM system for $c = 4$. Twenty-six of the macrotile locations are outlined in black squares. The macrotile simulating $t$ from Figure \ref{fig:standard_3x3+ATG_cell_coloring} would attach into the middle of the leftmost of such macrotile locations, containing a green square. Its (light blue) body slats would attach to the 2 south output slats from the macrotile to the north (dark blue) and the 2 north output slats from the macrotile to the south (dark blue). These light blue slats can bind in any order, and as soon as one binds, the macrotile resolves to $t$. Only once they have all 4 bound can the purple or blue body slats bind. Only once all 4 of either of those have bound can the red or yellow output slats bind, then finally the grey or orange output slats. Thus, the growth of a macrotile is well-ordered, and outputs are only presented after a macrotile resolves, enforcing the restrictions of simulation.}
    \label{fig:standard_3x3+ATG_construction}
    \end{figure}
    
    Let $t_n \in T$, for $0 \le n < |T|$, be the $n$th tile in tile set $T$. We will refer to the string ``$t_n$'' as the unique name of $t_n$. Given the directions of growth, and the dynamics of a standard with across-the-gap aTAM system, the following list contains all possible valid signatures for any $t_n$ of a standard with across-the-gap system.

    \begin{enumerate}
    \item Strength 2 Input Tiles (Figure \ref{fig:standard_3x3+ATG_S2}):
        \begin{enumerate}
            \item Seed tile: Input=$\emptyset$, Output=(N,$\{0,1,2\}$),(E,$\{0,1,2\}$),(S,$\{0,1,2\}$),(W,$\{0,1,2\}$) 
            \item North tile: Input=(N,2), Output=(E,$\{0,1,2\}$),(S,$\{0,1,2\}$),(W,$\{0,1,2\}$) 
            \item East tile: Input=(E,2), Output=(N,$\{0,1,2\}$),(S,$\{0,1,2\}$),(W,$\{0,1,2\}$) 
            \item South tile: Input=(S,2), Output=(N,$\{0,1,2\}$),(E,$\{0,1,2\}$),(W,$\{0,1,2\}$)
            \item West tile: Input=(W,2), Output=(N,$\{0,1,2\}$),(E,$\{0,1,2\}$),(S,$\{0,1,2\}$) 
        \end{enumerate}
    \item Adjacent Input Tiles (Figure \ref{fig:standard_3x3+ATG_adjacent}):
        \begin{enumerate}
            \item North-East tile: Input=(N,1),(E,1) Output=(S,$\{0,1,2\}$),(W,$\{0,1,2\}$)
            \item North-West tile: Input=(N,1),(W,1), Output=(E,$\{0,1,2\}$),(S,$\{0,1,2\}$)
            \item South-East tile: Input=(S,1),(E,1), Output=(N,$\{0,1,2\}$),(W,$\{0,1,2\}$)
            \item South-West tile: Input=(S,1),(W,1), Output=(N,$\{0,1,2\}$),(E,$\{0,1,2\}$)
        \end{enumerate}
    \item Across-the-gap Input Tiles (Figure \ref{fig:standard_3x3+ATG_ATG}):
        \begin{enumerate}
            \item North-South tile: Input=(N,1),(S,1) Output=(E,$\{0,1,2\}$),(W,$\{0,1,2\}$)
            \item East-West tile: Input=(E,1),(W,1) Output=(N,$\{0,1,2\}$),(S,$\{0,1,2\}$)
        \end{enumerate}
    \end{enumerate}

    When $t_n$ is simulated using a standard with across-the-gap macrotile, its output slats may only assemble into a small, fixed set of configurations to allow for cooperation with macrotiles which will simulate in the $\calT$-frontier. Each valid output slat template is shown in Figure \ref{fig:standard_3x3+ATG_output_templates}. Figure \ref{fig:standard_3x3+ATG_S2} shows a macrotile template for any strength-2 input tile signatures, and Figures \ref{fig:standard_3x3+ATG_adjacent} and \ref{fig:standard_3x3+ATG_ATG} show tiles with adjacent, and across-the-gap input tile signatures and their corresponding macrotile templates respectively.
    Examples showing how cooperation is handled can be seen in Figure \ref{fig:standard_3x3+ATG_output_cooperation} and examples of instantiated macrotiles can be seen in Figure \ref{fig:standard_3x3_cell_coloring2}.
    Figure \ref{fig:standard_3x3+ATG_construction} shows a portion of an example assembly of a system $\mathcal{S}$.
    The longest slat-length to be exhibited by a macrotile is of length $4c$, which is shown in Figure \ref{fig:standard_3x3+ATG_S2_example}.
    All of said figures have output locations marked to signify the ability to include an output slat template. To build $\mathcal{S}$ for each $t_n$, we use the same methods as are described in Section \ref{sec:zig-zag}. As with the macrotile representation $R$ in the proof of Theorem \ref{thm:zig-zag}, $R$ can simply inspect each macrotile location (each of size $3c \times 3c$) and resolve to the correct corresponding tile as soon as a body slat appears, determining the tile of $\calT$ to which it should map by the prefix of its name. Furthermore, following the same arguments as in the proof of Theorem \ref{thm:zig-zag}, the construction succeeds in building $\mathcal{S}$ so that it simulates $\calT$ under $R$ by proper and consistent alignment of input and output regions and correct assignment of glues. Therefore, $\mathcal{S}$ simulates $\calT$, an arbitrary standard with across-the-gap aTAM system, under $R$ using cooperativity $c$ and macrotiles of size $3c \times 3c$ and maximum slat length $4c$, and Theorem \ref{thm:standard+across-the-gap} is proven.

\end{proof}


\subsection{Technical details for the simulation of the class of temperature-2 directed systems}\label{sec:determ-with-mismatches-append}

In this section we provide the full technical details for the proof of Theorem~\ref{thm:deterministic+mismatches}. This section expands on the overview provided in Section~\ref{sec:determ-with-mismatches}.

\begin{proof}

    We prove Theorem~\ref{thm:deterministic+mismatches} by construction, and thus, start with an arbitrary directed temperature-2 aTAM system $\calT = (T, \sigma, 2)$. Given a cooperativity $c > 2$ such that $c$ mod $2 = 0$ we show how to create an aSAM system $\mathcal{S} = (S, \sigma', c)$ and macrotile representation function $R$ such that $\mathcal{S}$ simulates $\calT$ under $R$. First, we are given that $\calT$ is directed, an IO-TAS (or we apply the simple transformation discussed in Section \ref{sec:IO-marking} to make it so), and for every $t \in \calT$, the sides that have input markings are either a single side with a strength-2 glue, or two sides each with a strength-1 glue. Additionally, it may be the case that adjacent tiles have mismatching glues, but being a directed system, only 1 terminal assembly is possible and thus tiles of only one type can bind into any location.

    Each tile in $\calT$ is simulated by a macrotile of size $4c \times 4c$ in $\mathcal{S}$. For $c = 4$, this means that the $16 \times 16$ square whose southwest coordinate is $(16i, 16j)$, for every $i, j \in \mathbb{Z}$, will map under $R$ to either empty space or a tile in $\calT$. An example is shown in Figure \ref{fig:deterministic_4x4_cell_coloring}. Each slat in a macrotile is of a unique type.\footnote{Using techniques of \cite{SlatAcceleration}, it is possible to reuse slat types within macrotiles to reduce the slat complexity, but for ease of explanation we present our constructions without that optimization.} We use the same definitions for \emph{slat group}, \emph{output slat group}, \emph{body slats}, and \emph{output slats} as are defined in Sections \ref{sec:zig-zag} and \ref{sec:standard}.

    \begin{figure}
    \centering
    \begin{subfigure}{0.10\textwidth}
        \centering
        \includegraphics[width=1.0\textwidth]{images/zig-zag-tile.png}
        \caption{\label{fig:deterministic_tile}}
    \end{subfigure}
    \hspace{12pt}
    \begin{subfigure}{0.4\textwidth}
        \centering
        \includegraphics[width=1.0\textwidth]{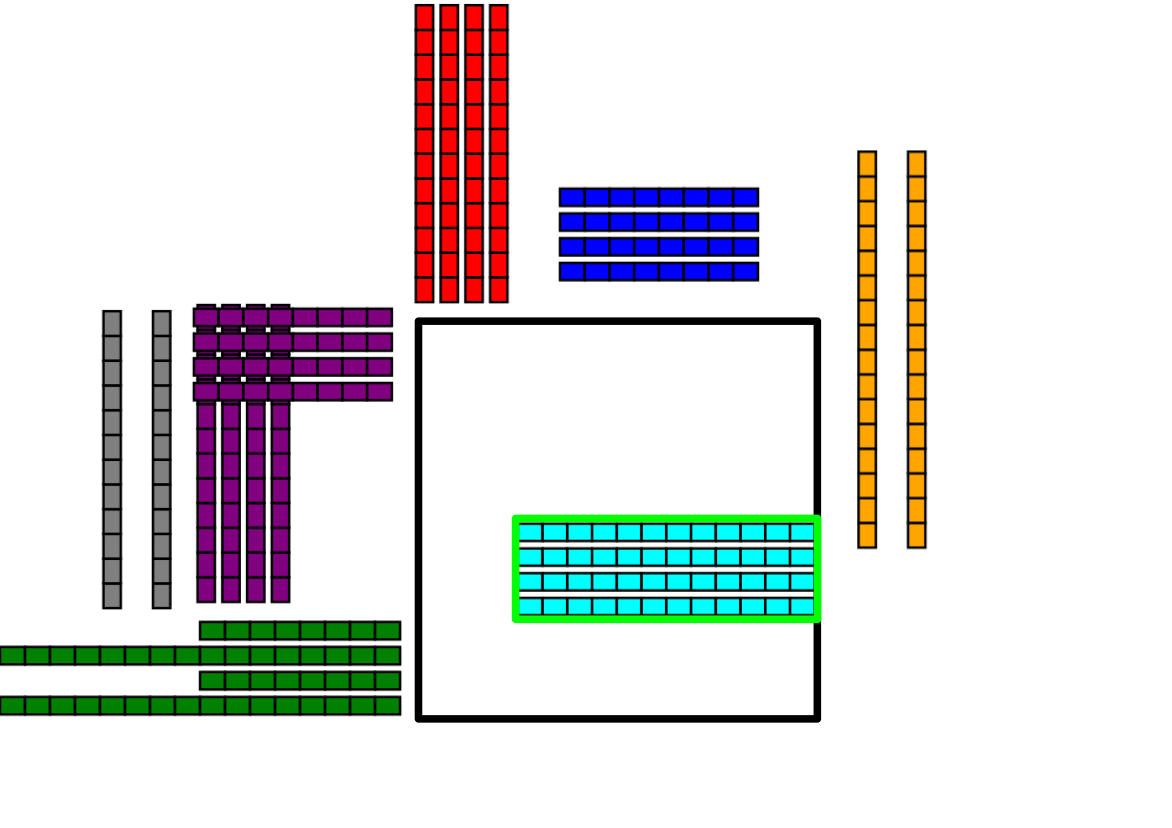}
        \caption{\label{fig:deterministic_4x4_macrotile_parts}}
    \end{subfigure}
    \hspace{-12pt}
    \begin{subfigure}{0.25\textwidth}
        \centering
        \includegraphics[width=1.0\textwidth]{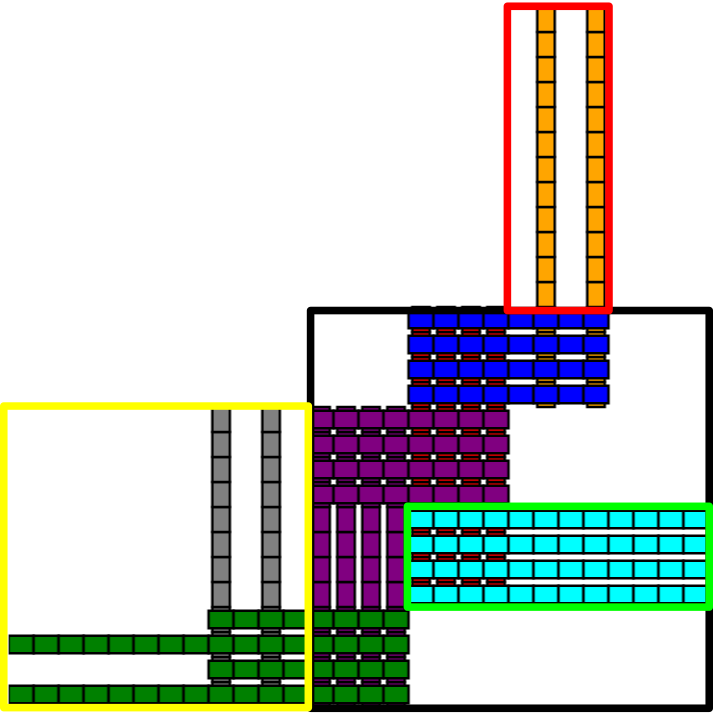}
        \caption{\label{fig:marked_deterministic_4x4_macrotile}}
    \end{subfigure}
    \hspace{12pt}
    \begin{subfigure}{0.15\textwidth}
        \centering
        \includegraphics[width=1.0\textwidth]{images/cell-color-conventions.png}
        \caption{\label{fig:deterministic_4x4_color_conventions}}
    \end{subfigure}
    \caption{(a) An IO-marked tile type \textit{t} from a directed temperature-2 aTAM system. It has strength-1 inputs on the east and south, and strength-1 outputs on the north and west. (b) A set of 8 slat groups for the macrotile simulating $t$ at $c = 4$. The light blue group contains body slats that are entirely within the square of the macrotile which they cause to resolve to $t$. The red group contains body slats which serve to construct a region in which the purple and blue body slat groups may bind. The blue group contains body slats that serve to construct a region where the orange output slats may bind, which extend into the northern neighboring macrotile to exhibit north output. Similarly, the 2 purple groups serve to construct a region where the green output slats may bind.
    The green and the grey groups combined are the output slats that serve as an west output and extend into the eastern neighboring macrotile location. 
    (c) An example of the assembled 4c × 4c macrotile for $t$, with cells marked to show the portions of $t$ that they represent, following the conventions of (d). (d) A cell enclosed in a green square represents the cell in which the initial body slats of a macrotile bind, causing it to resolve to $t$. The cells enclosed in red, gold, light blue, and yellow squares denote the cells in which the slats expose glues representing the output glues in the north, east, south, and west directions, respectively. \label{fig:deterministic_4x4_cell_coloring}}
    \end{figure}
    
    \begin{figure}
        \centering
        \includegraphics[width=0.9\textwidth]{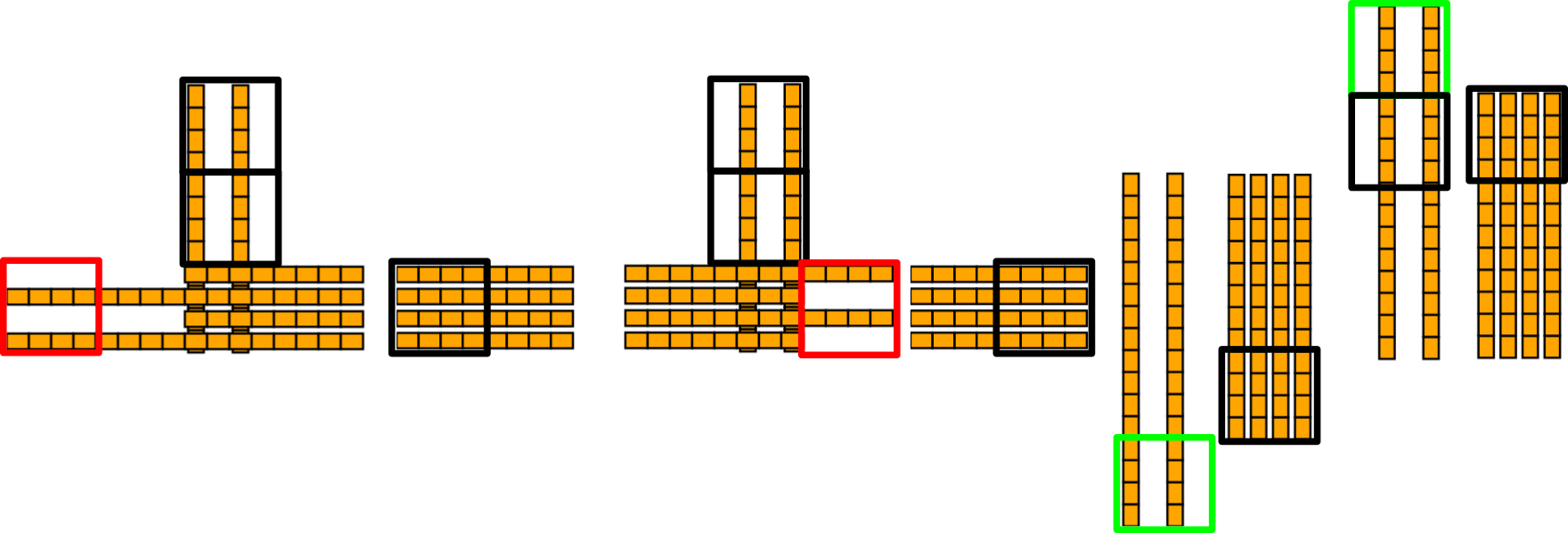}
        \caption{Output slat templates for directed temperature-2 aTAM systems. Output glue locations which are used for adjacent, or strength-2 output are enclosed in a black box, ones which are used for only across-the-gap output are enclosed in a red box, and ones which are used for both adjacent and across-the-gap output are enclosed in a green box. From left to right: west strength-1, west strength-2, east strength-1, east strength-2, south strength-1, south strength-2, north strength-1, north strength-2.}
        \label{fig:deterministic_4x4_output_templates}
    \end{figure}
    
    Let $t_n \in T$, for $0 \le n < |T|$ be the $nth$ in tile set $T$. We will refer to the string ``$t_n$'' as the unique name of $t_n$. Given the directions of growth, and the dynamics of a directed temperature-2 aTAM system, the following list contains all possible valid signatures for any $t_n$.

    \begin{enumerate}
    \item Strength 2 Input Tiles (Figure \ref{fig:deterministic_4x4_S2}):
        \begin{enumerate}
            \item Seed tile: Input=$\emptyset$, Output=(N,$\{0,1,2\}$),(E,$\{0,1,2\}$),(S,$\{0,1,2\}$),(W,$\{0,1,2\}$) 
            \item North tile: Input=(N,2), Output=(E,$\{0,1,2\}$),(S,$\{0,1,2\}$),(W,$\{0,1,2\}$) 
            \item East tile: Input=(E,2), Output=(N,$\{0,1,2\}$),(S,$\{0,1,2\}$),(W,$\{0,1,2\}$) 
            \item South tile: Input=(S,2), Output=(N,$\{0,1,2\}$),(E,$\{0,1,2\}$),(W,$\{0,1,2\}$)
            \item West tile: Input=(W,2), Output=(N,$\{0,1,2\}$),(E,$\{0,1,2\}$),(S,$\{0,1,2\}$) 
        \end{enumerate}
    \item Adjacent Input Tiles (Figure \ref{fig:deterministic_4x4_adjacent}):
        \begin{enumerate}
            \item North-East tile: Input=(N,1),(E,1) Output=(S,$\{0,1,2\}$),(W,$\{0,1,2\}$)
            \item North-West tile: Input=(N,1),(W,1), Output=(E,$\{0,1,2\}$),(S,$\{0,1,2\}$)
            \item South-East tile: Input=(S,1),(E,1), Output=(N,$\{0,1,2\}$),(W,$\{0,1,2\}$)
            \item South-West tile: Input=(S,1),(W,1), Output=(N,$\{0,1,2\}$),(E,$\{0,1,2\}$)
        \end{enumerate}
    \item Across-the-gap Input Tiles (Figure \ref{fig:deterministic_4x4_ATG}):
        \begin{enumerate}
            \item North-South tile: Input=(N,1),(S,1) Output=(E,$\{0,1,2\}$),(W,$\{0,1,2\}$)
            \item East-West tile: Input=(E,1),(W,1) Output=(N,$\{0,1,2\}$),(S,$\{0,1,2\}$)
        \end{enumerate}
    \end{enumerate}

    \begin{figure}
    \centering
    \includegraphics[width=0.4\textwidth]{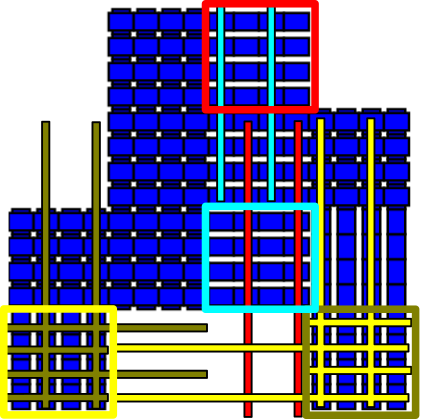}
    \caption{Strength-2 macrotile template for a directed temperature 2 aTAM system. Cells are bounded by squares to show their functionality, and signify cell locations where output slat templates may be added to the macrotile. One of the cells marked by a square may instead be designated as an input, and have its domains assigned such that they connect with those provided by the output slats of a neighboring macrotile whose output is of the same glue type.Overlapping output slats which may enter the macrotile through a mismatch with a neighboring macrotile are shown. Cells, and overlapping output slats are signified using the same color conventions as Figure \ref{fig:deterministic_4x4_color_conventions}. Strength-1 overlapping outputs are displayed, as they are within more $c \times c$ cells than their strength-2 counterparts. No binding interactions occur between a macrotile and any overlapping outputs. Such overlaps can only possibly prevent the resolved macrotile from adding output slats for that direction (if the resolved tile has corresponding output glues), which is clearly unnecessary if a resolved macrotile already exists in that direction and has output the overlapping slats.}
    \label{fig:deterministic_4x4_S2}
    \end{figure}
    
    \begin{figure}
    \centering
    \includegraphics[width=0.9\textwidth]{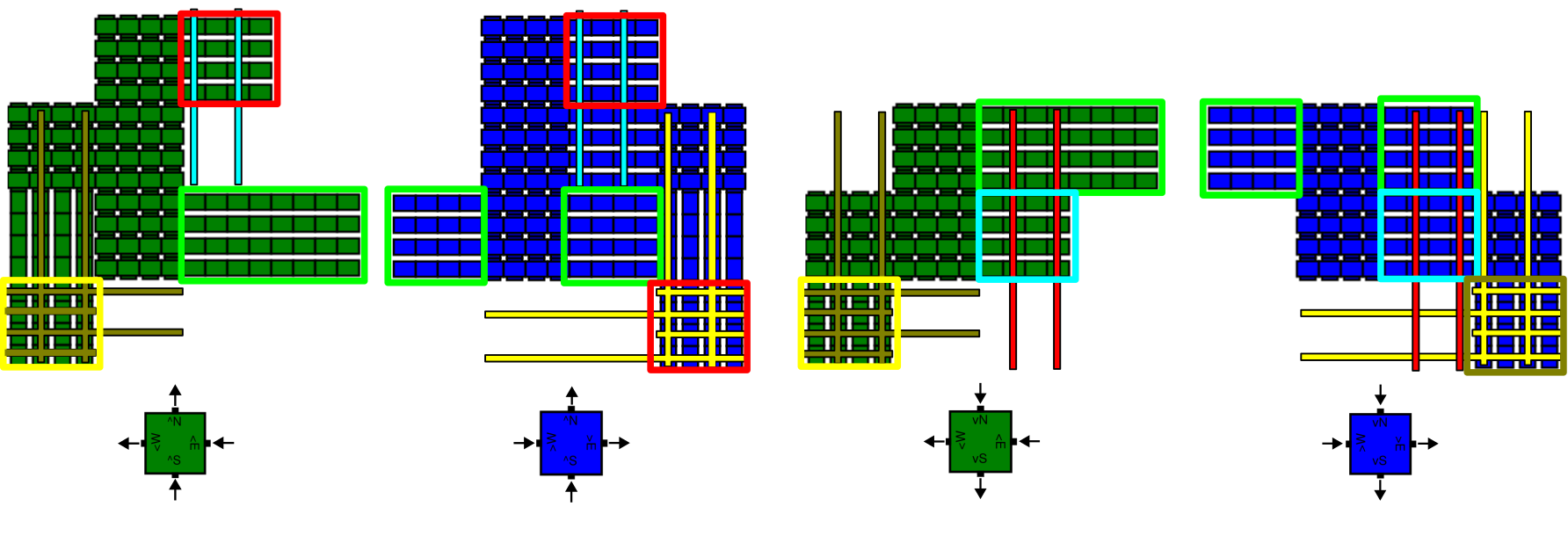}
    \caption{The slat-based macrotile templates for tiles which experience adjacent inputs. Cells are bounded by squares to show their functionality, and correspond to an output slat template which may connect to the associated cell. Overlapping output slats which may enter the macrotile through a mismatch with a neighboring macrotile are shown. Cells, and overlapping output slats are signified using the same color conventions as Figure \ref{fig:deterministic_4x4_color_conventions}. Strength-1 overlapping outputs are displayed, as they are within more $c \times c$ cells than their strength-2 counterparts. No binding interactions occur between a macrotile, and any overlapping outputs. Such overlaps can only possibly prevent the resolved macrotile from adding output slats for that direction (if the resolved tile has corresponding output glues), which is clearly unnecessary if a resolved macrotile already exists in that direction and has output the overlapping slats.}
    \label{fig:deterministic_4x4_adjacent}
    \end{figure}

    \begin{figure}
    \centering
    \includegraphics[width=0.4\textwidth]{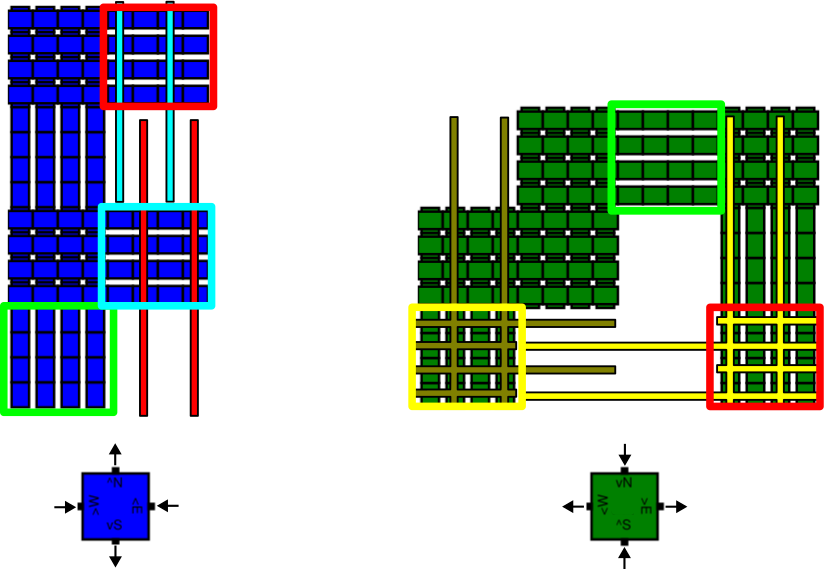}
    \caption{The slat-based macrotile templates for tiles which experience across-the-gap inputs. Cells are bounded by squares to show their functionality, and correspond to an output slat template which may connect to the associated cell. Cells are signified using the same color conventions as Figure \ref{fig:deterministic_4x4_color_conventions}. Overlapping output slats which may enter the macrotile through a mismatch with a neighboring macrotile are shown, signified using the same color conventions as Figure \ref{fig:deterministic_4x4_color_conventions}. Strength-1 overlapping outputs are displayed, as they are within more $c \times c$ cells than their strength-2 counterparts. No binding interactions occur between a macrotile, and any overlapping output. Such overlaps can only possibly prevent the resolved macrotile from adding output slats for that direction (if the resolved tile has corresponding output glues), which is clearly unnecessary if a resolved macrotile already exists in that direction and has output the overlapping slats.}
    \label{fig:deterministic_4x4_ATG}
    \end{figure}
    
    When $t_n$ is simulated by a macrotile, its output slats may only assemble into a small, fixed set of configurations to allow for cooperation with macrotiles in adjacent locations. Each valid output slat template is shown in Figure \ref{fig:deterministic_4x4_output_templates}.    
    
    In order to guaranteed that mismatches do not block the construction of $t_n$, the following conditions must be met: 
    \begin{enumerate}
        \item The presence of any additional output slats over what is required to cause a given macrotile to resolve to $t_n$ may not block the ability of that macrotile to resolve to $t_n$.

        \item The presence of any additional output slats over what is required to cause a given macrotile to resolve to $t_n$ may not block the ability of the macrotile representing $t_n$ to exhibit output in any direction apart from that which is currently occupied by incident output slats.
        
    \end{enumerate}

    Both of these conditions are satisfied in our macrotiles for simulating directed temperature-2 aTAM systems. Figure \ref{fig:deterministic_4x4_S2} shows a macrotile template for any strength-2 input tile signatures, and Figures \ref{fig:deterministic_4x4_adjacent} and \ref{fig:deterministic_4x4_ATG} show tiles with adjacent and across-the-gap input tile signatures and their corresponding macrotile templates, respectively, while Figure \ref{fig:deterministic_4x4_examples} shows some additional macrotile examples and Figure \ref{fig:deterministic_4x4_output_cooperation} shows all possible combinations of cooperation.
    In these figures, output locations are marked to signify the ability to include an output slat template.
    Additionally, Figures \ref{fig:deterministic_4x4_S2}-\ref{fig:deterministic_4x4_ATG} include the overlapping placements of potentially incident output slat groups which may be placed into the macrotile in the cases of simulating mismatches in $\calT$. By inspecting these figures, it may be verified that none of those incident output slat groups would block either the ability for the macrotile to resolve to $t_n$, or prevent it from creating output in any direction apart from that which is currently occupied by the incident slat group.
    The longest slat-length to be exhibited by a macrotile is of length $4c$, which is shown in Figure \ref{fig:deterministic_4x4_S2_example}. To build $\mathcal{S}$ for each $t_n$, we use the same methods as are described in Section \ref{sec:zig-zag}.
    As with the macrotile representation $R$ in the proof of Theorem \ref{thm:zig-zag}, $R$ can simply inspect each macrotile location (each of size $4c \times 4c$) and resolve to the correct corresponding tile as soon as a body slat appears determining the tile of $\calT$ to which it should map by the prefix of its name. Furthermore, following the same arguments as in the proof of Theorem \ref{thm:zig-zag}, the construction succeeds in building $\mathcal{S}$ so that it simulates $\calT$ under $R$ by proper and consistent alignment of input and output regions and correct assignment of glues.  (An example of a partial assembly can be seen in Figure \ref{fig:standard_4x4+ATG_construction}.)
    Therefore, $\mathcal{S}$ simulates $\calT$, an arbitrary directed temperature-2 aTAM system, under $R$ using cooperativity $c$ and macrotiles of size $4c \times 4c$, with a maximum slat length of $4c$, and Theorem \ref{thm:deterministic+mismatches} is proven.

    \begin{figure}
    \centering
    \begin{subfigure}{0.3\textwidth}
        \centering
        \includegraphics[width=1.0\textwidth]{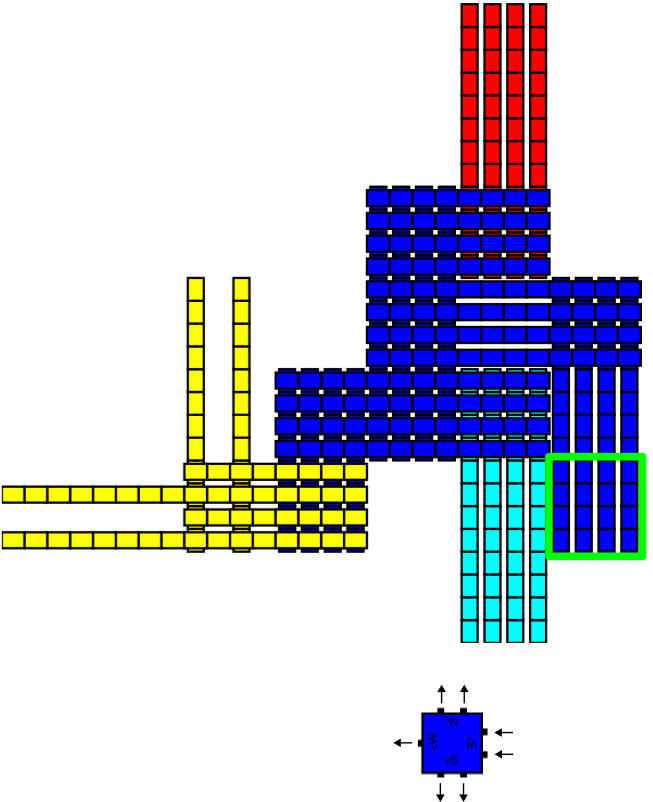}
        \caption{\label{fig:deterministic_4x4_S2_example}}
    \end{subfigure}
    \hspace{20pt}
    \begin{subfigure}{0.3\textwidth}
        \centering
        \includegraphics[width=1.0\textwidth]{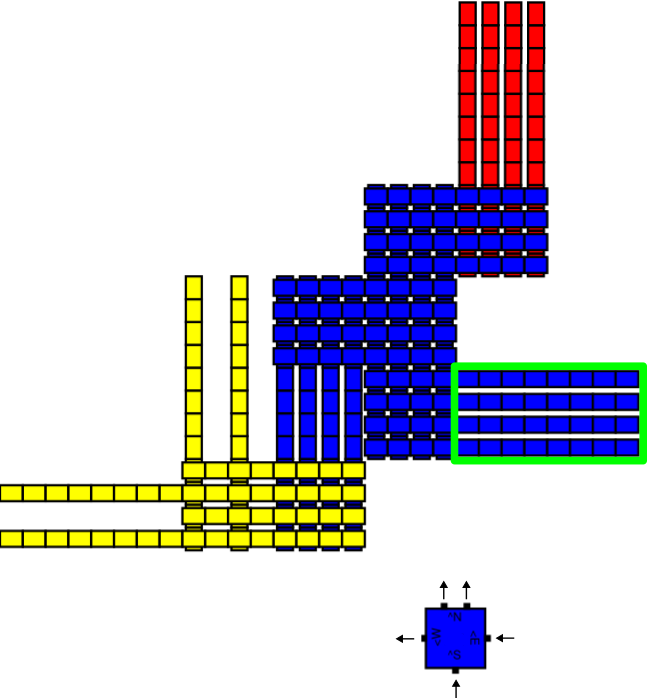}
        \caption{\label{fig:deterministic_4x4_SE_example}}
    \end{subfigure}
    \hspace{20pt}
    \caption{Example macrotiles for a directed temperature-2 aTAM system. (a) Macrotile which recieves an east strength-2 input, and provides north strength-2, south strength-2, and west strength-1 outputs. (b) Macrotile which recieves east and south strength-1 inputs, and provides north strength-2, and west strength-1 outputs. Output slat templates are colored in accordance to Figure \ref{fig:deterministic_4x4_color_conventions}, and input domains are marked with a green box. \label{fig:deterministic_4x4_examples}}
    \end{figure}
    
    \begin{figure}
    \centering
    \begin{subfigure}{0.24\textwidth}
        \centering
        \includegraphics[width=1.0\textwidth]{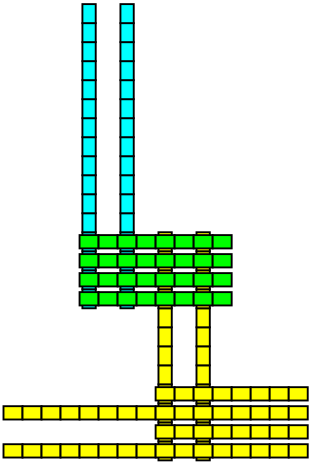}
        \caption{\label{fig:deterministic_4x4_NE_coop}}
    \end{subfigure}
    \hspace{20pt}
    \begin{subfigure}{0.24\textwidth}
        \centering
        \includegraphics[width=1.0\textwidth]{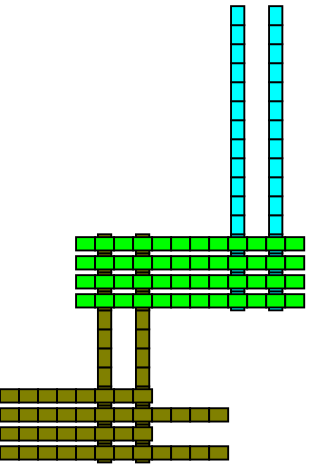}
        \caption{\label{fig:deterministic_4x4_NW_coop}}
    \end{subfigure}
    \hspace{20pt}
    \begin{subfigure}{0.24\textwidth}
        \centering
        \includegraphics[width=1.0\textwidth]{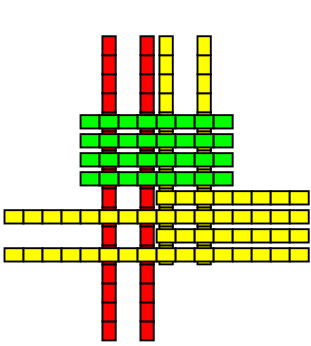}
        \caption{\label{fig:deterministic_4x4_SE_coop}}
    \end{subfigure}
    \begin{subfigure}{0.24\textwidth}
        \centering
        \includegraphics[width=1.0\textwidth]{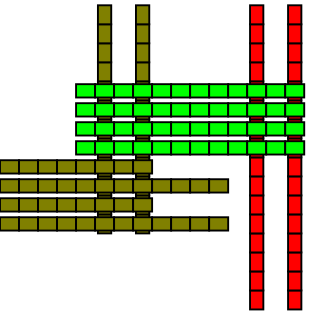}
        \caption{\label{fig:deterministic_4x4_SW_coop}}
    \end{subfigure}
    \begin{subfigure}{0.36\textwidth}
        \centering
        \includegraphics[width=1.0\textwidth]{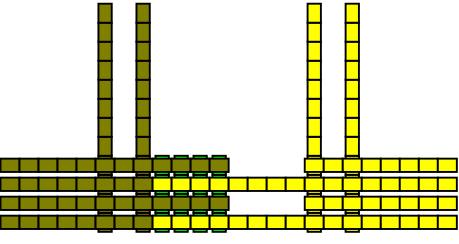}
        \caption{\label{fig:deterministic_4x4_EW_coop}}
    \end{subfigure}
    \begin{subfigure}{0.06\textwidth}
        \centering
        \includegraphics[width=1.0\textwidth]{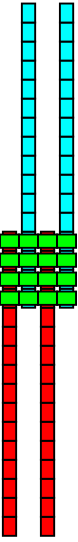}
        \caption{\label{fig:deterministic_4x4_NS_coop}}
    \end{subfigure}
    \caption{ Examples of cooperation in the construction of directed temperature-2 aTAM systems, showing only the output slats from neighbors and body slats of the resolving macrotile. (a) north and east output slat templates resolve macrotile $t_n$ through adjacent cooperation. (b) north and west output slat templates resolve macrotile $t_n$ through adjacent cooperation. (c) south and east output slat templates resolve macrotile $t_n$ through adjacent cooperation. (d) south and west output slat templates resolve macrotile $t_n$ through adjacent cooperation. (e) east and west output slat templates resolve macrotile $t_n$ through across-the-gap cooperation. (f) north and south output slat templates resolve macrotile $t_n$ through across-the-gap cooperation. Slat groups are colored in accordance with the color conventions in Figure \ref{fig:deterministic_4x4_color_conventions}.\label{fig:deterministic_4x4_output_cooperation}}
    \end{figure}

    \begin{figure}
    \centering
    \includegraphics[width=1.0\textwidth]{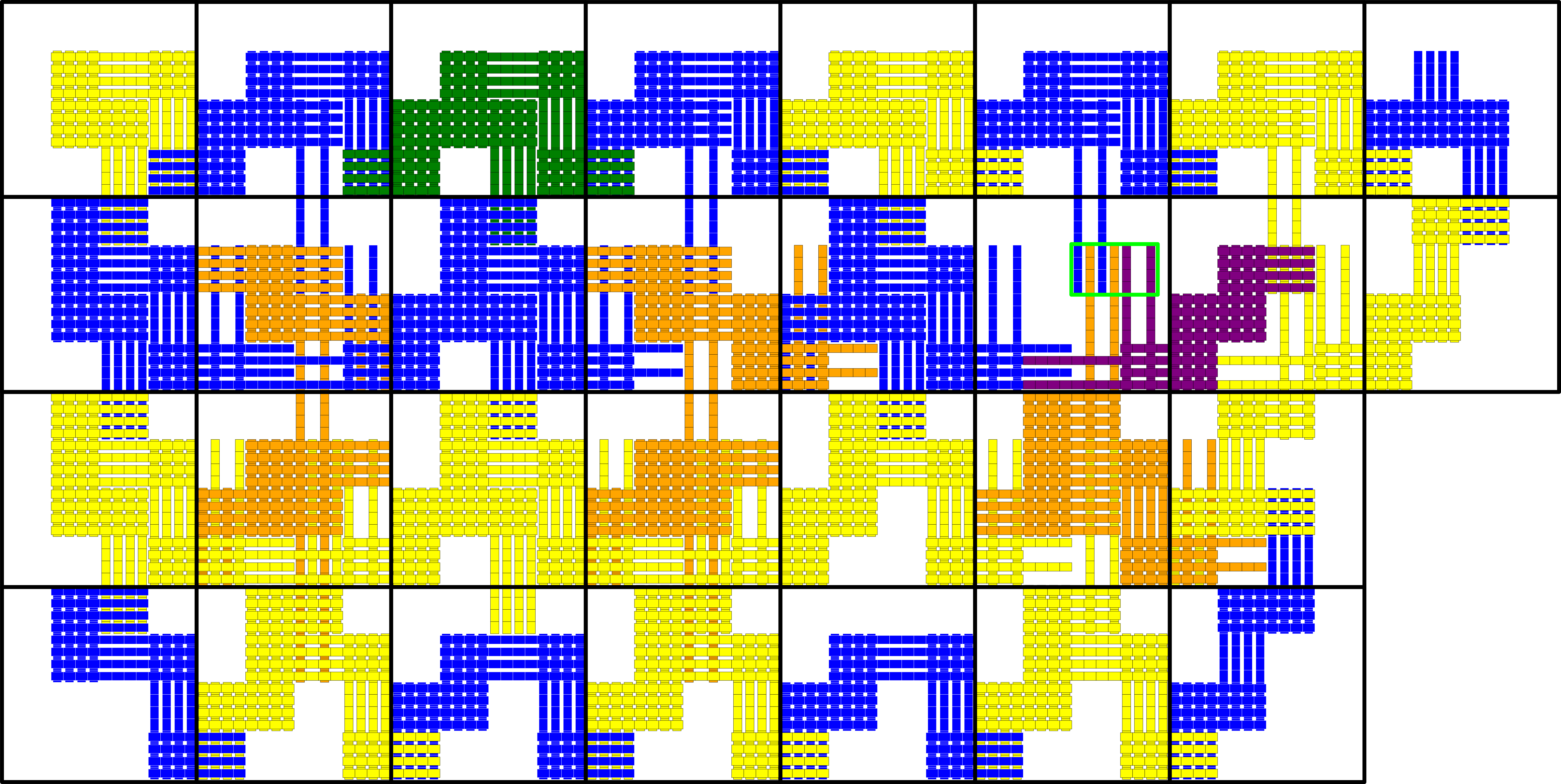}
    \caption{An example an assembly composed of (partial) $4c \times 4c$ macrotiles simulating a directed temperature-2 aTAM system for $c = 4$. Thirty of the macrotile locations are outlined in black squares. The macrotile simulating $t$ from Figure \ref{fig:deterministic_4x4_cell_coloring} would attach into the macrotile location at coordinates $(5,2)$, containing a green square. Its (light blue) body slats would attach to the 2 west output slats from the macrotile to the east (purple) and the 2 north output slats from the macrotile to the south (orange). These light blue slats can bind in any order, and as soon as one binds, the macrotile resolves to $t$. Only once they have all 4 bound can the red body slats bind, then the purple or blue body slats can bind. Only once either of those have bound can the green or orange output slats bind, then after the green output slats bind, may the grey output slats bind. Thus, the growth of a macrotile is well-ordered, and outputs are only presented after a macrotile resolves, enforcing the restrictions of simulation. In this example, the growth of purple body slats, and orange output slats will be blocked from assembly, this is expected behavior, as it is the result of a mismatch occurring. Despite the blocking that occurs, the macrotile still resolves to $t$, and the mismatch does not block simulation.}
    \label{fig:standard_4x4+ATG_construction}
    \end{figure}
    
\end{proof}


\subsection{Technical details for the simulation of the full class of aTAM systems}\label{sec:full-sim-appendix}

In this section we provide the full technical details for the proof of Theorem \ref{thm:full-aTAM}. This section expands on the overview provided in Section~\ref{sec:full-sim}.

\begin{proof}

To prove Theorem~\ref{thm:full-aTAM}, let $\calT = (T, \sigma, 2)$ be an arbitrary directed temperature-2 aTAM system. Here we describe the construction of a slat system $\calS = (S, \sigma', c)$ which intrinsically simulates $\calT$. We note that without loss of generality, $\calT$ is assumed to be an IO-TAS so that all of its tile types are IO-marked. (Otherwise, we can apply the simple transformation discussed in Section \ref{sec:IO-marking}) Keep in mind that the process of adding IO-marks to an arbitrary aTAM tile set will result in tiles with a minimal number of input glues. That is, we can enforce that every set of input glues on a specific tile type is both sufficient and necessary for the attachment of that tile. This will be important for this construction. Let $G_T$ be the set of all glues appearing on tiles in $T$ and let the cooperativity $c$ be any positive even integer. The slats in $S$ can be organized into 3 different categories: \emph{input} slats which bind to the output of adjacent macrotiles, \emph{decision} slats which act to choose the tile type in $T$ to which a macrotile will resolve, and \emph{output} slats which carry the information from the decision slats to the sides of the macrotile after it has resolved.

We note that all slats in this construction are defined with lengths that are multiples of $c$. For a slat of length $L=Nc$, there are two convenient coordinate systems: the \emph{natural coordinates} and the \emph{segment coordinates}. In the natural coordinate system, locations along the slat are indexed by a variable $x$ ranging from $0$ to $L-1$, increasing from west to east for horizontal slats and north to south for vertical slats. In segment coordinates, the slat is logically divided into $N$ contiguous sections of length $c$ and each section is indexed by variable $i$ ranging from $0$ to $N-1$. Then, each location within a section is indexed by variable $j$ ranging from $0$ to $c-1$, so that the location along the slat at segment coordinates $(i,j)$ is the same as the location with natural coordinates $x=Ni+j$. Throughout this section, we will describe the glues of slats as functions of these coordinates, noting again that we adhere to the convention that any horizontal slat will only have glues on its D faces and vertical slats only on their U faces. Each location along the length of a slat therefore only admits one potential glue location. Additionally, all slats will generally be defined as a group of $c$ slats, so a 3rd coordinate $k$ ranging from $0$ to $c-1$ may be introduced to refer to a specific slat among $c$ slats in a group.

\begin{figure}
    \centering
    \includegraphics[width=0.5\textwidth]{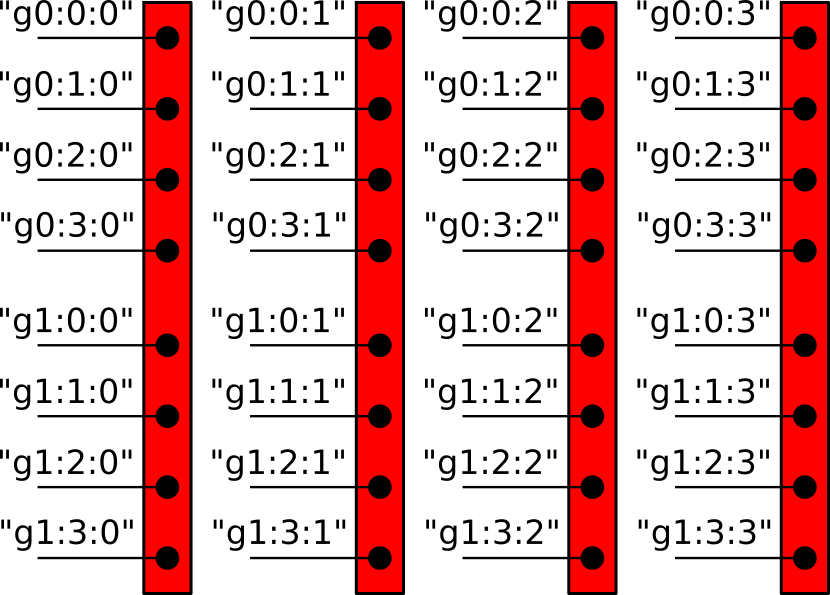}
    \caption{A group of $c$ slats, each of length $2c$ defined by the glue function $f(i,j,k) = \text{\say{g[i][j][k]}}$. In this case $c=4$. Glues on the slats are illustrated as black dots with associated string labels. The segment coordinate $i$ increases for each group of $c$-glues along the length of a slat, $j$ increases within each group of $c$-glues on a slat, and $k$ indexes among a group of $c$ slats.}
    \label{fig:undirected-segment-coordinates}
\end{figure}

Using these segment coordinates for a slat type, along with the additional coordinate $k$ to distinguish between the $c$ slats in a group, we may conveniently describe the glues on an entire slat group using a function, called a \emph{glue function}, that maps the coordinate triple $(i,j,k)$ to a glue label. Figure~\ref{fig:undirected-segment-coordinates} illustrates a group of vertical slats defined by the glue function $f(i,j,k) = \text{\say{g[i]:[j]:[k]}}$, where the \say{guillemets} denote a \emph{string form} where variables enclosed in square brackets are replaced by the corresponding value of the respective variable, similar to format strings available in many programming languages. If $i=2$, $j=3$, and $k=0$ for instance, then the string form \say{g[i]:[j]:[k]} is the same as the string ``g2:3:0''. If a location along a slat is not specified by a glue function or if the glue is assigned the empty string, then it is assumed that no glue exists at that location.

\begin{figure}
    \centering
    \includegraphics[width=0.9\textwidth]{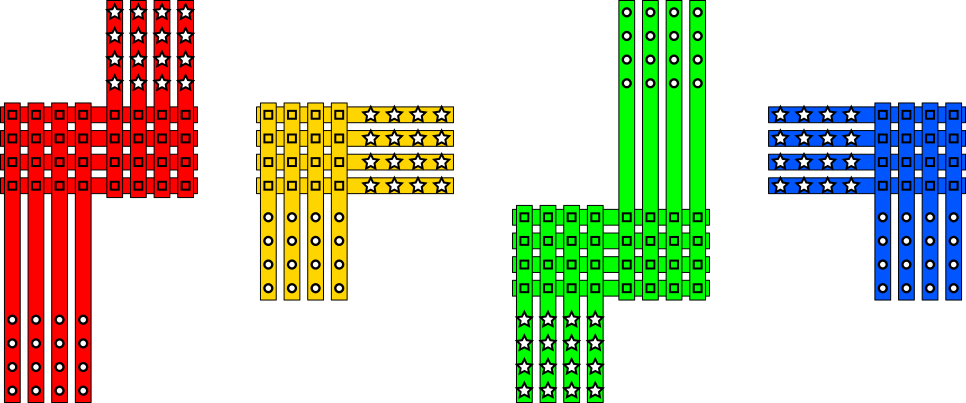}
    \caption{The slats associated with glues for each IO-marked direction assuming the cooperativity $c$ is $4$. The red slats are associated with a `$\vee$' marked glue, the yellow with a `$<$' marked glue, the green with a `$\wedge$' marked glue, and the blue with a `$>$' marked glue. White stars on the slats represents the part of the slat which encode the output from a neighboring macrotile. Squares indicate a pair of complementary interior slat glues unique to the specific location between the corresponding pair of slats. Circles represent slat glues encoding the tile glue being represented by the slats.}
    \label{fig:undirected-glue-slats}
\end{figure}

\paragraph*{Input Slats} The input slats are defined with respect to the glues in $G_T$. Let $g\in G_T$ be a glue with input direction marking $d$. The shape and number of slats defined for $g$ depends solely on $d$, but the glues appearing on the slats depends on the particular glue as well. Figure~\ref{fig:undirected-glue-slats} illustrates the general shape of these slats for each direction of input marking. We denote the set of input slats associated with a glue $g$ as $S_g$. If the input marking $d$ is $\vee$, then $S_g$ will consist of $3c$ slats, illustrated in red in Figure~\ref{fig:undirected-glue-slats} and described here. The first $c$ of these slats are vertical slats of length $2c$ described by the glue function:
\[
f_1(i,j,k) = 
\begin{cases}
\text{\say{IO([g])-[j]:[k]}} & \text{when } i = 0 \\
\text{\say{C1([g])-[j]:[k]}} & \text{when } i = 1 \\
\end{cases}
\]
The next $c$ slats are horizontal with length $2c$ and are described by the glue function:
\[
f_2(i,j,k) = 
\begin{cases}
\text{\say{C2([g])-[k]:[j]*}} & \text{when } i = 0 \\
\text{\say{C1([g])-[k]:[j]*}} & \text{when } i = 1 \\
\end{cases}
\]
And the remaining $c$ slats are vertical with length $3c$ and are described by:
\[
f_3(i,j,k) = 
\begin{cases}
\text{\say{C2([g])-[j]:[k]}} & \text{when } i = 0 \\
\text{\say{}} & \text{when } i = 1 \\
\text{\say{GI([g])-[j]:[k]}} & \text{when } i = 2 \\
\end{cases}
\]

In the glue definitions above, ``[g]'' refers to the label of the IO-marked glue $g$. The glues prefixed with ``IO'' are intended to bind to the output glues from a neighboring macrotile which will have the same labels. The glues prefixed with ``C1'' or ``C2'' denote internal glues which serve only to connect the 3 different slat groups so that they attach in the configuration shown in Figure~\ref{fig:undirected-glue-slats}. Lastly, the glues prefixed by ``GI'' simply act to encode the input glue $g$ for use by the decision slats. Each of the glue ends with a string of the form \say{[row]:[col]} denoting the row and column of the glue within the corresponding $c\times c$ cell. Note that in horizontal slats, the coordinate $k$ acts as a row and $j$ as a column, while the opposite is true for vertical slats. This serves only to distinguish the glues and to ensure that only the desired slats attach at the intended offsets. In the cases where $d$ is $<$, $\wedge$, or $>$, the slats are defined similarly so that the stars in Figure~\ref{fig:undirected-glue-slats} are realized as ``IO'' prefixed glues and the circles are realized as ``G'' prefixed glues. 

\begin{figure}
    \centering
    \includegraphics[width=0.5\textwidth]{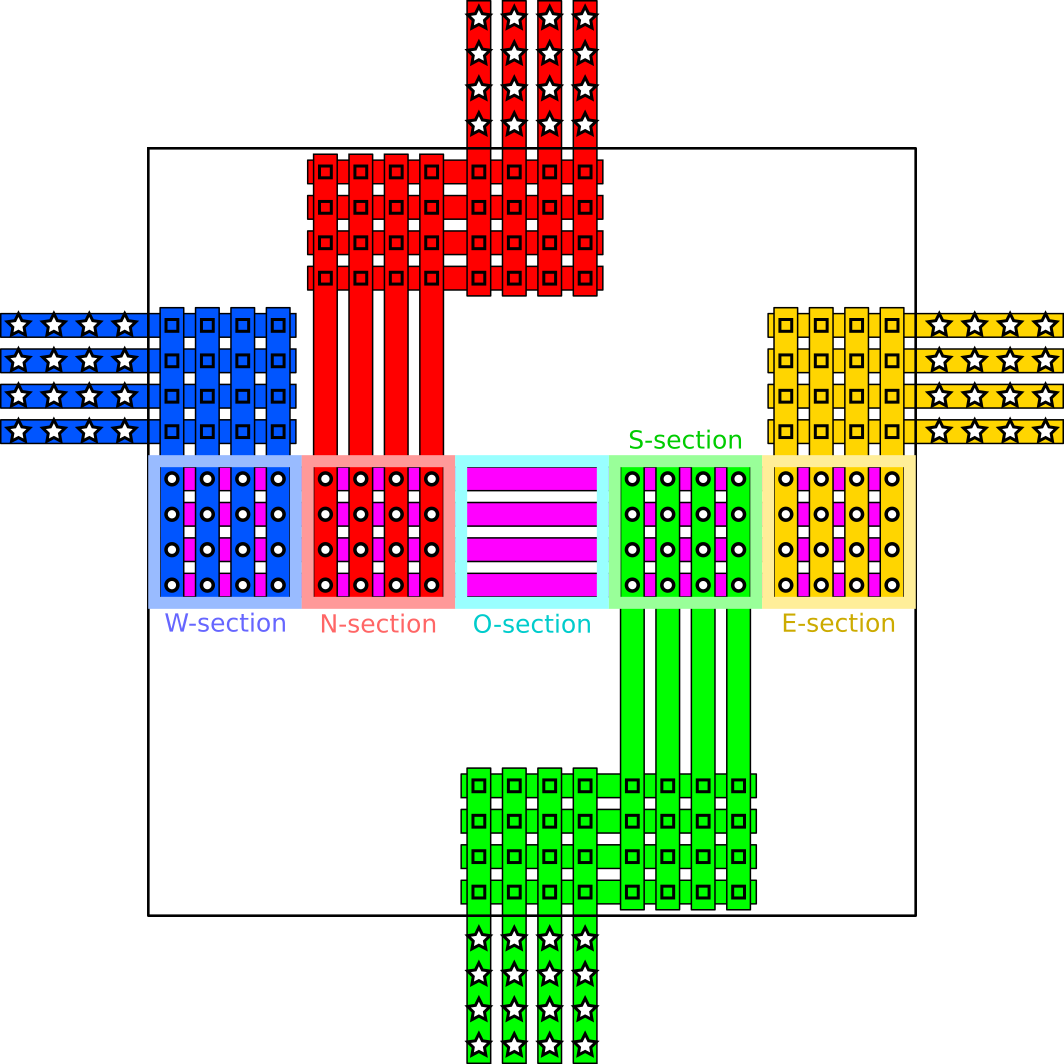}
    \caption{Decision slats (magenta) are logically divided into 5 sections, 4 for the directional inputs from adjacent macrotiles and one which acts as an output. The various sections are organized along the decision slats so that they match with the relative positions of the input slats within a macrotile.}
    \label{fig:undirected-decision-sections}
\end{figure}

\paragraph*{Decision Slats}

Recall the tile types in $T$ are assumed to be IO-marked. For each tile type $t\in T$ we define $c$ horizontal decision slats $s^k_t$ for $k\in\{0,\ldots,c-1\}$, each of length $5c$, which we think of as being logically divided into $5$ contiguous sections of length $c$. From west to east, these sections are called the $W$-section, $N$-section, $O$-section, $S$-section, and $E$-section, where $W,N,S,E$ correspond to the cardinal directions and $O$ denotes an \emph{output section} whose purpose will be explained shortly. In other words, the segment coordinate $i$ of the $W$-, $N$-, $O$-, $S$-, and $E$-sections are $0$, $1$, $2$, $3$, and $4$ respectively. The slat type $s^0_t$ will play a slightly different role from the slats $s^1_t, \ldots, s^{c-1}_t$, specifically this slat will be responsible for causing the macrotile to resolve. The glues appearing in the $O$-section of slat type $s^0_t$ will encode that the macrotile represents tile type $t$, while the glues appearing in the $O$-section of the other $s^k_t$ slat types will simply consist of generic glues, not specific to any tile type in $T$.
In regards to the directional sections, the glues on slat type $s^k_t$ depend on how many input glues tile type $t$ has. If the number of input glues is $1$, then that glue $g$ must have a strength of at least 2. In this case, let $d$ be the direction of glue $g$ and define the slats $s^k_t$ so that it has glues \say{GI([g])-[k]:[j]*} in all $c$ of it's $j$ coordinates in the $d$ section. In this way the slats $s^k_t$ may attach to the input slats encoding the glue $g$ from direction $d$. On the other hand, if tile type $t$ has 2 input glues $g_1$ from direction $d_1$ and $g_2$ from direction $d_2$, then the slats $s^k_t$ will have glues with label \say{GI([g$_1$])-[k]:[j]*} in the $d_1$-section for locations where $j$ is between $0$ to $c/2 - 1$, and likewise glues with label \say{GI([g$_2$])-[k]:[j]*} in the $j\in\{0,\ldots,c/2-1\}$ half of the $d_2$-section. In other words, decision slats corresponding to tiles with a strength-2 input will be able to attach any time input slats encoding that strength-2 input are present, while decision slats corresponding to tiles with cooperative inputs only have half of their glues present in each directional region and require both sets of input slats to be present before they may attach. This is illustrated in Figure~\ref{fig:undirected-decision-slat-example}.

\begin{figure}
    \centering
    \includegraphics[width=\textwidth]{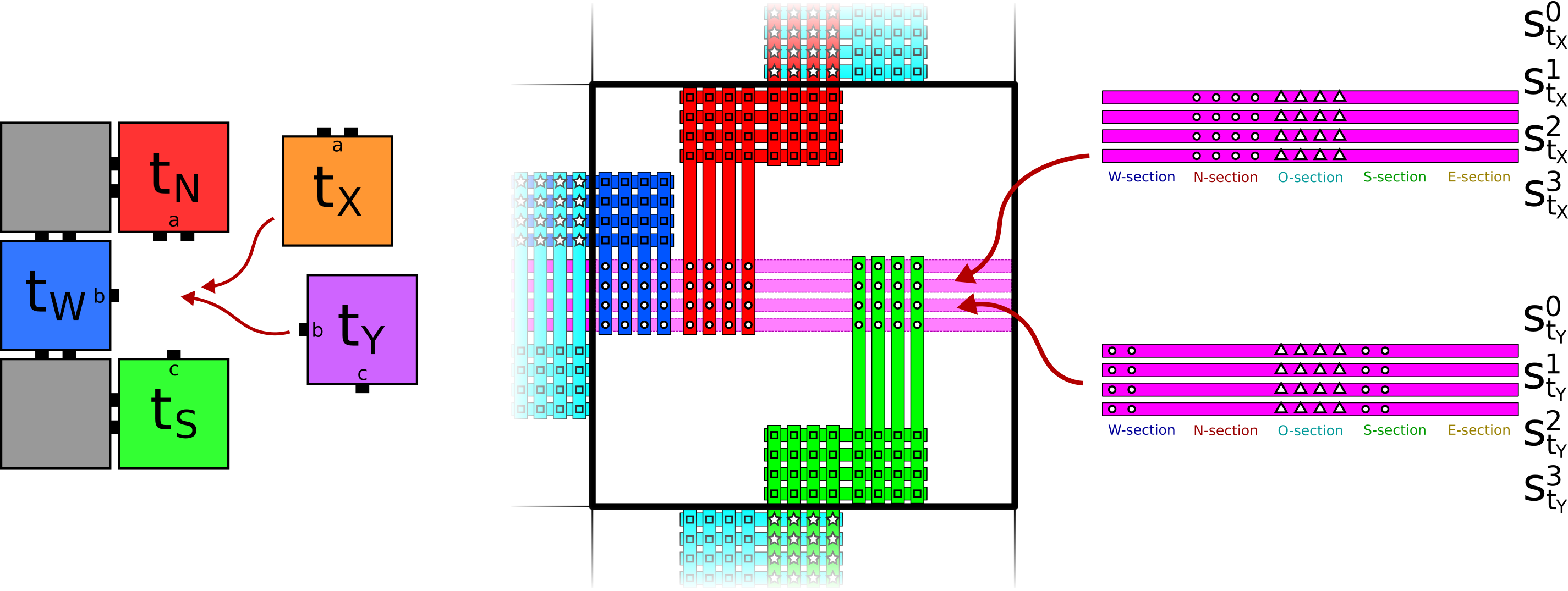}
    \caption{An example undirected aTAM system (left) where both tiles $t_X$ and $t_Y$ may attach in the same location. Tile $t_X$ may attach with a strength-2 glue on its north and $t_Y$ by cooperation between strength-1 glues on its west and south. The example aSAM simulation (right) has input slats with glues in locations for the north, south, and west corresponding to the glues presented by $t_X$'s north and $t_Y$'s south and west, respectively. In order to simulate the undirected attachment, the decision slats (magenta) may attach non-deterministically into the center locations within the slat macrotile. The decision slats corresponding to $t_X$ have glues occupying all locations in their $N$-section, allowing them to attach when only a set of north input slats are present corresponding to the tile glue $a$. The $t_Y$ decision slats have their glues in the $W$- and $S$-sections, but only in half of the possible locations. In this way, both the input slats corresponding to tile glues $b$ and $c$ must be present for these slats to attach. Keep in mind that a mix of decision slats corresponding to $t_X$ and $t_Y$ may attach since each is independent from the others; however only the slat with segment coordinate $k=0$ will determine the output, the other slats have generic glues in their $O$-section. Note that it is not shown in this figure, but glues in the aTAM system are IO marked so that an $a$ glue presented from the north would be incompatible with slats representing an $a$ glue from the east for instance.}
    \label{fig:undirected-decision-slat-example}
\end{figure}

The $O$-section of $s^0_t$ will have glues with label \say{T[t]-0:[j]} for each $j\in\{0,\ldots,c-1\}$ where ``[t]'' refers to some string distinct to tile type $t$. In other words, the $O$-section encodes the tile $t$ as well as its $j$ coordinate within the section. The reason for encoding the $j$ coordinate is simply to ensure that the output slats attaching to these decision slats each have a dedicated attachment site. While the slats $s^1_t,\ldots, s^{c-1}_t$ are defined similarly to $s^0_t$ in regards to their directional sections, the output section for these slats contain generic glues. Specifically, regardless of which tile $t$ is being represented, these slats will have the glue \say{TX-[k]:[j]*} in the respective locations of their $O$-section. Consequently, only the $s^0_t$ slat has any say in how the macrotile will resolve, the others simply act to ensure that there are $c$ glues arranged vertically so that output slats may eventually attach to the $O$-section. The macrotile representation function for this simulation simply checks for the slat $s^0_t$; whichever tile is represented by the specific slat that attaches in that location is the tile to which the macrotile resolves.

In addition to the slats $s^k_t$ defined for every tile type $t\in T$, we also define the slats $c^k_t$ and $d^k_t$ for $k\in\{0,\ldots,c-1\}$ and $t\in T$ called \emph{decision layout slats}. These slats attach to the decision slats described above once the macrotile has resolved and provide a place for the output slats described below to attach and present an encoding of the output glues of $t$ to the relevant sides. The tiles $c^k_t$ are each vertical slats of length $2c$; the northernmost $c$ glues on these tiles attach to the $O$-section glues from the decision slats. These can be described by the glue function:
\[
c^k_t(i,j,k) = 
\begin{cases}
\text{\say{T([t])-0:[k]}} & \text{when } i = 0 \\
\text{\say{TX-[j]:[k]}} & \text{when } i = 0 \text{ and } j>0 \\
\text{\say{D([t])-[j]:[k]}} & \text{when } i = 1 \\
\end{cases}
\]
Notice that the northernmost glue is different than the next $c-1$ glues corresponding to the way the decision tiles $s^k_t$ are defined. The $d^k_t$ slat types depend on the output glues of tile type $t$, very similar to how the glues on the decision tiles depended on the input glues. These slats are also horizontal with length $5c$ and may be thought of as divided into the same directional sections. The glues in the $O$-section simply match with the slat types $c^k_t$. That is, the $O$-sections contain glues of the form \say{D([t])-[k]:[j]*} so that they may only attach if the macrotile resolved to tile type $t$. The remaining sections will have all $c$ glues if $t$ has an output glue in the corresponding direction. In other words, if $t$ has output glue $g$ on its $d$ side, then the $c$ slats $d^k_t$ will have glues of the form \say{GO([g])-[k]:[j]*} in the $c$ locations indexed by $j$ of the $d$-section (to restate: $i=0$ for $W$, $i=1$ for $N$, $i=3$ for $S$, and $i=4$ for $E$).

\paragraph*{Output slats}
Finally, for each glue $g\in G_T$, a set of output slats is defined to propagate the output glues to the respective side. These are designed similarly to the input slats, but instead of attaching to glues with the prefix ``IO'' and placing glues with the prefix ``GI'' (for \emph{glue input}), the output slats attach to glues with the prefix ``GO'' (for \emph{glue output}) and eventually place glues prefixed by ``IO''. Figure~\ref{fig:undirected-out-slats} illustrates what the output slats look like for the different directions an output glue may have.
                  {'coords': [(6,-6),(6,-5)]}],

\begin{figure}
    \centering
    \includegraphics[width=0.8\textwidth]{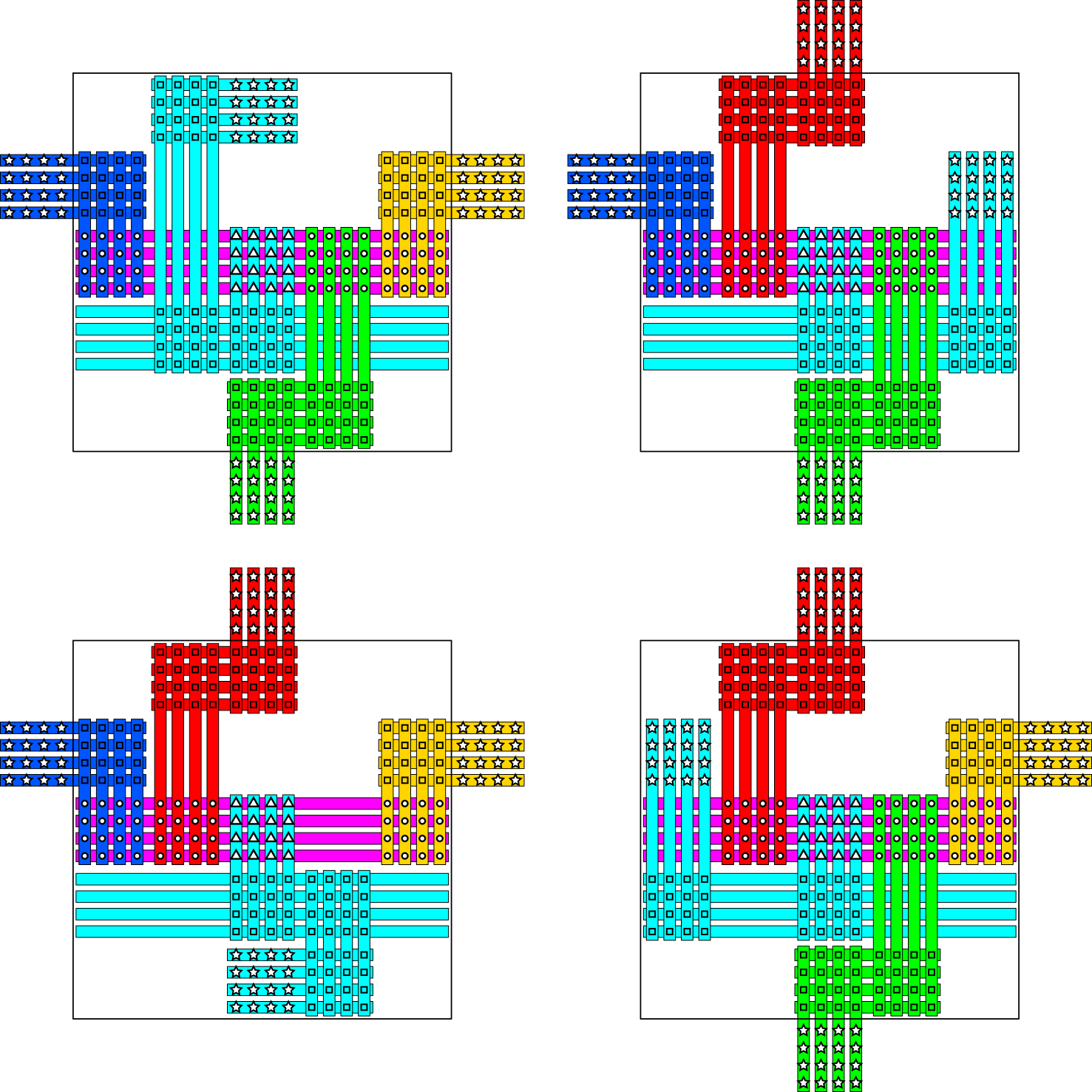}
    \caption{Output slats (cyan) growing into each of the 4 directions: north (top left), east (top right), south (bottom left), and west (bottom right). Output slats begin by attaching to the glues (illustrated as white triangles) in the $O$-section of the decision slats (magenta). A row of horizontal slats then attach with glues that encode the output directions of the aTAM tile into which the macrotile resolved.}
    \label{fig:undirected-out-slats}
\end{figure}

The output slats for a glue $g\in G_T$ with output direction $\wedge$ consist of $2$ groups of $c$ slats. The first group consists of vertical length-$4c$ slats, while the second consists of horizontal length-$2c$ slats. The glues on these slats may be described respectively by the following glue functions
\[
f_4(i,j,k) = 
\begin{cases}
\text{\say{C3([g])-[j]:[k]}} & \text{when } i = 0 \\
\text{\say{GO([g])-[j]:[k]}} & \text{when } i = 3 \\
\end{cases}
\]
and
\[
f_5(i,j,k) = 
\begin{cases}
\text{\say{C3([g])-[k]:[j]*}} & \text{when } i = 0 \\
\text{\say{IO([g])-[k]:[j]*}} & \text{when } i = 1 \\
\end{cases}
\]
The output slats for glues with other directions are defined similarly so that they match the illustrations in Figure~\ref{fig:undirected-out-slats}. The result will be a set of glues associated with the glue $g$ with labels prefixed by``IO''. These will act to allow the input glues to attach in neighboring macrotiles.

\paragraph*{Simulation}

To verify that the slats presented in this construction correctly simulate the tiles in $T$, we first describe how to convert the seed $\sigma$ into the simulation seed $\sigma'$. For each tile type $t$ at location $(x,y)\in \dom{\sigma}$, we simply place the decision slats $s^k_t$, $c^k_t$ and $d^k_t$ corresponding to $t$ in their respective locations within the macrotile whose south-westernmost corner is $(cx, cy)$. Then corresponding to each output glue on $t$, we place the corresponding output slats so they attach in the appropriate locations along the decision tiles. It is often the case that the seed $\sigma$ consists of a single tile with no input glues, but in the case that the seed consists of multiple tiles, there will be input glues on all but one of the tiles. For these we also include the corresponding input slats in the seed $\sigma'$.

As previously mentioned, the macrotile representation function $R$ is simply defined so that it inspects the location of the decision slat in each macrotile location. If it is empty, $R$ maps the location to empty space, else it maps the macrotile to the unique tile type of $T$ corresponding to the decision slat.

It should be clear from the construction above that the glues between the input slats corresponding to a particular glue $g\in G_T$ are dedicated to particular $c\times c$ cells within each macrotile. In other words, because these glues are labelled with the corresponding direction and glue they represent, the input slats are required to attach as a group if the neighboring macrotile has presented the corresponding ``IO'' prefixed glues. The only exception to this occurs when there is a mismatch between the glues of 2 tiles in $\calT$. That being said, the process of IO-marking tiles in $T$ ensures that all glues present on the exterior of $\calT$-assembly are output glues. This fact ensures that the only type of mismatch that may occur between the glues of tiles in a $\calT$-assembly is between 2 output glues. In this way, the only thing that may hinder the input slats corresponding to a glue from $G_T$ from attaching as a group is if they attempt to grow in an already resolved macrotile with output slats already present. This clearly cannot affect the resolution of either macrotile since both must have resolved already, and furthermore it cannot affect the other output slats (those for glues being presented in the other directions) since any newly placed input slats will not affect the decision tiles already placed. Once enough input slats are present for the decision slats to attach, there may be a non-deterministic choice for the which of the slats $s^k_t$ end up attaching, but by the definition of the slats above, after these attachments it is clear which output slats must then attach. The fact that mismatches are handled solely in the cells of the macrotile dedicated to slats corresponding to the respective directions, along with the fact that all undirectedness in $\calT$ is handled solely by which horizontal decision slats attach, and that all other slats within a macrotile attach with dedicated glues specific to each location ensures that the system $\calS$ must model the system $\calT$ (here the set $\Pi_\alpha$ in the definition of ``models'' can easily be taken to be the full set of $\calS$-assemblies mapping to each $\calT$-assembly $\alpha$). Equivalent productions follow from the simplicity of the macrotile representation function (which depends only on the decision slats in the center of the macrotile); and because the decision tiles for tile type $t\in T$ may attach only when the respective input glues are present as encoded by the corresponding input slats, it is not difficult to see that $\calT$ must follow $\calS$. Thus, $\mathcal{S}$ correctly simulates $\calT$ under $R$. Furthermore, it does so using macrotiles of size $5c \times 5c$ and with the longest slats being of length $5c$. Therefore, Theorem \ref{thm:full-aTAM} is proven.

\end{proof}

\subsubsection[Handling arbitrary values of tau in T]{Handling arbitrary values of $\tau$ in $\calT$}

It should also be noted that a similar construction works to simulate an aTAM system using any temperature value, not just $\tau=2$. To do this, the cooperativity $c$ of $\calS$ needs to be at least as large as the temperature being simulated. In this context, input and output slats would work identically, however the decision slats need to be adjusted. In the above construction, decision slats are introduced for every way in which a tile $t$ might attach to a combination of input glues. To extend this to handle arbitrary temperatures, the decision slats corresponding to some tile type $t$ to be simulated, must only be able to attach if input slats are present corresponding to all of the necessary input glues of $t$. If for instance, we were simulating a temperature-3 system and we had a tile type $t$ that may attach so long as the north, east, and west each present a strength-1 glue, we could simulate this using a decision slat corresponding to $t$ with at least 1 glue in each of the $N$-, $E$-, and $W$-sections. Specifically, the total number of glues in these sections would have to add to the cooperativity $c$ to allow the tile to attach, so if $c=4$, then it could be the case that the $N$-section of the decision slat corresponding to $t$ had $2$ glues while the $E$- and $W$-sections each had 1 glue. In this way the decision slat would only be able to attach when all input glues are present. We also note that the choice of the $N$-section having 2 glues is arbitrary; it could have been the $W$-section with 2 glues for instance, so long as slats corresponding to each input glue of the simulated tile are necessary to be present. It will always be the case that the input slats present all $c$ glues to the decision row and it will always be the case that decision slats bind to $c$ glues. It is the distribution of glues in the various directional sections of the decision slats that enforces that they correctly represent their respective tiles in $\calT$. Whichever decision slat attaches in the north most decision row will then encode the simulated tile and output slats will attach accordingly to encode the output glues. Consequently accommodating arbitrary temperatures in this way requires that the cooperativity value $c$ is at least $\tau$. With this small modification our construction can simulate arbitrary aTAM systems of any temperature.

\end{document}